%% file: socg-main.tex
\documentclass[a4paper,USenglish,cleveref,autoref,thm-restate]
{lipics-v2021}

\usepackage{mathtools, bbm, amsthm, amsmath}
\numberwithin{theorem}{section}
\numberwithin{claim}{section}
\numberwithin{lemma}{section}
\numberwithin{definition}{section}
\numberwithin{corollary}{section}

\def\final{1}  % set this to 1 to get a comment-free version
\def\iflong{\iffalse}
\ifnum\final=0  %namely if we allow comments in the output
\newcommand{\yonggang}[1]{{\color{blue}[{\tiny Yonggang: \bf #1}]\marginpar{*}}}
\newcommand{\danupon}[1]{{\color{red}[{\tiny Danupon: \bf #1}]\marginpar{\color{red}*}}}
\newcommand{\sagnik}[1]{{\color{green!50!black}[{\tiny Sagnik: \bf #1}]\marginpar{\color{green!50!black}*}}}
\newcommand{\todo}[1]{{\color{red}[{\tiny TODO: \bf #1}]\marginpar{\color{red}*}}}
\newcommand{\yuval}[1]{{\bf \color{red!50!black} YUVAL: #1}}
\newcommand{\jan}[1]{{\bf \color{green!50!black} Jan: #1}}
\newcommand{\blikstad}[1]{\textup{\color{magenta} [\textbf{Joakim}: #1]}}
\newcommand{\joachim}[1]{\textup{\color{blue} [\textbf{Joachim}: #1]}}
\newcommand{\zihang}[1]{\textup{\color{magenta} [\textbf{Zihang}: #1]}}
\newcommand{\mar}[1]{\textup{\color{red} [\textbf{Martin}: #1]}}
\newcommand{\nithin}[1]{\textup{\color{magenta} [\textbf{Nithin}: #1]}}
\newcommand{\TODO}[1]{{\color{blue!50!black} [{\bf Todo:} #1]}}

\else % in this case [final=1] we don't want any comments to show

\newcommand{\yonggang}[1]{}
\newcommand{\danupon}[1]{}
\newcommand{\sagnik}[1]{}
\newcommand{\todo}[1]{}
\newcommand{\yuval}[1]{}
\newcommand{\jan}[1]{}
\newcommand{\blikstad}[1]{}
\newcommand{\TODO}[1]{}
\newcommand{\nithin}[1]{}
\newcommand{\mar}[1]{}
\fi  %ok, we're done with defining comment macros
\usepackage[ruled,vlined,linesnumbered]{algorithm2e} % params: noresetcount to have global counter

\newcommand{\eps}{\varepsilon}

\newcommand{\poly}{\mathrm{poly}}
\newcommand{\polylog}{\mathrm{polylog}}

\newcommand{\set}[2][ ]{\{#2 \ifthenelse{\equal{#1}{ }}{ }{~|~#1}\}}

\newcommand{\cost}{\textsf{cost}}

\newcommand{\R}{\mathbb{R}}

\newcommand{\dist}{\mathsf{dist}}
\newcommand{\transpose}{\intercal}
\newcommand{\E}{\mathbb{E}}
\newcommand{\trunc}{\mathsf{truncate}}
\newcommand{\collapse}{\mathsf{collapse}}

\newcommand{\prob}[2]{\underset{#1}{\mathbb{P}}\left[#2\right]}

\newcommand{\Good}{\mathsf{Good}}
\newcommand{\Bad}{\mathsf{Bad}}

\Crefname{algocf}{Algorithm}{Algorithms}

\DeclareMathOperator*{\argmax}{arg\,max}
\DeclareMathOperator*{\argmin}{arg\,min}
\renewcommand{\paragraph}[1]{\medskip\noindent{\bf #1}\xspace}

\Crefname{algocfline}{Algorithm}{Algorithms}

\EventEditors{Oswin Aichholzer and Haitao Wang}
\EventNoEds{2}
\EventLongTitle{41st International Symposium on Computational Geometry
(SoCG 2025)}
\EventShortTitle{SoCG 2025}
\EventAcronym{SoCG}
\EventYear{2025}
\EventDate{June 23--27, 2025}
\EventLocation{Kanazawa, Japan}
\EventLogo{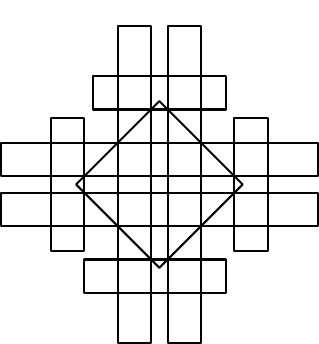}
\SeriesVolume{332}
\ArticleNo{XX}     % <-- This will be filled in by the typesetters
\title{Sublinear Data Structures for Nearest Neighbor in Ultra High Dimensions}
\author{Martin G. Herold}{Max Planck Institute for Informatics \& Saarland Informatics Campus, Saarbr\"ucken, Germany}{mherold@mpi-inf.mpg.de}{https://orcid.org/0009-0002-1804-2842}{Funded by the Deutsche Forschungsgemeinschaft (DFG, German Research Foundation) – Project number 399223600.}
\author{Danupon Nanongkai}{Max Planck Institute for Informatics \& Saarland Informatics Campus, Saarbr\"ucken, Germany}{danupon@gmail.com}{https://orcid.org/0000-0003-4468-2675}{}
\author{Joachim Spoerhase}{University of Liverpool, Liverpool, United Kingdom}{Joachim.Spoerhase@liverpool.ac.uk}{https://orcid.org/0000-0002-2601-6452}{Part of this work was done when the author was a researcher at Max Planck Institute for Informatics \& Saarland Informatics Campus, Saarbr\"ucken, Germany, and at Aalto University, Finland. Partially supported by European Research Council (ERC) under the European Union’s Horizon 2020 research and innovation programme (grant agreement No. 759557).}
\author{Nithin Varma}{University of Cologne, Cologne, Germany}{nithvarma@gmail.com}{https://orcid.org/0000-0002-1211-2566}{Part of this work was done when the author was a researcher at Max Planck Institute for Informatics \& Saarland Informatics Campus, Saarbr\"ucken, Germany.}
\author{Zihang Wu}{Max Planck Institute for Informatics \& Saarland Informatics Campus, Saarbr\"ucken, Germany}{wuzihang98@gmail.com}{https://orcid.org/0009-0002-5396-5036}{}

\authorrunning{M.\,G.\,Herold, D. Nanongkai, J. Spoerhase, N. Varma, Z. Wu}
\Copyright{Martin G. Herold, Danupon Nanongkai, Joachim Spoerhase, Nithin Varma, and Zihang Wu}



\hideLIPIcs

\begin{document}

	%\begin{titlepage}
		\maketitle \pagenumbering{roman}
		
\begin{abstract}

Geometric data structures have been extensively studied in the regime where the dimension is much smaller than the number of input points. But in many scenarios in Machine Learning, the dimension can be much higher than the number of points and can be so high that the data structure might be unable to read and store all coordinates of the input and query points. 

Inspired by these scenarios and related studies in feature selection and explainable clustering, we initiate the study of geometric data structures in this ultra-high dimensional regime. Our focus is the {\em approximate nearest neighbor}  problem. 

In this problem, we are given a set of $n$ points $C\subseteq \mathbb{R}^d$  and have to produce a {\em small} data structure that can {\em quickly} answer the following query: given $q\in \mathbb{R}^d$, return a point $c\in C$ that is approximately nearest to $q$, where the distance is under $\ell_1$, $\ell_2$, or other norms. 
Many groundbreaking $(1+\epsilon)$-approximation algorithms have recently been discovered for $\ell_1$- and $\ell_2$-norm distances in the regime where $d\ll n.$

The main question in this paper is: {\em Is there a data structure with sublinear ($o(nd)$) space and sublinear ($o(d)$) query time when $d\gg n$?} This question can be partially answered from the machine-learning literature:
\footnote{All data structures discussed here are randomized and answer each query correctly with $\Omega(1)$ probability. The space complexity refers to the number of words and the time is in the word-RAM model. $\tilde O(f(n))$ hides $\mathrm{polylog}(f(n))$ factors}

\begin{itemize}
    \item   For $\ell_1$-norm distances, an $\tilde O(\log(n))$-approximation data structure with $\tilde O(n \log d)$  space and $O(n)$ query time can be obtained from explainable clustering techniques
    [Dasgupta et al. ICML'20; Makarychev and Shan ICML'21; Esfandiari, Mirrokni, and Narayanan SODA'22; Gamlath et al. NeurIPS'21; Charikar and Hu SODA'22]. 
      
    \item For $\ell_2$-norm distances, a $(\sqrt{3}+\epsilon)$-approximation data structure with $\tilde O(n \log(d)/\mathrm{poly}(\epsilon))$ space and $\tilde O(n/\mathrm{poly}(\epsilon))$ query time can be obtained from feature selection techniques [Boutsidis, Drineas, and Mahoney NeurIPS'09; Boutsidis~et~al. IEEE Trans. Inf. Theory'15; Cohen~et~al. STOC'15]. 

    \item For $\ell_p$-norm distances, a $O(n^{p-1}\log^2(n))$-approximation data structure with $O(n\log(n) + n\log(d))$ space and $O(n)$ query time can be obtained from the explainable clustering algorithms of [Gamlath et al. NeurIPS'21]. 

\end{itemize}
An important open problem is whether a $(1+\epsilon)$-approximation data structure exists. This is not known for any norm, even with higher  (e.g. $\mathrm{poly}(n)\cdot o(d)$) space and query time.

In this paper, we answer this question affirmatively. We present   $(1+\epsilon)$-approximation data structures with the following guarantees.
\begin{itemize}
    \item For $\ell_1$- and $\ell_2$-norm distances:
    $\tilde O(n \log(d)/\mathrm{poly}(\epsilon))$  space and $\tilde O(n/\mathrm{poly}(\epsilon))$ query time. We show that these space and time bounds are tight up to $\mathrm{poly}{(\log n/\epsilon)}$ factors.
    \item For $\ell_p$-norm distances: 
    $\tilde O(n^2 \log(d) (\log\log (n)/\epsilon)^p)$  space and $\tilde O\left(n(\log\log (n)/\epsilon)^p\right)$ query time.
\end{itemize}

Via simple reductions, our data structures imply sublinear-in-$d$ data structures for some other geometric problems; e.g. approximate orthogonal range search (in the style of [Arya and Mount SoCG'95]), furthest neighbor, and give rise to a sublinear $O(1)$-approximate representation of  $k$-median and $k$-means clustering. We hope that this paper inspires future work on sublinear geometric data structures. 
\end{abstract}
		
		%\setcounter{tocdepth}{3}
		% \newpage
		% \tableofcontents
		% \newpage
	%\end{titlepage}
	
	%\newpage
	\pagenumbering{arabic}

\section{Introduction}

\paragraph{Nearest Neighbor.}
The {\em nearest neighbor} problem in geometric spaces is defined as follows: We are given a set of $n$ points $C\subseteq \R^d$ to preprocess and then have to produce a data structure that can answer the following queries: given $q\in \R^d$, return the point $c^*$ in $C$ whose distance to $q$ is smallest, where the distance can be induced by $\ell_1$-, $\ell_2$-, and other norms. 

%whose distance to $q$ is smallest where the distance can be induced by $\ell_1$-, $\ell_2$-, and other norms. 
%The focus is typically on data structures with low space and query time complexity \cite{andoni-etal18:approximate-nearest-neighbor}. [CITE MORE]
%; here, the distances may be induced by any norm. 
%$\ell_1$-, $\ell_2$-, and other norms. 

%\footnote{There are other important parameters, such as the preprocessing time. For simplicity, we will focus on the query time and space.}

In the so-called {\em low-dimensional regime} \cite{AndoniNNRW17}---when $d=o(\log n)$ and thus the data structure's space and time complexities can be {\em exponential} in $d$---classical data structures can solve this problem efficiently (e.g. \cite{DBLP:journals/dcg/Clarkson99, DBLP:conf/stoc/KargerR02, DBLP:conf/soda/KrauthgamerL04, DBLP:conf/icml/BeygelzimerKL06}). Recently, novel techniques such as locality-sensitive hashing and sketching (e.g. \cite{DBLP:conf/stoc/IndykM98, DBLP:journals/cacm/AndoniI08, DBLP:conf/soda/AndoniINR14, DBLP:conf/stoc/AndoniR15, indyk-wagner18:apx-nn-limited-space}) have led to a new wave of data structures in the {\em high-dimensional regime} 
%hat solves the problem approximately: 
%that solves approximate variants of this problem: In the {\em $\alpha$-approximate nearest neighbor problem} ($\alpha$-ANN), we are allowed to return a point $c\in C$ whose distance to $q$ is at most $\alpha$ times the distance between $c^*$ and $q$.
%
%ANN is a well-studied problem with a huge impact on real-world applications.  **Talk about current best space and time complexity**
% 
that solve the problem approximately: 
%we say that an algorithm solves the {\em $\alpha$-approximate nearest neighbor problem} ($\alpha$-ANN) 
%if it returns a point $c\in C$ whose distance to $q$ is at most $\alpha$ times the distance between $c^*$ and $q$.
for any $\alpha\geq 1$, a $\alpha$-approximation data structure is the one that returns a point $c\in C$ whose distance to $q$ is at most $\alpha$ times the distance between $c^*$ and $q$.
These recent data structures typically require space and time complexities that are linear or super linear in $d$, which is acceptable in most applications. In fact, many papers start with an assumption that $d=O(\log n)$ which can be assumed due to Johnson-Lindenstrauss lemma \cite{johnson-lindenstrauss84:jl-lemma} at the cost of additional $\tilde O(d)$ space and query time. 

%\paragraph{Ultra-High Dimensions.} But what about the case where the dimension is so large ($d\gg n$) that we cannot read all coordinates of the query $q$ and cannot store all coordinates of all $n$ points in $C$? In other words, is there a data structure with {\em sublinear} ($o(nd)$) space and {\em sublinear} ($o(d)$) query time when $d\gg n$? For example, each point in $C$ may be a center of a cluster among some small number of clusters obtained from prior training. We can classify a new data point $q\in \R^d$ by assigning it to the nearest representative. 
%
%We may be unable to know the whole $q$ due to privacy (e.g. for family and medical history) and the cost of obtaining the data (e.g. via medical tests or questionnaires). In this case, to accurately classify $q$ without accessing all attributes of $q$, we need to exploit the fact that there are typically a small number of clusters. Additionally, classifying $q$ based on a few attributes is crucial for {\em explainability}---when it is important for human users to know which factors play a role in the classifying decision. 
%\danupon{Sagnik/Nithin: Please help me improve this example.} 

\paragraph{Ultra-High Dimensions.} But what if we cannot read all coordinates of the query $q$ and cannot store all coordinates of all $n$ points in $C$? There can be various reasons for such a situation to arise. For example, we may be unable to access the whole $q$ due to the high-cost concern of obtaining the data (e.g. high-cost medical tests or genome sequencing) and privacy concerns (e.g. providing only the necessary medical or genetic information while protecting unrelated sensitive data). More crucially, the size of the data, i.e., $d$ might be prohibitively large to query all of its attributes (as is usually the case for genetic data) and, yet, $n$ can be much smaller than $d$. For example, each point in $C$ may be a center of a cluster among some small number of clusters obtained from prior training and we aim to classify a new data point $q\in \R^d$ by assigning it to the nearest representative. In this case, to accurately classify $q$ without accessing all attributes of $q$, we need to exploit the fact that there are typically a small number of clusters. Additionally, classifying $q$ based on a few attributes is crucial for {\em explainability}---when it is important for human users to know which factors play a role in the classifying decision.

\paragraph{Related Work and Existing Results.}
Scenarios like the above are a motivation behind at least two lines of research in machine learning. 
{\bf (i)} In {\em Feature Selection}, the goal is to select a subset of relevant features of the input dataset. For example, Boutsidis, Drineas, and Mahoney in NeurIPS'09 \cite{BoutsidisDM09} (also see \cite{DBLP:journals/tit/BoutsidisZMD15}) show that given a set $S \subseteq \R^d$ and a parameter $k\ll d$, we can select roughly $k$ coordinates such that we are required to only access these coordinates when computing $k$-means clustering, at the expense of a $(3+\epsilon)$ multiplicative approximation factor.
The approximation factor was subsequently improved to $(1+\epsilon)$ in the STOC'15 paper by Cohen~et~al. \cite{DBLP:conf/stoc/CohenEMMP15}. 
%
%\joachim{Their alg reads all coordinates in preprocessing and gives (2+eps)-apx but with Zihang's argument it can be used to give 3-apx for our setting also. Maybe we write ``we argue that as a consequence of their result, we can get a 3-apx?''}\danupon{I don't understand Joachim's comments. Let's ignore it for now. Zihang, please see my comment.}
%
{\bf (ii)} In \emph{Explainable Clustering}, given data that are already classified into $k$ clusters via cluster centers,
% \joachim{I believe the original intent of explainable clustering is just to compute a low cost explainable clustering. But as a byproduct their method can be used to classify also previously unclassified points. There may be other methods that also require knowledge of data points. Should we clarify that we are repurposing their result?}\danupon{
% There seem to be two definitions. In Makarychev and Shan and Svensson, it's like this.
% }, 
the goal is to construct a decision tree with some desired properties that can classify the input data points into clusters without too much loss in the clustering cost compared to the original clustering. For example, Makarychev and Shan \cite{MakarychevS21} show that given a $k$-median clustering on $\ell_1$ distance, one can construct such a decision tree that returns a clustering whose cost is $\tilde O(\log k)$ times the original cost. 
%
%Since one desired property is for the decision tree to have $O(k)$ nodes and each node should only compare between two coordinates of the data point, this alg
%
Since the number of accessed coordinates intuitively plays an important role in explainability, it is not surprising that one of the desired properties of the decision tree is that it never accesses more than $O(k)$ coordinates. 

% %The other natural setting is when the dimension $d$ is much larger than $n$
% %Despite these successes, there is currently little understanding of the nearest neighbor problem in the regime where the dimension is much larger than the number of data points.
% This \emph{ultra-high dimensional} regime is natural, such as in the case of genetic data and several scenarios in Machine Learning.
% \emph{Feature selection} [CITE], where the goal is to select a subset of \emph{relevant} features of the input dataset and \emph{Explainable Clustering} [CITE] where the clustering is performed with respect to a selected set of features are examples of other problems to be considered in this regime. 
% The aforementioned dependence of the space or query time on the dimension $d$ is clearly undesirable in this regime. 
% %and a reasonable goal is to only have these quantities depend sublinearly on $d$.  

%\paragraph{Known Results.} [TO DO - SIMILAR TO ABSTRACT]

%\paragraph{Existing Results.} 
%Observe that the above Feature Selection and Explainable Clustering are {\em static problems}, i.e. ... 
%While the above Feature Selection and Explainable Clustering problems are closely related to nearest neighbor, they 
%Clustering problems are known to have connections with

The feature selection and explainable clustering problems are {\em static problems}; i.e., the algorithms for these problems receive all of their inputs before returning the output. This is in contrast to the nearest neighbor problem which is a {\em data structure problem}. 
%We believe that we are the first to study geometric data structures in this ultra-dimensional regime.) 
Nevertheless, given the connections between clustering and nearest neighbor problems, techniques from some of the previous works on these static problems can be extended to the nearest neighbor problem. (All data structures discussed here are randomized and answer each query correctly with $\Omega(1)$ probability. The space complexity refers to the number of words and the time is in the word-RAM model. $\tilde O(f(n))$ hides $\polylog(f(n))$ factors.)

\begin{itemize}
    \item For $\ell_1$-norm distances, an $O(\log(n))$-approximation data structure with  $O(n \log (n) + n \log d)$  space and $O(n)$ query time can be obtained by adapting the explainable clustering algorithms of \cite{GPSY23, MakarychevS23}.
    %from ICML'21 and SODA'22. 
    This is because these algorithms only require the center of each cluster, and not the remaining input data points. 
    Weaker results can also be obtained from \cite{DBLP:conf/soda/EsfandiariMN22, MakarychevS21, MoshkovitzDRF20, DBLP:conf/nips/GamlathJPS21}.
    % \item   For $\ell_1$-norm distances, an $O(\log(n)\log\log(n))$-approximation data structure with  $O(n \log (n) + n \log d)$  space and $O(n)$ query time can be obtained by adapting the explainable clustering algorithms of \cite{DBLP:conf/soda/EsfandiariMN22, MakarychevS21}.
    % %from ICML'21 and SODA'22. 
    % This is because these algorithms only require the center of each cluster, and not the remaining input data points. 
    % %
    % Weaker results can also be obtained from \cite{MoshkovitzDRF20, DBLP:conf/nips/GamlathJPS21}.
    %(Algorithms of  [Gamlath et al. NeurIPS'21] and  [Dasgupta et al. ICML'20] can be adapted to our problem too, but the results will be weaker.)

    %A $O(\log^2(n))$-approximation data structure with  $O(n \log (n) + n \log d)$  space and $O(n)$ query time. [Gamlath et al. NeurIPS'21]

    %A $O(k)$-approximation data structure with  $O(n \log (n) + n \log d)$  space and $O(n)$ query time. [Dasgupta et al. ICML'20]

    % [Dasgupta et al. ICML'20; Makarychev and Shan ICML'21; Esfandiari, Mirrokni, and Narayanan SODA'22; Gamlath et al. NeurIPS'21; Charikar and Hu SODA'22]. TO DO: CITATION. MORE DETAIL OF THE PARAMETERS. SAY A FEW WORDS WHY IT CAN EXTEND; I.E. IT'S BECAUSE SOME OF THESE PAPERS ACTUALLY DON'T NEED TO KNOW INPUT, BUT JUST THE CENTERS.
      
    \item For $\ell_2$-norm distances, a $(\sqrt{3}+\epsilon)$-approximation data structure with $\tilde O(n \log(d)/\epsilon^2)$ space and $O(n\log(n)/\epsilon^2)$ query time can be obtained by reducing to feature selection algorithms \cite{BoutsidisDM09, DBLP:journals/tit/BoutsidisZMD15, DBLP:conf/stoc/CohenEMMP15} (see~\cref{sec:reduction-clustering}).\footnote{In \cref{sec:reduction-clustering} we also argue why a natural reduction to feature selection~\cite{DBLP:conf/stoc/CohenEMMP15} fails to give a $(1+\epsilon)$-approximate nearest neighbor data structure.} Weaker results for the case of $\ell_2$-norm distances can also be obtained from explainable clustering algorithms of \cite{makarychev-shan22:explainable-k-means, MakarychevS21, DBLP:conf/soda/EsfandiariMN22,DBLP:conf/nips/GamlathJPS21, DBLP:conf/soda/CharikarH22, MoshkovitzDRF20}.
    
    % A $\tilde O(\log^2 (n))$-approximate data structure with $O(n\log(n) + n\log(d))$ space and $O(n)$ query time can be obtained from [Makarychev and Shan STOC'22]
    
    % A $\tilde O(k)$-approximation data structure with $O(n \log (n) + n \log d)$  space and $O(n)$ query time. [Makarychev and Shan ICML’21; Esfandiari, Mirrokni, and Narayanan SODA’22; Gamlath et al. NeurIPS’21; Charikar and Hu SODA’22].

    % A $\tilde O(k^2)$-approximation data structure with $O(n \log (n) + n \log d)$  space and $O(n)$ query time. [Dasgupta et al. ICML’20;]
    
    \item For $\ell_p$-norm distances, an $O(n^{p-1}\log^2(n))$-approximation data structure with $O(n\log(n) + n\log(d))$ space and $O(n)$ query time can be obtained from the explainable clustering algorithms of \cite{DBLP:conf/nips/GamlathJPS21}.  
\end{itemize}

In contrast to other regimes, where many efficient $(1+\epsilon)$-approximation data structures are known, no  $(1+\epsilon)$-approximation data structure is known for the ultra-high dimensional regime for any norm, and even when we allow higher space and query time (e.g. $\poly(n)\cdot o(d)$ space and query time).

% An important open problem is whether a $(1+\epsilon)$-approximation data structure exists. This is not known for any norm, even with higher  (e.g. $\poly(n)\cdot o(d)$) space and query time.
% %
% In this paper, we answer this question affirmatively. We present   $(1+\epsilon)$-approximation data structures with the following guarantees.
%\textcolor{magenta}{ The space bound follows from the one-way communication complexity of approximate nearest neighbor. To see the lower bound on the number of queries, consider the case that $C \subseteq \{0,1\}^n$ consists of the $n$ unit vectors and the query point $q$ is co-located with one of the points in $C$. Every deterministic approximate nearest neighbor algorithm must query all of the $n$ coordinates in $q$ and the same argument extends to randomized algorithms.}

\paragraph{Our Results.} We first observe that every $(1+\epsilon)$-approximate nearest neighbor data structure requires $\Omega(n)$ space and must make $\Omega(n)$ queries (see \cref{thm:approx-NN-lb}). Our main contributions are data structures with nearly matching space and query bounds in the ultra-high dimensional regime. Our data structures require polynomial preprocessing time, answer each query correctly with $1-\delta$ probability, and guarantee the following bounds.

\begin{itemize}
    \item For $\ell_1$-norm distances: 
    $O(n \log^2(n/\delta) \log(d)/(\epsilon^5\delta^4))$  space and 
    $O(n \log (n/\delta)/(\epsilon^3\delta^2))$ query time. See \cref{theorem:sample-reuse-l1}.
    \item For $\ell_2$-norm distances: 
    $O(n \log^2(n/\delta) \log(d)/(\epsilon^6\delta^4))$  space and
    $O(n \log (n/\delta)/(\epsilon^3\delta^2))$ query time. See \cref{theorem:sample-reuse-l2}.
    \item For $\ell_p$-norm distances: 
    $O(n^2 \log(n/\delta) \log(d)(\log\log n)^p/\epsilon^{p+2})$  space
    and $O(n \log (n/\delta)\allowbreak (\log\log n)^p/\epsilon^{p+2})$ query time. See \cref{theorem:Lp-NN}.
\end{itemize}

Our data structure for $\ell_1$- and $\ell_2$-norm distances uses essentially the same space and query time as the previously best result and provides a stronger approximation guarantee.  For other norms, our data structure needs higher (but sublinear) space and gives strong approximation guarantees.

\paragraph{Connection to Explainable Clustering.} 
A recent line of research investigates explainable clustering \cite{DBLP:conf/soda/EsfandiariMN22, MakarychevS21, MoshkovitzDRF20, DBLP:conf/nips/GamlathJPS21, GPSY23, MakarychevS23}. The input is a set $C$ of $k$ cluster centers. The output is a (random) decision tree assigning any query point to a cluster represented by a label in $[k]$. The current best result \cite{GPSY23} for explainable $k$-median clustering under $\ell_1$-distances guarantees that for \emph{any} set $P$ of data points the clustering generated by the decision tree is in expectation within a factor $O(\log k)$ of the optimal assignment of $P$ to the centers in $C$, that is, of the nearest-center assignment. The decision tree has $k$ leaves. Therefore, it can be interpreted as $O(\log k)$-approximate representation of the clustering $C$ with sublinear $\tilde{O}(k)$ space and query complexity. There are two aspects why this representation is considered explainable: (i) the decision tree is easy to interpret for humans (ii) it is compact as it has only sublinear $\tilde{O}(k)$ space and query complexity. When applied to a single query point these data structures constitute a $O(\log k)$-approximate nearest neighbor data structure in expectation. 

Our techniques have implications in the converse direction. We argue (see Section~\ref{sec:expected-apx}) that our techniques give a sublinear $\tilde{O}(k)$ space and query complexity data structure for an $O(1)$-approximate nearest neighbor in \emph{expectation}. This implies a sublinear $O(1)$-approximate representation for $k$-median clustering improving upon the $O(\log k)$-approximation obtained via decision trees. Notice that while our data structure does not produce decision trees it does constitute a compact representation of the clustering. Our techniques also imply a sublinear $O(1)$-approximate representation of $k$-means clustering in Euclidean space improving upon the $O(\log^2 k)$-bound implied by explainable $k$-means clustering \cite{makarychev-shan22:explainable-k-means}.

% [TO DO - SIMILAR TO ABSTRACT BUT WITH MORE PARAMETER DETAILS AS THEOREMS]\danupon{Zihang}
% \begin{restatable}[Main Theorem 1]{theorem}{maintheorem}\label{theorem:Lp-NN}
%     Consider a set $P$ of $n$ points $c_1, \dots, c_n \in \mathbb{R}^d$ with $\ell_p$ metric, $1\leq p < \infty$. Given $0 < \epsilon < 1/(4p)$, we can construct a data structure $B$ of word size $O(n^2 \log(n/\delta) \log(d)2^p/\epsilon^{p+2})$ that, given any query $q\in R^d$, reads $O(n \log (n/\delta)2^p/\epsilon^{p+2})$ coordinates of $q$, returns $i$ where $c_i$ is $(1+\epsilon)$-approximate nearest neighbor of $q$ from $P$ with probability $1-\delta$:
%     \[
%     \mathsf{dist}(q, c_i) \leq (1+\epsilon) \cdot \min_{j\in [n]}\mathsf{dist}(q, c_j). 
%     \]
% \end{restatable}

% \begin{restatable}[Main Theorem 2]{theorem}{secondtheorem}\label{theorem:sample-reuse-l1}
%     Consider a set $P$ of $n$ points $c_1, \dots, c_n \in \mathbb{R}^d$ with $\ell_1$ metric. Given $0 < \epsilon < 1/(4p)$, we can construct a data structure $B$ of word size $O(n \log(n/\delta) \log(d)/\epsilon^{2})$ that, given any query $q\in \mathbb{R}^d$, reads $O(n \log (n/\delta)/\epsilon^{2})$ coordinates of $q$, returns $i$ where $c_i$ is $(1+\epsilon)$-approximate nearest neighbor of $q$ from $P$ with probability $1-\delta$:
%     \[
%     \mathsf{dist}(q, c_i) \leq (1+\epsilon) \cdot \min_{j\in [n]}\mathsf{dist}(q, c_j). 
%     \]
% \end{restatable}

\paragraph{Perspective: Sublinear Data Structures.} Geometric data structures are fundamental computational objects that have been extensively studied and used for decades. To the best of our knowledge, there was no prior result on geometric data structure problems in the ultra-high dimensional regime. Showing that it is possible to design sublinear geometric data structures is the main conceptual contribution of our paper.

% TRY TO SAY THESE NICELY + TALK ABOUT OTHER RESULTS (STATE AS THEOREMS OR COROLLARIES). While there have been a number of works on ultra-high dimensional {\em static} problems, especially in the directions of feature selection and explainable clustering, we are not aware of any previous works on {\em data structures} in this regime. Showing that sublinear data structures are possible is the main conceptual contribution of this paper. While our focus is on **ANN**, we show below that our results and techniques can be extended to some other problems.

As another example of where our techniques are applicable, consider the orthogonal range search problem. Here, we are given a set $C\subseteq \R^d$ to preprocess and need to answer the following query: Given $q_1$ and $q_2$ in $\R^d$, return all points $p \in C$ such that the $i^{th}$ attribute of $p$ has value between the $i^{th}$ attributes of~$q_1$ and~$q_2$, for every $i$. This problem can be solved using range trees and $k$-$d$ trees \cite{Bentley80, McCreight85, Chazelle86, Chazelle88} in the low-dimensional regime. For higher dimensions, some approximation algorithms have been proposed where some other points can be returned as long as they are not ``far'' from the query \cite{AryaM00,ChazelleLM08,FonsecaM10}, i.e., the distance between the point returned and the range is within an $\epsilon$ fraction of the distance between $q_1$ and $q_2$---we call such a range search as $(1+\epsilon)$-approximate range search. We approach this problem, as before, from a data-structure point of view. Using similar techniques, we can design a data structure of size $\tilde O(n)$ that processes $n$ centers such that given any `valid' orthogonal range $(q_1, q_2)$ (i.e., a range which contains at least one center point), queries only $\tilde O(n/\epsilon^2)$ coordinates of $q_1$ and $q_2$ to report all centers that are within the $(1+\epsilon)$-approximate range. 
The details of this can be found in \Cref{sec:range-search}. 

Besides orthogonal range search, simple reductions also lead to data structures for other proximity problems such as the furthest neighbor problem, where we can guarantee the same bounds as the nearest neighbor problem.
% Our techniques can be extended to other problems:
% \begin{itemize}
%     \item Furthest neighbor. Our techniques directly work for this.
%     \item Orthogonal range search. Given $c_1, c_2, \dots, c_n$, we can construct a data structure $B$ of word size $O(n \log(n) \log(n/\delta) \log(d)/\epsilon^{2})$ that, given any query range $R$ by the corner points $q_1,q_2 \in R^d$, reads $O(n \log(n) \log (n/\delta)/\epsilon^{2})$ coordinates of $q_1,q_2$. 
% \end{itemize}

We hope that this paper inspires future work on sublinear geometric data structures. 

\paragraph{Organization.}
Section~\ref{sec:tech-overview} gives an overview of the main technical contributions of the paper.
Section~\ref{sec:sample-reuse} describes our $(1+\epsilon)$-approximate nearest neighbor data structures with $\tilde O(n)$ space and query complexity for the $\ell_1$- and $\ell_2$-norm. 

\cref{sec:warm-up} contains the outline of a simple data structure with space $\tilde O(n^2)$ and query time $\tilde O(n)$ for $\ell_1$-norm, with some simplifying assumptions as warmup. \cref{sec:missing-proof} contains missing proofs of \cref{sec:sample-reuse} and extension to $\ell_2$ case. \cref{sec:expected-apx} describes our expected constant-approximate nearest neighbor data structure, and shows the connection to explainable clustering. \cref{sec:point-reduction} describes our $(1+\epsilon)$-approximate nearest neighbor data structure with space $\tilde O(n^2)$ and query time $\tilde O(n)$ that works for general $\ell_p$-norms. \cref{sec:range-search} shows the implications of our results for Range Search problem. \cref{sec:reduction-clustering} shows the connection between approximate nearest neighbor and feature selection for clustering. \cref{sec:lower-bound} shows $\Omega(n)$ space and query tight lower bound.

\section{Technical Overview}\label{sec:tech-overview}

\paragraph{Warmup: Quadratic Space via Pairwise Sampling.}Let $c$ be an input point, let $q$ be an unknown query point, and let $\|c-q\|_1=\sum_b|c^{(b)}-q^{(b)}|$ denote their $\ell_1$ distance. First, observe that approximating this distance requires memorizing {\em all} coordinates of $c$ because $c$ and $q$ may differ in only a single coordinate. Our strategy is to \emph{compare} distances to multiple input points without knowing the (approximate) distances.
    For intuition, consider when we have two input points $c_1$ and $c_2$. Consider an extreme case when $c_1$ and $c_2$ differ only in one coordinate~$b$. In this case, keeping only this coordinate suffices to tell for any query $q$ if it is closer to $c_1$ or $c_2$. To extend this intuition to when $c_1$ and~$c_2$ differ in many coordinates, a natural idea is to sample each coordinate~$b$ independently with probability proportional to the difference $|c_1^{(b)}-c_2^{(b)}|$; more precisely, the probability that we sample coordinate~$b$ is\footnote{This can be viewed as an application of so-called importance sampling---a technique for variance reduction. The independent sampling process described here differs slightly from the one in Section~\ref{sec:warm-up}.}
\begin{displaymath}
p^{(b)}=\frac{|c_1^{(b)}-c_2^{(b)}|}{\|c_1-c_2\|_1}\,.
\end{displaymath} 
We further scale the sample entry by a factor of $1/p^{(b)}$. The expected number of sampled coordinates is then $\sum_b p^{(b)}=\sum_b |c_1^{(b)}-c_2^{(b)}|/\|c_1-c_2\|_1=1$. Let $\epsilon\in (0,1)$ be the approximation parameter, and let $\delta\in (0,1)$ be a bound on the error probability. We repeat this process to store a set $I$ of $O(\log (1/\delta)/\epsilon^2)$ coordinates in total. For $x\in\R^d$, we denote by $x^{(I)}$ the restriction of $x$ to coordinates in $I$. We denote by $S=\textsf{diag}((1/p^{(b)})_{b\in I})$ denote the diagonal matrix representing the scaling factors. Then $Sx^{(I)}$ is the vector obtained from applying the above sampling process to~$x$.

Under what we call {\em bounding box assumption} (which we will remove later), where $\min(c_1^{(b)},c_2^{(b)})\leq q^{(b)}\leq \max(c_1^{(b)},c_2^{(b)})$ for each coordinate $b$, we can find a $(1+\epsilon)$-approximate nearest neighbor between $c_1$ and $c_2$. This is due to the following property. (Details in Section~\ref{sec:warm-up}.) 
\begin{quote}
     {\em Comparator Property:} Let $\epsilon, \delta\in (0,1)$ and assume we sample a set $I$ of $O(\log (1/\delta)/\epsilon^2)$ coordinates as described above. If $\|c_2-q\|_1>(1+\epsilon)\|c_1-q\|_1$ then $\|Sc_2^{(I)}-Sq^{(I)}\|_1>(1+\epsilon/2)\|Sc_2^{(I)}-Sq^{(I)}\|_1$ with probability $1-\delta$.
\end{quote}
With this idea, we can construct an $\tilde O(n^2)$-space data structure by repeating the above for every pair $(c_i, c_j)$ of input points. It is not hard to find the approximately nearest neighbor to $q$ with $n-1$ comparisons: this is very simple if the comparator is exact, and the same idea can be extended for our approximate comparator as well; see Section~\ref{sec:warm-up} for details (also see \cite{AjtaiFHN16} for previous work on imprecise comparisons).

\paragraph{Towards Linear Space via Global Sampling.} The drawback of the previous algorithm is its  $\tilde O(n^2)$ space complexity. To avoid this, a natural idea is to reduce the number of pairs for which we perform the above sampling process. However, it is unclear to us if this is possible at all. Another idea is to do the above process {\em more globally} instead of pairwise. Here is a simple way to do this: Let $C=\{c_1,\dots,c_n\}$ be the set of input points. We (independently) sample each coordinate $b$ with probability $p^{(b)}=\max_{ij}p^{(b)}_{ij}$ where $p_{ij}^{(b)}=|c_i^{(b)}-c_j^{(b)}|/\|c_i-c_j\|_1$ denotes the probability of sampling $b$ under the pairwise sampling process for $c_i,c_j\in C$.

At least two questions arise from the idea above: 
(i) For any pair $c_i$ and $c_j$ and query $q$ in the bounding box can we still tell (approximately) if $c_i$ or $c_j$ is nearer to $q$ as in the pairwise process (in particular, are we creating noise by sampling some coordinates with higher probility than $p_{ij}^{(b)}$)? %
This is not so hard to address: Despite sampling each coordinate with probability higher than $p_{ij}^{(b)}$, the expectation does not change due to rescaling by $1/p^{(b)}$, and the concentration behavior is only better. More precisely, the comparator property can still be proven. 

A more unclear point is: (ii) How many coordinates do we sample in total? For example, why is it not possible that in the worst case every pair $c_i$ and $c_j$ needs different coordinates and thus we are forced to select $\Omega(n^2)$ coordinates in total. We show that this bad situation will not happen. To give an idea, we consider the special case $C\subseteq\{0,1\}^d$. The overall expected number of coordinates sampled is $S(C)=\sum_b p^{(b)}=\sum_b \max_{ij} p_{ij}^{(b)}$. (We ignore sampling repetition, which increases our bound only by~$O(\log n)$.) Assume w.l.o.g.\ that $(c_1,c_2)$ is the closest pair in $C$. Notice that $p^{(b)}=p_{12}^{(b)}$ for every coordinate~$b$ in the set $D$ of coordinates where $c_1$ and $c_2$ differ; they minimize $\|c_i-c_j\|_1$ and hence maximize $p_{ij}^{(b)}$. Therefore $S(C)$ drops by at most $\sum_b p_{12}^{(b)}=\sum_b |c_1^{(b)}-c_2^{(b)}|/\|c_1-c_2\|_1=1$ if we remove the coordinates in $D$ from all input points. On the other hand, we can identify $c_1$ and $c_2$ after removing $D$ obtaining a set $C'$ with at most~$n-1$ points. Thus $S(C)\leq S(C')+1$. Using induction we can upper bound $S(C)$ by $n$. This inductive proof strategy can be extended to the general $\ell_1$ metric although the proof is notably more technical. Even though $\ell_1$ and $\ell_2$ metrics are equivalent on the Boolean cube, we do not know how to extend the above proof strategy to the general $\ell_2$ metric. To obtain a linear bound for $\ell_2$ we rely on an entirely different, linear-algebraic approach leveraging a latent connection between our sampling probabilities and the singular-value decomposition.
%\danupon{Write at most 2 sentences. IS this novel, standard, advanced, etc.? How does it compare to the literature (which papers)?}\joachim{Will add sth here} %\zihang{Singular-value decomposition is heavily used in $k$-means clustering with $\ell_2$ metric because of the linear-algebraic formulation of this problem, e.g.\cite{DBLP:journals/tit/BoutsidisZMD15}, and we have a reduction from approximate nearest neighbor to clustering \cref{sec:reduction-clustering}. Our advantages are (i) we have better approximation factor, (ii)our sampling probabilities are upper bounded by their sampling probabilities, which implies our data structure consumes potentially smaller space complexity than reducing to feature-selection for $k$-means clustering.}
Our proof is based on a connection to the feature selection approach by Boustidis et al.~\cite{DBLP:journals/tit/BoutsidisZMD15} for $k$-Means clustering. In fact, our sampling probabilities can be upper bounded by theirs, which are based on SVD. Notice, however, that we achieve a $(1+\epsilon)$-approximation for nearest neighbor as compared to their $(3+\epsilon)$-approximation for $k$-Means.\footnote{In Section~\ref{sec:reduction-clustering}, we provide a formal reduction to feature selection by Boustidis et al.~\cite{DBLP:journals/tit/BoutsidisZMD15}. There is follow-up work achieving even guarantee $1+\epsilon$ for feature selection  in clustering~\cite{DBLP:conf/stoc/CohenEMMP15}. However, we do not know how to apply their approach to our setting. In particular, the above reduction is not applicable.}

Let $I$ denote the set of coordinates sampled by our global scheme. Memorizing these coordinates for all input points would still require $n|I|=\tilde O(n^2)$ space. To obtain $\tilde O(n)$ space we apply standard dimension reduction tools such as the Johnson-Lindenstrauss transform for $\ell_2$ or the technique by Indyk~\cite{indyk06:l1-dimension-reduction} for $\ell_1$. Specifically, we sample a suitable random $m\times |I|$ matrix $M$, where $m=O(\log n)$, and store the set $C'=\{\,MSc_i^{(I)}\mid c_i\in C \,\}$ requiring $mn=\tilde O(n)$ space. Storing $M$ itself requires $m|I|=\tilde O(n)$ space, too, as does storing $I$ itself. Upon receiving a query $q$, we compute $MSq^{(I)}$ and output the nearest neighbor in $C'$, which takes $\tilde O(n)$ coordinate queries. Here we use that the comparator property guarantees $\|Sq^{(I)}-Sc_j^{(I)}\|_1>(1+\epsilon/2)\|Sq^{(I)}-Sc_i^{(I)}\|_1$ if $\|q-c_j\|_1>(1+\epsilon)\|q-c_i\|_1$. Therefore, if we bound the distance distortion of the dimension reduction by, say, $1+\epsilon/10$ then $F(MSq^{(I)},MSc_j^{(I)})>(1+\epsilon/5)F(MSq^{(I)}, MSc_i^{(I)})$ for the function $F$ from the $\ell_1$ dimension reduction by Indyk~\cite{indyk06:l1-dimension-reduction}. This is sufficient for approximately comparing the distances. 

\paragraph{Removing the Bounding Box Assumption.} If a query point $q$ does not satisfy the bounding box property for two points $c_i,c_j\in C$ the comparator property may fail with too high probability. Specifically, it may happen that $\|q-c_j\|_1>(1+\epsilon)\|q-c_i\|_1$ and yet $\|Sq^{(I)}-Sc_j^{(I)}\|_1=(1+\xi)\|Sq^{(I)}-Sc_i^{(I)}\|_1$ for some positive $\xi\ll\epsilon$. In such a scenario, applying dimension reduction via matrix $M$ may result in $F(MSq^{(I)},MSc_j^{(I)})<F(MSq^{(I)},MSc_i^{(I)})$ because it distorts distances by factor~$1+\Theta(\epsilon)$.

A natural attempt to resolve this issue is to ``truncate'' the query point so it comes sufficiently close to the bounding box. This approach is described in Section~\ref{sec:point-reduction} and can in fact be applied to general $\ell_p$ metrics. However, it requires explicitly storing the projected points and therefore $\tilde O(n^2)$ space.

To retain our near-linear space bound, we do not actually change the sampling or the query procedure. Rather, we relax in our analysis on the requirement that the comparator property needs to hold for all point pairs. Specifically, we argue (using Markov) that with sufficiently high probability our distance estimate to the (unknown) nearest neighbor $o$ is within a bounded factor (*). We argue that under this condition all comparisons with $o$ do satisfy the comparator property, which is enough to satisfy correctness. Towards this, we distinguish for each $c\in C$ between \emph{good} coordinates $b$ where $q^{(b)}$ is relatively close to the interval $[c^{(b)},o^{(b)}]$. For the set of good coordinates, strong concentration properties hold---similar to those under the bounding box assumption. For the remaining (bad) coordinates, we use two properties: First, the query point is indifferent between $c$ and $o$ on the set of bad coordinates as it is sufficiently distant and thus indifferent between $o$ and $c$ on these coordinates. This and property (*) automatically implies that the total contribution of the bad coordinates cannot be too high. Together, these properties imply that (a relaxation of) the comparator property holds for all comparisons with $o$ thereby implying the $(1+\epsilon)$-approximation guarantee also after removing the bounding box assumption. For details, we refer to the full technical sections below.

\section{Preliminaries}\label{sec:preliminaries}
In this section, we list some basic terminology, definitions and notation that we use throughout.

\begin{definition}[Approximate Notation]
    Given real numbers $a,b>0$ and $0<\epsilon<1$, we say that $a\approx_{\epsilon} b$ if $(1-\epsilon)b \leq a \leq (1+\epsilon)b$.
\end{definition}
\begin{definition}[Coordinate Index Notation]
    For a point $a\in \R^d$, we use $a^{(i)}$ to represent the $i$-th coordinate of $a$. For a set or multiset $I = \{i_1, i_2, \dots, i_m\} \subseteq [d]$, we denote $a^{(I)}$ to be $\left(a^{(i_1)}, a^{(i_2)}, \dots, a^{(i_m)}\right) \in \R^m$.
\end{definition}

In this paper, we look into the approximate nearest neighbor problem, defined as follows.
\begin{definition}[$\alpha$-Approximate Nearest Neighbor Problem]
    Given a set $C=\{c_1, c_2, \dots, c_n\}$ of $n$ points in a metric space $(X, \dist)$, build a data structure that given any query point $q \in X$, returns index $i\in [n]$, such that $
    \dist(q, c_i) \leq \alpha \cdot \min_{j\in [n]}\dist(q, c_j). $
\end{definition}
When $C$ only has 2 points, we use the shorthand \emph{2-point comparator} to denote such an approximate nearest neighbor data structure. 
\begin{definition}[2-point Comparator]
    Given a set $C=\{c_1, c_2\}$ of $2$ points in a metric space $(X, \dist)$ and parameters $0\leq \epsilon, \delta < 1$, we call a data structure as a 2-point comparator, if for any query point $q \in X$, with probabilty $1-\delta$, it returns index $i\in \{1,2\}$, such that $
    \dist(q, c_i) \leq (1+\epsilon) \cdot \min_{j\in \{1,2\}}\dist(q, c_j). $
\end{definition}

We use the following Chernoff bounds in this paper.
\begin{theorem}[Chernoff Bound]
    Let $X_1, X_2, \dots, X_n$ be independent random variables, and $X_i \in [0, c]$. Define $X = \sum_{i=1}^n X_i$ and $\mu = \E[X]$. Then
    \begin{align*}
        &\prob{}{X \geq (1+\delta)\mu} \leq \exp(-\delta^2 \mu /(3c)), \text{ for } 0 < \delta \leq 1\\
        &\prob{}{X \geq (1+\delta)\mu} \leq \exp(-\delta \mu /(3c)), \text{ for } \delta \geq 1\\
        &\prob{}{X \leq (1-\delta)\mu} \leq \exp(-\delta^2 \mu /(2c)), \text{ for } \delta \geq 0.
    \end{align*}
\end{theorem}

\section{Near-linear Space Data Structures for \texorpdfstring{$\ell_1$}{l1} and \texorpdfstring{$\ell_2$}{l2} Metrics}
\label{sec:sample-reuse}
%In this section, we describe a data structure with better space complexity for approximate nearest neighbor under $\ell_1$ and $\ell_2$ metric.
%We prove the following theorem.
In this section, we design a data structure with near-linear space and query complexity for approximate nearest neighbor under $\ell_1$ and $\ell_2$ metrics. 
We prove the following theorem.

\begin{theorem}\label{theorem:sample-reuse-l1l2}
    Consider a set $C$ of $n$ points $c_1, \dots, c_n \in \R^d$ equipped with $\ell_1$ or $\ell_2$ metric. Given $0 < \epsilon < 1/4$ and $0 < \delta < 1$, we can efficiently construct a randomized data structure such that for any query point $q\in \R^d$ the following conditions hold with probability $1-\delta$. (i) The data structure reads $O(n \log (n)/(\poly(\epsilon,\delta)))$ coordinates of $q$. (ii) It returns $i\in [n]$ where $c_i$ is a $(1+\epsilon)$-approximate nearest neighbor of $q$ in $C$. (iii) The data structure has a size of $O(n \log^2(n) \log(d)/(\poly(\epsilon,\delta)))$ words. 
\end{theorem}

The proof of \cref{theorem:sample-reuse-l1l2} for the case of $\ell_1$-metric is presented in \cref{sec:l1samplereuse} with missing proofs deferred to \cref{sec:appendixforlone}, where as the proof for the case of $\ell_2$-metric is in \cref{sec:linear-space-l2}.
Both proofs follow the same structure although they differ in details.

\subsection{Data Structure for \texorpdfstring{$\ell_1$}{l1}-metric with Near-linear Space and Query Time}\label{sec:l1samplereuse}

The preprocessing and querying procedures of our data structure are described in \cref{alg:sample-reuse-l1}.
The algorithm receives a set $C$ of $n$ points and parameters $\epsilon, \delta \in (0,1)$ as inputs for its preprocessing phase. During preprocessing, it iteratively samples coordinates from $[d]$ to form a multiset $I$.\footnote{Multiple coordinates may be added to $I$ in the same iteration, and $I$ may contain repeated coordinates.} Let $R$ denote the set $C$ of $n$ points restricted to $I$. An $\ell_1$ dimension reduction is then applied to $R$ to further reduce the space. 
In the query phase, the algorithm receives a query $q \in \R^d$. It probes $q$ on all coordinates in $I$ and applies the same $\ell_1$ dimension reduction to the coordinate-selected query. The output, identifying which point in $C$ is the nearest neighbor of $q$, is determined by which point, after coordinate selection and dimension reduction, is closest to the query.

\begin{algorithm}[htp]\label{alg:sample-reuse-l1}
    \caption{An approximate nearest neighbor data structure for $n$ points in $\R^d$ for $\ell_1$ metric.}
    \SetKwProg{preprocessing}{Preprocessing}{}{}
    \SetKwProg{query}{Query}{}{}
    \SetKwProg{comparator}{$E$}{}{}
    \SetKw{store}{store}
    \preprocessing{$(C, \epsilon, \delta)$\tcp*[f]{Inputs: $C=\{c_1, \dots, c_n\}\subseteq \R^d, \epsilon \in (0, \frac{1}{4}), \delta\in (0,1)$}}{
        % Let $\epsilon \leftarrow \epsilon'/\log \log n$ and $\delta \leftarrow \delta' / \Theta(n \log \log n)$\;
        Let $I \leftarrow \emptyset$ be a multiset and
        $T \leftarrow O( \log (n/\delta)/(\epsilon^3 \delta^2))$\;
        % \footnote{According to Theorem \ref{theorem:Lp-NN}, it should be $O(\log(1/\delta)/\epsilon^3)$, but we can achieve a better bound of $O(\log(1/\delta)/\epsilon^2)$  by using $0\leq X_t, Y_t\leq 1$, instead of using $0\leq X_t, Y_t\leq 1/\epsilon$ in the proof. }
        
        \For{$t\in [T]$}{
            \For{$b\in [d]$}{
                Add $b$ to $I$ with probability  $p^{(b)}\triangleq \max_{(i,j)\in {\binom{n}{2}}} \frac{\left|c_i^{(b)} - c_j^{(b)}\right|}{\|c_i- c_j\|_1}$\;
            }
        }
        Let $R=\{r_1, r_2, \dots, r_n\}\subseteq \R^{I}$, where for $i \in [n], b \in I$, we have $r_i^{(b)} = c_i^{(b)} / p^{(b)}$\;
        % \tcp{Apply the $l_1$ dimension reduction $(M, F)$ by \cite{indyk06:l1-dimension-reduction} to $R$. It guarantees that $\forall (i,j)\in {n \choose 2}, F(Mr_i, Mr_j) \in [1-\epsilon\delta/100, 1+\epsilon\delta/100] \|r_i- r_j\|_1$ with probability $1-\delta/4$.}
        Let $m \leftarrow O(\log(n/\delta)/(\epsilon^2 \delta^2))$\;
        Sample $M\in \mathbb{R}^{[m] \times I}$, where each entry of $M$ follows the Cauchy distribution, whose density function is $c(x) = \frac{1}{\pi(1+x^2)}$. 
        Notice that $M(\cdot)\colon\mathbb{R}^{I} \rightarrow \mathbb{R}^m$ is an oblivious linear mapping \cite{indyk06:l1-dimension-reduction}\;

        \store{$I$, $M(R) = \{Mr_i|i\in [n]\}$, $\{p^{(b)}| b\in I\}$, $M$ }
    }
    \query{$(q)$ \tcp*[f]{Inputs: $q\in \R^d$}}{
        Query $q^{(b)}, \forall b\in I$\;
        Let $u\in \mathbb{R}^{I}$, where $\forall b \in I, u^{(b)} = q^{(b)} / p^{(b)}$\;
        Let $F((x_1, \dots, x_m), (y_1, \dots, y_m)) \coloneq \mathsf{median}(|x_1-y_1|, \dots, |x_m- y_m|)$\;
        Let $\hat{i} = \argmin_{i\in [n]} F(Mr_i, Mu)$\;
        \Return{$\hat{i}$}
    }
\end{algorithm}

% \begin{theorem}\label{theorem:sample-reuse-l1}
%     Consider a set $C$ of $n$ points $c_1, \dots, c_n \in \R^d$ equipped with Hamming distance. Given $0 < \epsilon < 1/4$, we can construct a data structure $B$(see \cref{alg:sample-reuse-l1}) of size $O(n \log(n/\delta) \log(d)/\epsilon^{2})$ words that, given any query $q\in \R^d$,  with probability $1-\delta$, reads $O(n \log (n/\delta)/\epsilon^{2})$ coordinates of $q$ and returns $i\in [d]$ where $c_i$ is a $(1+\epsilon)$-approximate nearest neighbor of $q$ in $C$. 
% \end{theorem}
% This is an easier version with constant probability
\begin{theorem}\label{theorem:sample-reuse-l1}
    Consider a set $C$ of $n$ points $c_1, \dots, c_n \in \R^d$ equipped with $\ell_1$ metric. Given $0 < \epsilon < 1/4$ and $0 < \delta < 1$, with probability $(1-\delta)$, we can efficiently construct a randomized data structure (See \cref{alg:sample-reuse-l1}) such that for any query point $q\in \R^d$ the following conditions hold with probability $1-\delta$. (i) The data structure reads $O(n \log (n/\delta)/(\epsilon^3\delta^2))$ coordinates $q$. (ii) It returns $i\in [d]$ where $c_i$ is a $(1+\epsilon)$-approximate nearest neighbor of $q$ in $C$. (iii) The data structure has a size of $O(n \log^2(n/\delta) \log(d)/(\epsilon^5\delta^4))$ words. 
\end{theorem}

To prove~\cref{theorem:sample-reuse-l1}, we consider the following four events and prove that each of them holds with probability $1-\delta'$, where $\delta' = \delta/4$. We establish the correctness of the data structure by conditioning on the first three events. By conditioning on the last event, we demonstrate the space and query complexity.
\begin{enumerate}
    \item The first event is that the distance between a query and its nearest neighbor is overestimated at most by the factor $4/\delta$.
    \item The second event is that for a query point the estimated distance to its nearest neighbor is significantly smaller than the estimated distance to any center that is not an approximate nearest neighbor.
    \item The third event is that the dimension reduction preserves the distances between the coordinate-selected points and the query.
    \item The fourth event is that at most near-linear many coordinates are sampled in the preprocessing and accessed in the query phase. 
\end{enumerate}
Finally, by applying a union bound, we show that all four events hold with probability $1-\delta$, which implies that the requirements of \Cref{theorem:sample-reuse-l1l2} are met.
We show the proof in the following subsections in line with the structure described above.

\subsubsection{Upper Bound for Estimate of Distance between Query and its Nearest Neighbor}\label{sec:first-section-l1}
We first show that the distance between a query point and its nearest neighbor is overestimated at most by a factor $1/\delta'$ with probability $1-\delta'$. 
%We use Markov's inequality for that.}
\begin{claim}\label{claim:Markovlone}
    Let $q\in \R^n$ be a query and let $i^*\in [n]$ be the index of its nearest neighbor. It holds that
    %\begin{align*}
        $\frac{1}{T}\|r_{i^*} - u\|_1 \leq \frac{1}{\delta'} \|c_{i^*} - q\|_1,$
    %\end{align*}
    with probability $1-\delta'$.
\end{claim}

\subsubsection{Coordinate Selection Preserves Distance Ratio}
In this section, we show that, the following holds: 
for any query point $q\in \R^n$, whose nearest neighbor is $c_{i^*}$, with probability $1-2\delta'$, for all $c_i\in C$ such that $\frac{\|c_i-q\|_1}{\|c_{i^*}-q\|_1} \geq 1+\epsilon$,  the estimated ratio is significantly  larger than one, that is $\frac{\|r_i-u\|_1}{\|r_{i^*}-u\|_1} \geq 1+\epsilon\delta'/4$.
This statement is formalized in \Cref{lemma:robust-feature-selection-l1}.

To prove the statement, we first partition the coordinates in $[d]$ into good and bad regions for each pair of points. 
In the good region are the coordinates on which the query is not too far away from one of the points. 
The remaining coordinates are in the bad region.
The intuition is that, restricted to the coordinates in the good region, the distance between one of the points and the query  is bounded.
This implies that we can estimate these distances using a Chernoff bound.
For the bad coordinates the intuition is that the distances from the query to the 2 points are so large that the ratio between these distances is close to one. Additionally these distances are not overestimated because of the previous event in \cref{sec:first-section-l1}.

\begin{definition}[Good region and bad region]
    Consider two points $c_1, c_2\in \R^d$, and a query point $q\in\R^d$. We define the good region $\Good(c_1, c_2, q)$ w.r.t.\ $(c_1,c_2,q)$ as $$\{z\in [d] \mid\min\{c_1^{(z)}, c_2^{(z)}\} - \frac{8}{\epsilon \delta'} |c_1^{(z)} - c_2^{(z)}| \leq q^{(z)} \leq \max\{c_1^{(z)}, c_2^{(z)}\} + \frac{8}{\epsilon \delta'} |c_1^{(z)} - c_2^{(z)}|\}.$$ We define the bad region w.r.t.\ $(c_1,c_2,q)$ as $\Bad(c_1, c_2, q) = [d]\setminus \Good(c_1, c_2, q)$. When there is no ambiguity, we use $\Good$ for $\Good(c_1, c_2, q)$, and $\Bad$ for $\Bad(c_1, c_2, q)$. 
\end{definition}
Additionally, we define the distance between points restricted to subsets of coordinates.
\begin{definition}[Restricted distance]
    Given $x,y \in \R^d$, we define the distance between $x$ and $y$ restricted to region $J\subseteq [d]$ as 
    $\dist_J(x,y) = \sum_{z\in J} \left|x^{(z)} - y^{(z)}\right|$. 
    Furthermore, for $u,v \in \R^I$, where $I$ is a multiset, whose elements are from $[d]$, we define $\dist_J(u,v) = \sum_{z\in I} \mathbbm{1}_{z\in J}\cdot\left|x^{(z)} - y^{(z)}\right|$. 
\end{definition}
%\nithin{Is $J$ also a multiset above?}

%For bounding distance estimates on the good region, we can use Chernoff bound. For distance estimates on the bad region, we need Markov's inequality. 
%This motivates our name of the regions.

The following lemma shows that for two points, whose distances to the query differ significantly, then the majority of the distance between the two points is present in the good region. The intuition of the lemma is that if the bad region contributes a lot to the distance $\|b-a\|_1$, it implies that the bad region contributes much to $\|b-q\|_1$ and $\|a-q\|_1$, then the ratio between the two distances should be very close to 1. %The lemma is the contrapositive statement of the intuition.

\begin{lemma}\label{lemma:good-region-is-large-l1}
Consider three points \(a, b, q \in \mathbb{R}^d\). If \(\frac{\|b - q\|_1}{\|a - q\|_1} \geq 1 + \epsilon\), then \(\frac{\dist_{\Good(a, b, q)}(b, a)}{\|b - a\|_1} \geq \frac{13}{16}\).
\end{lemma}

The following lemma shows that for two points $c_i, c_j$ and a query $q$ such that $q$ is significantly closer to $c_j$, we can estimate the distances, restricted to the good region, between one of the points and the query.
\begin{lemma}\label{lemma:chernoff_for_good_l1}
Let $c_i, c_j$ where $i,j\in [n]$ be such that for the query point $q \in \R^d$ it holds that $\|c_i - q\|_1 \geq (1+\epsilon)\|c_j - q\|_1$. Let $\Good$ be shorthand for $\Good(c_i, c_j, q)$ and $\Bad$ for $\Bad(c_i, c_j, q)$. The following statements are true:

%     \begin{enumerate}
%         \item $\dist_{\Good} (c_i, q) \geq (1+\epsilon) \dist_{\Good} (c_j, q)$.
%         \item With probability $1-\frac{\delta'}{2n}$ we have $\dist_{\Good} (r_i, u) \approx_{\frac{\epsilon}{16}} T\cdot \dist_{\Good}(c_i, q)$.
%         \item 
% \begin{enumerate}[(i)]        
%         \item If $\dist_{\Good} (c_j, q) \geq \frac18 \dist_{\Good} (c_i, c_j)$, then with probability $1-\frac{\delta'}{2n}$ we have             $\dist_{\Good} (r_j, u) \approx_{\frac{\epsilon}{16}} T\cdot \dist_{\Good}(c_j, q)$.

%         \item If $\dist_{\Good} (c_j, q) < \frac18 \dist_{\Good} (c_i, c_j)$, then with probability $1-\frac{\delta'}{2n}$ we have $\dist_\Good(r_j, u) \leq \frac14 \dist_\Good(c_i, c_j)$.
%     \end{enumerate}
% \end{enumerate}

    \begin{enumerate}
        \item $\dist_{\Good} (c_i, q) \geq (1+\epsilon) \dist_{\Good} (c_j, q)$.
        \item If $\dist_{\Good} (c_j, q) \geq \frac18 \dist_{\Good} (c_i, c_j)$, then with probability $1-\delta'/n$, 
        \begin{align*}
            \dist_{\Good} (r_i, u) \approx_{\frac{\epsilon}{16}} T\cdot \dist_{\Good}(c_i, q),\\
            \dist_{\Good} (r_j, u) \approx_{\frac{\epsilon}{16}} T\cdot \dist_{\Good}(c_j, q).
        \end{align*}
        \item If $\dist_{\Good} (c_j, q) < \frac18 \dist_{\Good} (c_i, c_j)$, then $\frac{\dist_{\Good} (c_i, q)}{\dist_{\Good} (c_j, q)} > 7$. With probability $1-\delta'/n$, 
        \begin{align*}
            \dist_{\Good} (r_i, u) &\approx_{\frac{\epsilon}{16}} T\cdot \dist_{\Good}(c_i, q),\\
            \frac{\dist_{\Good} (r_i, u)}{\dist_{\Good} (r_j, u)} &\geq 3.
        \end{align*}
    \end{enumerate}
\end{lemma}

In the following lemma we state that our bounds on the distance estimates are good enough. Concretely, for a point that is not an approximate nearest neighbor, the estimated distance between this point and the query is significantly larger than the estimated distance between the nearest neighbor and the query.

\begin{lemma}\label{lemma:robust-feature-selection-l1}
    For any query point $q \in \R^d$, with probability $1-2\delta'$, it holds for all $i\in [n]$, if $\frac{\|c_i-q\|_1}{\|c_{i^*}-q\|_1} \geq 1+\epsilon$, then $\frac{\|r_i-u\|_1}{\|r_{i^*}-u\|_1} \geq 1+\epsilon\delta'/4$.
\end{lemma}

\subsubsection{\texorpdfstring{$\ell_1$}{l1} Dimension Reduction Preserves the Distance Ratio}
\begin{theorem}[Theorem 1 in \cite{indyk06:l1-dimension-reduction}]\label{thm:indyk06}
    Consider 2 points $a, b \in \R^d$ and parameters $0<\epsilon, \delta < 1$. We sample a random matrix $M\in \R^{m \times d}$, where $m = O(\log(1/\delta)/(\epsilon^2))$ and each entry of $M$ is sampled from Cauchy distribution, whose density function is $c(x) = \frac{1}{\pi(1+x^2)}$. Then with probability $1-\delta$, $F(Ma, Mb) \approx_\epsilon \|a-b\|_1$, where $F((x_1, \dots, x_m), (y_1, \dots, y_m)) \coloneq \mathsf{median}(|x_1-y_1|, \dots, |x_m- y_m|)$.
\end{theorem}
\begin{claim}
    Let $q\in \R^n$ be a query. Furthermore, let $i^*\in [n]$ be the index of its nearest neighbor. \Cref{alg:sample-reuse-l1} outputs an index $i$ such that
    \begin{align*}
        \frac{\|c_i - q\|_1}{\|c_{i^*} - q\|_1} \leq 1+\epsilon
    \end{align*}
    with probability $1-3\delta'$.
\end{claim}
\begin{proof}
Applying \cref{thm:indyk06} with $m = O(\log(n/\delta')/(\epsilon^2 \delta'^2))$, given any $v \in R$, it holds with probability $1-\delta'/(n+1)$, that $F(Mv, Mu) \in [1-\epsilon\delta'/100, 1+\epsilon\delta'/100] \|v- u\|_1$. By union bound, with probability $1-\delta'$, it holds for all $v\in R$.
By \Cref{lemma:robust-feature-selection-l1} it holds with probability $1-2\delta'$ that for all $i\in [n]$, if $\frac{\|c_i - q\|_1}{\|c_{i^*} - q\|_1} \geq 1+\epsilon$, then $\frac{\|r_i - u\|_1}{\|r_{i^*} - u\|_1} \geq 1+\epsilon\delta'/4$, thus $\frac{F(Mr_i, Mu)}{F(Mr_{i^*}, Mu)} > 1$ with probability $1-3\delta'$, i.e. the algorithm will not output $i$ as the index of the approximate nearest neighbor.
\end{proof}

\subsubsection{Space and Query Complexity}\label{sec:l1-space-query-analysis}
Next, we argue the space and query complexity.
The main result in this section is that for any set $C$ of centers the sum of probabilities $\sum_{b\in [d]} p^{(b)}$ is upper bounded by the number of centers.

\begin{lemma}\label{lemma:math-program}
    $1\leq \sum_{b\in [d]} p^{(b)} \leq n$.
\end{lemma}
\begin{proof}
The lower bound for $\sum_{b\in [d]} p^{(b)}$ can be seen as follows: Fix arbitrary $i_1, j_1 \in [n]$ s.t.\ $\|c_{i_1} - c_{j_1}\|_1 \neq 0$. We have, $\sum_{b\in [d]} p^{(b)} \geq \sum_{b\in [d]} \frac{\left|c_{i_1}^{(b)} - c_{j_1}^{(b)}\right|}{\|c_{i_1} - c_{j_1}\|_1} = 1$.

Next we prove the upper bound for $\sum_{b\in [d]} p^{(b)}$. We first introduce an operation that we call ``$\collapse$''.
%\begin{definition}{\bf (collapse)}
    For real numbers $\tau_1 \leq \tau_2$, 
    $$
    \collapse_{\tau_1, \tau_2}(x) \coloneq 
    \begin{cases}
        x, & x \leq \tau_1, \\
        \tau_1, & \tau_1 < x \leq \tau_2, \\
        x - (\tau_2 - \tau_1), & x > \tau_2.
    \end{cases}
    $$
    For $\tau_1 > \tau_2$, $\collapse_{\tau_1, \tau_2}(x) := \collapse_{\tau_2, \tau_1}(x)$. 
    
    We also slightly abuse the notation, and define for $a, b, c \in \R^d$, 
    $$\collapse_{a,b}(c) \triangleq (\collapse_{a_1, b_1}(c_1), \collapse_{a_2, b_2}(c_2), \dots, \collapse_{a_d, b_d}(c_d)) \in \R^d.$$
%\end{definition}

\begin{claim}\label{lemma:collapse-property}
Let $x,y,\tau_1,\tau_2 \in \R$, where $x \geq y$. We have 
    (i) $x - |\tau_1 -\tau_2| \leq \collapse_{\tau_1, \tau_2}(x) \leq x$ and (ii) $x-y - |\tau_1 -\tau_2| \leq\collapse_{\tau_1, \tau_2}(x) - \collapse_{\tau_1, \tau_2}(y) \leq x-y.$
\end{claim}
\begin{proof}
    Denote $x' = \collapse_{\tau_1, \tau_2}(x)$ and $y' = \collapse_{\tau_1, \tau_2}(y)$. Assume w.l.o.g.\ that $\tau_1 \leq \tau_2$, otherwise we switch them.
    
    For the first argument, $\forall x\in \R, x - |\tau_1 -\tau_2| \leq \collapse_{\tau_1, \tau_2}(x) \leq x$,  this is true because of the definition of collapse.
    
    For the second argument, $\forall x\geq y \in \R$, $x-y - |\tau_1 -\tau_2| \leq\collapse_{\tau_1, \tau_2}(x) - \collapse_{\tau_1, \tau_2}(y) \leq x-y,$ consider $$f(x) = x-\collapse_{\tau_1, \tau_2}(x) = \begin{cases}
        0, &x\leq \tau_1,\\
        x-\tau_1, &\tau_1 \leq x \leq \tau_2,\\
        \tau_2 - \tau_1, &x > \tau_2.
    \end{cases}$$
    Notice that $f(x)$ is a non-decreasing function, thus $f(x) \geq f(y)$, which proves $x-y \geq x'-y'$. Notice that $0\leq f(x) \leq \tau_2 - \tau_1$, thus $f(x) \leq f(y) + (\tau_2 -\tau_1)$, which proves $x'-y' \geq x-y-(\tau_2 - \tau_1)$.
\end{proof}

\begin{claim}\label{claim:collapse}
Let $c_{i_0}, c_{j_0}$ be the closest pair of 
points in a set $C = \{c_1, c_2, \dots, c_n\}$ of points, i.e., $(i_0, j_0) = \argmin_{(i,j)\colon i\neq j} \|c_i - c_j \|_1$.    
Define $P_n(C)$ to be $\sum_{b \in [d]} \max_{i\neq j}\frac{|c_i^{(b)} - c_j^{(b)}|}{\|c_i - c_j\|_1}$.
The set $C' = \{c_i' = \collapse_{c_{i_0}, c_{j_0}}(c_i) | i\in [n] \}$ is such that (i) it contains $n-1$ points and (ii) $P_n(C) \leq P_{n-1}(C') + 1.$  
\end{claim}

\begin{proof}
Note, by the definition of the collapse operation,
    that $c'_{i_0} = c'_{j_0}$. Hence, $C'$ has at most $n-1$ points, proving part (i) of the claim.

    To prove (ii),
  define $p_b = \max_{c_i, c_j \in C, i\neq j}\frac{|c_i^{(b)} - c_j^{(b)}|}{\|c_i - c_j\|_1}$ and $p_b' = \max_{c'_i, c'_j \in C', i\neq j}\frac{|{c'_i}^{(b)} - {c'_j}^{(b)}|}{\|c'_i - c'_j\|_1}$. 
    We prove that $p_b - p_b' \leq \frac{| c_{i_0}^b - c_{j_0}^b |}{\|c_{i_0} - c_{j_0}\|_1}, \forall b \in [d]$. The claim follows, since $P_k(C) - P_{k-1}(C') = \sum_b (p_b - p_b') \leq \sum_b \frac{| c_{i_0}^{(b)} - c_{j_0}^{(b)} |}{\|c_{i_0} - c_{j_0}\|_1} = 1$.

    Fix some $b$, choose $(u,v) = \argmax_{(i,j)\colon c_i^{(b)} \neq c_j^{(b)}} \frac{| c_i^{(b)} - c_j^{(b)} |}{\|c_i - c_j\|_1}$, i.e. $p_b = \frac{| c_u^{(b)} - c_v^{(b)} |}{\|c_u - c_v\|_1}$.

    If ${c'_u}^{(b)} \neq {c'_v}^{(b)}$, ${p_b}' \geq \frac{| {c'_u}^{(b)} - {c'_v}^{(b)} |}{\|c'_u - c'_v\|_1} \geq \frac{| {c_u^{(b)}} - {c_v^{(b)}} | - |{c_{i_0}^{(b)}} - {c_{j_0}^{(b)}}|}{\|c_u - c_v\|_1}.$ The second inequality follows from \Cref{lemma:collapse-property}. Thus 
    $
    p_b - p_b' \leq \frac{|{c_{i_0}^{(b)}} - {c_{j_0}^{(b)}}|}{\|c_u - c_v\|_1} \leq \frac{|{c_{i_0}^{(b)}} - {c_{j_0}^{(b)}}|}{\|c_{i_0} - c_{j_0}\|_1}.
    $
    If ${c'_u}^{(b)} = {c'_v}^{(b)}$, it means $|{c_u^{(b)}} - {c_v^{(b)}}| \leq |{c_{i_0}^{(b)}} - {c_{j_0}^{(b)}}|$, thus 
    $
    p_b - p_b' \leq p_b = \frac{| c_u^{(b)} - c_v^{(b)} |}{\|c_u - c_v\|_1} \leq \frac{|{c_{i_0}^{(b)}} - {c_{j_0}^{(b)}}|}{\|c_{i_0} - c_{j_0}\|_1}.%\qedhere
    $
\end{proof}
Solving the recurrence given by \cref{claim:collapse} with the boundary condition that for a set $C''$ of two points, $P_2(C'') = 1$, we get that $P_n(C) \leq n$.
Since $P_n(C)$, for the set of $n$ points, is by definition, equal to $\sum_{b \in [d]} p^{(b)}$, the statement of the lemma follows. 
\end{proof}

\begin{claim}\label{clm:I-size-upperbound-l1}
    With probability $1-\delta'$, we have $|I| \leq 2Tn$.    
\end{claim}
\begin{proof}
   We define $X_{t,z}$ for $t \in [T], z \in [d]$ to be the indicator random variable for the event that we add $z$ to $I$ at iteration $t$. Note that $\Pr[X_{t,z} = 1] = p^{(z)}$.
   % we define the following random variables $\forall t\in [T], z\in [d]$,
    %\begin{equation*}
     %   X_{t, z} =  \begin{cases} 
     % 1, & w.p.\; p^{(z)},\quad\text{i.e. if we add $z$ to $I$ at iteration $t$.}\\
     % 0, & w.p.\; 1-p^{(z)}, \quad\text{i.e. if we don't add $z$ to $I$ at iteration $t$.}
      %\end{cases}
    %\end{equation*}
    We can see that $|I| = \sum_{t\in[T], z\in[d]} X_{t,z}$.
    Hence, $\E[|I|] = \sum_{t\in[T], z\in[d]} p^{(z)} \leq T\cdot \sum_{b \in [d]} p^{(b)} \leq T\cdot n$, where the last inequality follows from \cref{lemma:math-program}. Similarly, we can see that 
    $\E[|I|] \geq T$.
    Using Chernoff bound, we have 
      $\Pr\left[|I| - \E\left[|I|\right] \geq \E\left[|I|\right]\right] \leq \exp (-\E[|I|]/3) \leq \exp(-T/3) \leq \delta/4.$
    Thus, with probability $1-\delta/4$, we have $|I| \leq 2 E[|I|] \leq 2Tn$.
\end{proof}

% The expected query complexity bound follows from the above claim. Now let us analyze the space complexity. $M(R)$ is of wordsize $O(n\log(n/\delta)/\epsilon^2)$, $M$ is of wordsize $O(|I|\log(n/\delta)/\epsilon^2)$, $I$ and $\{p^{(b)}| b\in I\}$ are of wordsize $n$ in expectation, $F$ is of constant size. Using Chernoff bound to upper-bound the size of $|I|$ finishes the proof.

\begin{lemma}
    With prob. $1-\delta'$, the space complexity of \cref{alg:sample-reuse-l1} is $O(\frac{n\log^2(n/\delta)\log(d)}{\epsilon^5 \delta^4})$ wordsize, and the number of coordinates that we query is $O(\frac{n\log(n/\delta)}{\epsilon^3 \delta^2})$. 
\end{lemma}
\begin{proof}
    By \cref{clm:I-size-upperbound-l1}, we have that the number of coordinates of $q$ that we query is $|I| = O(Tn) = O(\frac{n\log(n/\delta)}{\epsilon^3 \delta^2})$. Next we analyze the space complexity of the data structure. Matrix $M$ is of wordsize $O(m|I|)$, $M(R)$ is of wordsize $O(mn)$, $I$ is of wordsize $O(|I| \log(d))$. Thus the space complexity of the our data strcuture is $O(m|I| + mn +|I|\log(d)) = O(\frac{n\log^2(n/\delta)\log(d)}{\epsilon^5 \delta^4})$ wordsize.
\end{proof}

 Using union bound, the success probability of the data structure is at least $1-4\delta' = 1-\delta$.

\bibliography{socg-main.bbl}
\bibstyle{plainurl}

\newpage
\appendix

\section{Warm-Up: Data Structure for \texorpdfstring{$\ell_1$}{l1} Metric with \texorpdfstring{$\tilde{O}(n^2)$}{near-quadratic} Space and \texorpdfstring{$\tilde{O}(n)$}{near-linear} Query Complexity}\label{sec:warm-up}

As a warm-up, we describe in this section our data structure with $O(n^2 \log(n/\delta) \log(d)/\epsilon^2)$ space and $O(n \log (n/\delta)/\epsilon^2)$ query complexity for approximate nearest neighbor problem under the $\ell_1$ metric. For simplicity, we make the assumption that the query is in the scaled bounding box of each pair of points defined as follows.

\begin{definition}[Scaled Bounding Box]
    For any two points $a,b\in \R^d$, the scaled bounding box $B_{a,b}$ specified by $(a,b)$ is $$\begin{aligned}
    \left\{x\in \R^d \mid \forall z\in [d], \min\{a^{(z)}, b^{(z)}\} - 100 \cdot |a^{(z)} - b^{(z)}| \leq x^{(z)} \right.\\
    \left.\leq \max\{a^{(z)}, b^{(z)}\} + 100 \cdot |a^{(z)} - b^{(z)}|\right\}.
    \end{aligned}$$
\end{definition}

We show the following.

\begin{theorem}\label{thm:warm-up-l1}
    Consider a set $C$ of $n$ points $c_1, \dots, c_n \in \mathbb{R}^d$ equipped with $\ell_1$ metric. Given parameters $0 < \epsilon, \delta < 1$, we can construct a randomized data structure such that for any query point $q\in R^d$, which lies in the bounding box of each pair of points, i.e.\ $\forall i,j\in [n], q \in B_{c_i, c_j}$, the following conditions hold with probability $1 - \delta$. (i) The data structure reads $O(n \log (n/\delta)/\epsilon^2)$ coordinates of $q$. (ii) It returns $i \in [n]$ where $c_i$ is a $(1+\epsilon)$-approximate nearest neighbor of $q$ in $C$. (iii) The data structure has a size of $O(n^2 \log(n/\delta) \log(d)/\epsilon^2)$ words.
\end{theorem}

\cref{thm:warm-up-l1} is a special case of \cref{theorem:Lp-NN}, where a rigorous proof is provided. In this section, we present a more concise proof to provide intuition. Our approach is the following. We first describe a data structure for the approximate nearest neighbor problem for two points (i.e. a 2-point comparator), and then employ a naive min-finding procedure with oracle access to the 2-point comparator, to extend to the approximate nearest neighbor with $n$ points.

\subsection{A Strong Comparator for Two Points}

In this section we describe our approximate nearest neighbor data strucutre for 2 points, which we refer to as 2-point comparator (see \cref{alg:strong-2points}). It gets two points $a, b \in \R^d$ and parameters $\epsilon, \delta \in (0,1)$ as inputs to its preprocessing phase. In the preprocessing phase, it iteratively samples coordinates from $[d]$ to form a multiset $I$. Specifically, in a single iteration, it samples each coordinate with probability proportional to the distance between $a$ and $b$ on that coordinate. The data structure stores $I$ as well as the coordinates of $a$ and $b$ in $I$. In the query phase, it receives a query $q \in B_{a,b}$ upon which it probes $q$ on all coordinates in $I$ to compute estimates $x$ and $y$ of the distances of $q$ to $a$ and $b$, respectively. Its final answer, as to which among $a$ or $b$ is the nearer neighbor of $q$, depends on the ratio $x/y$.

We analyze the correctness of \cref{alg:strong-2points} in \cref{lem:strong-comparator-l1}. In particular, we show the following strong guarantees: (i) if one point is significantly closer to the query than the other point, the comparator outputs the correct point with probability $1- \delta$, (ii) otherwise, the comparator may output a special symbol $\perp$, but never outputs a wrong point as the nearest neighbor. %The space and query complexity of our $2$-point comparator are, with probability $1-\delta$, $$

\begin{algorithm}[htbp]\label{alg:strong-2points}
    \caption{An approximate nearest neighbor data structure $A_{(a,b)}$ for two points $a,b\in \R^d$.}
    \SetKwProg{preprocessing}{Preprocessing}{}{}
    \SetKwProg{query}{Query}{}{}
    \SetKw{store}{store}
    \preprocessing{$(a,b, \epsilon, \delta)$\tcp*[f]{Inputs: $a,b\in \R^d, \epsilon, \delta\in (0,1)$}}{
    
        Let $I \leftarrow \emptyset$ be a multiset and
        $T \leftarrow O( \log (1/\delta)/\epsilon^2)$\;
        \For{$t\in [T]$}{
            Sample $i_t\in [d]$, where $\forall z\in [d], \Pr[i_t = z] = p^{(z)} = \frac{\left|b^{(z)} - a^{(z)}\right|}{\|b-a\|_1}$\;
            Add $i_t$ to $I$\;
        }
        \store{$I, a^{(I)}, b^{(I)}$, where $a^{(I)} = \{a^{(i)} | i\in I\}$}.
    }
    \query{$(q)$ \tcp*[f]{Inputs: $q\in B_{a,b} \subseteq \R^d$}}{
        Let $x,y \leftarrow 0$\;
        \For{$i\in I$}{
            Query $q$ on coordinate ${i}$, i.e. $q^{(i)}$\; 
             
            Let $x \leftarrow x+\frac{\left|q^{(i)} - a^{(i)}\right|}{p^{(i)}}, y \leftarrow y+\frac{\left|q^{(i)} - b^{(i)}\right|}{p^{(i)}}$\;
        }
        \uIf{$\frac{y}{x} \geq 1+\frac{\epsilon}{2}$}{
        \Return{$a$ as nearest neighbor}
        }
        \uElseIf{$\frac{x}{y} \geq 1+\frac{\epsilon}{2}$}{
            \Return{$b$ as nearest neighbor}
        }
        \Else{
            \Return{$\bot$}
        }
    }
\end{algorithm}

\begin{lemma}[Strong Comparator Correctness]\label{lem:strong-comparator-l1}
   \cref{alg:strong-2points}, given inputs $a,b \in \mathbb{R}^d, \epsilon,\delta \in (0,1)$, satisfies the following conditions with probability $1-\delta$:
    \begin{enumerate}
        \item if $\frac{\|b-q\|_1}{\|a-q\|_1} \geq 1+\epsilon$, it outputs $a$; if $\frac{\|a-q\|_1}{\|b-q\|_1} \geq 1+\epsilon$, it outputs $b$,
        \item if $\frac{\|b-q\|_1}{\|a-q\|_1} \geq 1$, it outputs $a$ or $\bot$; if $\frac{\|a-q\|_1}{\|b-q\|_1} \geq 1$, it outputs $b$ or $\bot$.
    \end{enumerate}
\end{lemma}
\begin{proof}[Proof Sketch]
    In this proof, we assume that $\|b-q\|_1 \geq \|a-q\|_1$, since the other case is symmetric. Together with the fact $\|b-q\|_1 + \|a-q\|_1 \geq \|b-a\|_1$, we get $\|b-q\|_1 \geq \frac12 \|b-a\|_1$.

The algorithm makes its decision based on the relative values of $x$ and $y$, which are estimates for the distances from $q$ to $a$ and from $q$ to $b$, respectively.
    To analyze the ratio $y/x$, we need to define some random variables to analyze how we sample $I$ in the preprocessing. 
    Let $X_t = \frac{\left|q^{(i_t)} - a^{(i_t)}\right|}{p^{(i_t)}}, Y_t = \frac{\left|q^{(i_t)} - b^{(i_t)}\right|}{p^{(i_t)}}, \forall t\in [T]$. We know that $x = \sum_{t\in [T]} X_t$ and $y = \sum_{t\in [T]} Y_t$. These nonnegative random variables are bounded: $X_t = \frac{\left|q^{(i_t)} - a^{(i_t)}\right|}{\left|b^{(i_t)} - a^{(i_t)}\right|}\cdot \|b-a\|_1 \leq 101\cdot\|b-a\|_1$, because $q$ lies in the bounding box specified by $a,b$. Similarly we have $Y_t \leq 101\cdot\|b-a\|_1, \forall t\in [T]$. 
    The expectation of $X_t$ is $\E[X_t] = \sum_{z\in [d]} p^{(z)} \cdot \frac{\left|q^{(z)} - a^{(z)}\right|}{p^{(z)}} = \|a-q\|_1$.
    Similarly $\E[Y_t] = \|b-q\|_1$. 

    Since $\E[Y_t] = \|b-q\|_1 \geq \frac12 \|b-a\|_1$, using Chernoff bound, we have $\Pr\left[|y - \E[y]| > \frac{\epsilon}{8} \E[y]\right] \allowbreak\leq \frac{\delta}{2}$ by choosing $T = \Theta(\log(1/\delta)/\epsilon^2)$.
    To analyze $x$, we consider the following two cases:
    
    \noindent \underline{Case 1.} If $\E[X_t] = \|a-q\|_1 \geq \frac18 \|b-a\|_1$, then using Chernoff bound, we have \\$\Pr\left[|x - \E[x]| > \frac{\epsilon}{8} \E[x]\right] \leq \frac{\delta}{2}$ by choosing $T = \Theta(\log(1/\delta)/\epsilon^2)$.
    
    \noindent \underline{Case 2.} If $\E[X_t] = \|a-q\|_1 < \frac18 \|b-a\|_1$, then using Chernoff bound, we have \\$\Pr\left[|x - \E[x]| > \frac{T}8 \|b-a\|_1\right] \leq \frac{\delta}{2}$ by choosing $T = \Omega(\log(1/\delta))$. 

To prove item 1 of~\cref{lem:strong-comparator-l1}, suppose that $\frac{\|b-q\|_1}{\|a-q\|_1} \geq 1+\epsilon$. If we are in Case 1, with probability $1-\delta$ we have $x \approx_{\frac{\epsilon}{8}} T\|a-q\|_1$ and $y \approx_{\frac{\epsilon}{8}} T\|b-q\|_1$. This implies that, with probability $1-\delta$ we get $\frac{y}{x} \geq \frac{1-\frac{\epsilon}{8}}{1+\frac{\epsilon}{8}} \cdot (1+\epsilon) \geq 1+\frac{\epsilon}{2}$. Thus the querying phase correctly outputs $a$ as the nearest neighbor. If we are in Case 2, with probability $1-\delta$ we have $x \leq \frac{1}{4} T\|b-a\|_1$, and $y \approx_{\frac{\epsilon}{8}} T\|b-q\|_1$, which implies that $\frac{y}{x} \geq \frac{(1-\frac{\epsilon}{8})(1-\frac18)}{\frac{1}{4} } \geq 3$. Thus the querying phase correctly outputs $a$.

    The second item of~\cref{lem:strong-comparator-l1} can be proved similarly.
\end{proof}

\subsection{Data Structure for \texorpdfstring{$\ell_1$}{l1} Metric with \texorpdfstring{$\tilde{O}(n^2)$}{near-quadratic} Space}
In this section, we show that combining \cref{alg:strong-2points} with a naive minimum-finding approach using an imprecise comparison oracle gives a data structure for $n$ points under $\ell_1$ metric (see \cref{alg:naive-min-finding}) that uses  $\tilde{O}(n^2)$ space. \cref{alg:naive-min-finding} maintains a candidate for the nearest neighbor and evaluates each point against this candidate using the two-point comparator. If a point is found to be significantly closer to the query than the current candidate, the candidate is updated accordingly.

\begin{algorithm}[htp]\label{alg:naive-min-finding}
    \caption{An approximate nearest neighbor data structure for $n$ points in $\R^d$ for $\ell_1$ metric.}
    \SetKwProg{preprocessing}{Preprocessing}{}{}
    \SetKwProg{query}{Query}{}{}
    \SetKwProg{comparator}{$E$}{}{}
    \SetKw{store}{store}
    \preprocessing{$(C, \epsilon, \delta)$\tcp*[f]{Inputs: $C=\{c_1, \dots, c_n\}\subseteq \R^d, \epsilon \in (0, \frac{1}{4}), \delta\in (0,1)$}}{
        \For{$1\leq i < j \leq n$}{
            $Preprocessing(c_i, c_j, \epsilon, \delta/n^2)$ \tcp*{Construct $A_{(c_i, c_j)}$ from \cref{alg:strong-2points}}
        }
    }
    \query{$(q)$ \tcp*[f]{Inputs: $q\in \R^d, \forall i,j\in [n], q\in B_{c_i, c_j}$}}{
        Let $j = 1$\;
        \For{$2\leq i \leq n$}{
            In $A_{(c_j, c_i)}$, run $Query(q)$\;
            \If{the query outputs $c_i$}{
            Let $j \leftarrow i$\;
            }      
        }
        \Return{$j$}
    }
\end{algorithm}

\begin{proof}(for \cref{thm:warm-up-l1})
    By a union bound, we get that with probability $1 - \delta$, for all $1 \leq i < j \leq n$, the data structure $A_{(c_i,c_j)}$ does not fail. We condition on this event in the rest of the proof. Notice that $j$ is a variable under updating.  

    \cref{alg:naive-min-finding} maintains an index $j \in [n]$ throughout its execution such that $c_j$ is a candidate $(1+\epsilon)$-approximate nearest neighbor of $q$. Initially, $j$ is set to $1$. The algorithm runs the 2-point comparator $n-1$ times, where in the $i$-th iteration, a comparison is made between $c_i$ and the current nearest neighbor candidate $c_j$. If $A_{c_i, c_j}$ outputs $c_i$, then we know, by \cref{lem:strong-comparator-l1}, that $c_i$ is strictly closer to $q$ than $c_j$. Additionally, if $c_i$ were closer to $q$ by a factor of $1 + \epsilon$ or more, then we are guaranteed that $c_i$ is output by $A_{c_i, c_j}$.
    
    Let $i^*$ be such that $c_{i^*}$ is the true nearest neighbor of $q$. If the value of $j$ is updated in iteration $i^*$, then \cref{alg:naive-min-finding} will not update it in any further iteration and we output the true nearest neighbor. If the value of $j$ is not updated in iteration $i^*$, then it must be the case that $||q - c_j||_1 \leq (1+\epsilon)\cdot ||q - c_{i^*}||_1$ and the final output is also a $1+\epsilon$ approximate nearest neighbor of $q$.

   % \zihang{changed the following paragraph}

    %From the second item of \cref{lem:strong-comparator-l1}, we know that $\|c_j - q\|_1$ is decreasing through the algorithm, i.e. if in iteration $i$, the index $j$ is going to be updated to $i$, then $\|c_j-q\|_1 > \|c_i-q\|$. Otherwise $\frac{\|c_i-q\|_1}{\|c_j-q\|_1} \geq 1$, the comparator never outputs $c_i$, which contradicts that $j$ is updated by $i$. \nithin{What do you mean when you say that '$\|c_j - q\|_1$ is decreasing through the algorithm'? $\|c_j - q\|_1$ is a fixed quantity. How can it increase or decrease? And what does it mean to decrease 'through the algorithm'?}\zihang{ I mean $j$ is under updating throughtout the algo}

    %Let $i^*$ be the index of the nearest neighbor. We always compare $c_{i^*}$ with some $c_i$, either $j=i$ or $j = i^*$. If $j = i^*$, $j$ will not be updated till the end. Thus we focus on the case $j=i$. If the comparator outputs $c_{i^*}$, then $j = i^*$ till the end. If the comparator never outputs $c_i$, otherwise it contradicts with the fact that $i^*$ be the index of the nearest neighbor. If the comparator outputs $\bot$, from \cref{lem:strong-comparator-l1}, we know $\frac{\|c_j-q\|_1}{\|c_{i^*} -q \|_1} \leq 1+\epsilon$, i.e. $c_j$ is $(1+\epsilon)$-approximate nearest neighbor. Since $\|c_j - q\|_1$ is decreasing through the algorithm, the finial $j$ is also index of some $(1+\epsilon)$-approximate nearest neighbor.

   % To conclude, we prove that with probability $1-\delta$, it outputs a $(1+\epsilon)$-approximate nearest neighbor. 
    Since the data structure consists of $n^2$ copies of data structure in \cref{alg:strong-2points}, the space complexity of \cref{alg:naive-min-finding} is $n^2 \times O(\log(n^2/\delta)/\epsilon^2) = O(n^2 \log(n/\delta) \log(d)/\epsilon^2)$. We make $n-1$ queries to the data structure in \cref{alg:strong-2points}, leading to the query complexity $O(n \log (n/\delta)/\epsilon^2)$.
\end{proof}

\section{Missing Proofs from \texorpdfstring{\Cref{sec:sample-reuse}}{Section 4} for \texorpdfstring{$\ell_1$}{l1}-metric and complete proof \texorpdfstring{$\ell_2$}{l2}-metric}\label{sec:missing-proof}
\subsection{Missing Proofs from \texorpdfstring{\Cref{sec:sample-reuse}}{Section 4} for \texorpdfstring{$\ell_1$}{l1}-metric}\label{sec:appendixforlone}

\begin{proof}[Proof of \cref{claim:Markovlone}]
    For any $i\in [n]$, we have $\E[\|r_i - u\|_1] = T \|c_i - q\|_1$, because 
    \begin{align*}
        \E[\|r_i - u\|_1] &= \E\left[\sum_{b\in I} \left|r_i^{(b)} - u^{(b)}\right|\right] \\ &= \E\left[\sum_{t\in [T], b\in [d]} \mathbbm{1}_{\{\text{in iteration $t$, we add $b$ to $I$\}}}\left|r_i^{(b)} - u^{(b)}\right|\right] \\
        &= \sum_{t\in [T], b\in [d]} p^{(b)}\left|r_i^{(b)} - u^{(b)}\right| = \sum_{t\in [T], b\in [d]} p^{(b)}\left|\frac{c_i^{(b)}}{p^{(b)}} - \frac{q^{(b)}}{p^{(b)}}\right| \\
        &= \sum_{t\in [T], b\in [d]} \left|c_i^{(b)} - q^{(b)}\right| = T \|c_i - q\|_1.
    \end{align*}
    Using Markov's inequality, we have with probability $1-\delta'$,
\begin{equation*}
    \|r_{i^*} - u\|_1 \leq \frac{1}{\delta'}\E[\|r_i - u\|_1]  \leq \frac{T}{\delta'} \|c_{i^*} - q\|_1. \qedhere 
\end{equation*}
\end{proof}

\begin{proof}[Proof of \cref{lemma:good-region-is-large-l1}]
By scaling the inequality $(1 + \epsilon) \|a - q\|_1 \leq \|b - q\|_1$, we get that
\[
\left(2 + \frac{2}{\epsilon}\right) \|a - q\|_1 \leq \frac{2}{\epsilon} \|b - q\|_1.
\]
Next, by rearranging and using the triangle inequality, we have:
\[
\|b - q\|_1 + \|a - q\|_1 \leq \left(1 + \frac{2}{\epsilon}\right) (\|b-q\|_1 - \|a-q\|_1) \leq \left(1 + \frac{2}{\epsilon}\right) \|b - a\|_1.
\]
Since \(\Bad(a, b, q) \subseteq [d]\), we have:
\begin{equation*}
    \begin{aligned}
        \|b - q\|_1 + \|a - q\|_1 \geq \dist_{\Bad(a, b, q)}(b, q) + \dist_{\Bad(a, b, q)}(a, q) \\
        \geq \sum_{z \in \Bad(a, b, q)} \left(\frac{8}{\epsilon \delta'} + \frac{8}{\epsilon \delta'}\right)|b^{(z)} - a^{(z)}|.
    \end{aligned}
\end{equation*}

This implies that: 
\[ 
\frac{16}{\epsilon \delta'} (\|b - a\|_1 - \dist_{\Good(a, b, q)}(b, a)) \leq \left(1 + \frac{2}{\epsilon}\right) \|b - a\|_1.
\]
Rearranging the terms, we finally get:
\[
\frac{\dist_{\Good(a, b, q)}(b, a)}{\|b - a\|_1} \geq \frac{14 - \epsilon \delta'}{16} \geq \frac{13}{16}. \qedhere
\]
\begin{comment}
\[
\dist_{\Bad(a, b, q)}(b, q) + \dist_{\Bad(a, b, q)}(a, q) \leq \left(1 + \frac{2}{\epsilon}\right) \|b - a\|_1
\]
Expanding the distances in the \(\Bad(a, b, q)\) region, we get:
\[
\sum_{z \in \Bad(a, b, q)} \left(\frac{8}{\epsilon \delta'} + \frac{8}{\epsilon \delta'}\right)|b^{(z)} - a^{(z)}| \leq \sum_{z \in \Bad(a, b, q)} (|b^{(z)} - q^{(z)}| + |a^{(z)} - q^{(z)}|) \leq \left(1 + \frac{2}{\epsilon}\right) \|b - a\|_1.
\]
which implies:
\[
\frac{16}{\epsilon \delta'} (\|b - a\|_1 - \dist_{\Good(a, b, q)}(b, a)) = \frac{16}{\epsilon \delta'} \dist_{\Bad(a, b, q)}(b, a) \leq \left(1 + \frac{2}{\epsilon}\right) \|b - a\|_1.
\]
\[
\left(\frac{16}{\epsilon \delta'} - 1 - \frac{2}{\epsilon}\right) \|b - a\|_1 \leq \frac{16}{\epsilon \delta'} \dist_{\Good(a, b, q)}(b, a).
\]
Finally, solving for the ratio \(\frac{\dist_{\Good(a, b, q)}(b, a)}{\|b - a\|_1}\), we obtain:
\[
\frac{\dist_{\Good(a, b, q)}(b, a)}{\|b - a\|_1} \geq \frac{14 - \epsilon \delta'}{16} \geq \frac{13}{16}.
\]
Thus, the claim is proven.
\end{comment}
\end{proof}

\begin{proof}[Proof of \cref{lemma:chernoff_for_good_l1}]
    Combining the following two properties,
    \begin{align*}
         \frac{\dist_{\Good} (c_i, q) + \dist_{\Bad} (c_i, q)}{\dist_{\Good} (c_j, q)+\dist_{\Bad} (c_j, q)} = \frac{\|c_i - q\|_1}{\|c_j - q\|_1} \geq 1+\epsilon, \tag{assumption}\\
         \frac{\dist_{\Bad} (c_i, q)}{\dist_{\Bad} (c_j, q)} \leq 1+\frac{\epsilon\delta'}{8} < 1+\epsilon. \tag{definition of bad region}         
    \end{align*}
    implies
    \begin{equation*}
        \frac{\dist_{\Good} (c_i, q)}{\dist_{\Good} (c_j, q)} \geq 1+\epsilon.
    \end{equation*}
    Therefore the first item of the lemma is proved.
    
    In the algorithm, we construct $\{r_i\mid i\in [n]\}$ and $u$ through sampling a multiset $I$. To analyze how the data structure samples $I$ in the preprocessing, we define the following random variables $\forall t\in [T], z\in \Good$
    \[P_{t, z} =  \begin{cases} 
      \frac{|q^{(z)} - c_i^{(z)}|}{p^{(z)} }, & w.p.\; p^{(z)},\hspace{0.95cm}\text{i.e. if we add $z$ to $I$ at iteration $t$.}\\
      0, & w.p.\; 1-p^{(z)}, \quad\text{i.e. if we do not add $z$ to $I$ at iteration $t$.}
   \end{cases}\]
   \[Q_{t, z} =  \begin{cases} 
      \frac{|q^{(z)} - c_j^{(z)}|}{p^{(z)} } , & w.p.\; p^{(z)},\hspace{0.95cm}\text{i.e. if we add $z$ to $I$ at iteration $t$.}\\
      0, & w.p.\; 1-p^{(z)}, \quad\text{i.e. if we do not add $z$ to $I$ at iteration $t$.}
   \end{cases}\]
   Notice that $\sum_{t\in [T],z\in \Good} P_{t,z} = \dist_\Good(r_i, u)$ and $\sum_{t\in [T],z\in \Good} Q_{t,z} = \dist_\Good(r_j, u)$.  

    Observe that $\forall t\in [T], z\in \Good$, $0 \leq P_{t,z} \leq \left(1+\frac{8}{\epsilon\delta'}\right) \|c_i - c_j\|_1$. This is because, if $c_i^{(z)} = c_j^{(z)}$, from the definition of good region we have $q^{(z)} = c_i^{(z)}$, thus $P_{t,z} = 0$; otherwise $c_i^{(z)} \neq c_j^{(z)}$, 
    \begin{align*}
        0 \leq P_{t,z} \leq \frac{|q^{(z)} - c_i^{(z)}|}{|c_i^{(z)} - c_j^{(z)}| } \cdot \|c_i - c_j\|_1 \leq \left(1+\frac{8}{\epsilon\delta'}\right) \|c_i - c_j\|_1,
    \end{align*}
    where the second inequality comes from ``oversampling'': $p^{(z)} \geq \frac{|c_i^{(z)} - c_j^{(z)}|}{\|c_i - c_j\|_1}$, and the third inequality comes from the definition of good region.
    By a similar argument, we have $0 \leq Q_{t,z}\leq \left(1+\frac{8}{\epsilon\delta'}\right) \|c_i - c_j\|_1$. 

    Notice $\forall t\in [T], \E[\sum_{z\in \Good} P_{t,z}] = \dist_{\Good}(c_i, q)$ and  $\E[\sum_{z\in \Good} Q_{t,z}] = \dist_{\Good}(c_j, q)$. Because 
    \begin{equation*}
    \begin{aligned}
        \E[\sum_{z\in \Good} P_{t,z}] &= \sum_{z\in \Good} \E[P_{t,z}] = \sum_{z\in \Good} p^{(z)} \cdot \frac{|q^{(z)} - c_i^{(z)}|}{p^{(z)}} = \sum_{z\in \Good} |q^{(z)} - c_i^{(z)}| \\
        &= \dist_{\Good}(c_i, q).
    \end{aligned}
    \end{equation*}
    Since $\forall t\in [T], \E[\sum_{z\in \Good} (P_{t,z} + Q_{t,z})] = \dist_{\Good}(c_i, q) + \dist_{\Good}(c_j, q)\geq  \dist_{\Good}(c_i, c_j)$, we have $\E[\sum_{z\in \Good}P_{t,z}] \geq \frac{1}{2} \dist_{\Good}(c_i, c_j)$. In the rest of the proof, we use $\E[\sum_{z\in \Good} P_{t,z}]$ without specifying $t\in [T]$, because they have the same value for all $t\in [T]$. Let $0<\epsilon_2 = \epsilon/16 < 1$, using Chernoff bound and \cref{lemma:chernoff_for_good_l1}, \begin{align*}
    &\Pr\left[\left|\sum_{t\in [T], z\in \Good} P_{t,z} - \E\left[\sum_{t\in [T], z\in \Good} P_{t,z}\right]\right| \geq \epsilon_2 \E\left[\sum_{t\in [T], z\in \Good} P_{t,z}\right]\right] \\
    \leq &2\exp\left(-\frac{\epsilon_2^2 T \E[\sum_{z\in [d]} P_{t,z}]}{3(1+\frac{8}{\epsilon\delta'}) \|c_i - c_j\|_1}\right) \\
    = & 2\exp\left(-\frac{\epsilon_2^2 T }{3(1+\frac{8}{\epsilon\delta'})} \frac{\dist_{\Good}(c_i, q)}{\dist_{\Good}(c_i, c_j)} \frac{\dist_{\Good}(c_i, c_j)}{\|c_i - c_j\|_1}\right) \\ 
    \leq & 2\exp\left(-\frac{\epsilon_2^2 T }{3(1+\frac{8}{\epsilon\delta'})} \cdot \frac12 \cdot \frac{13}{16}\right) \leq \frac{\delta'}{2n},
    \end{align*}
    by choosing 
    $T = \Theta(\frac{\log(n/\delta)}{\epsilon^3 \delta^2} )$. 
    
    %This proves claim 2 of the lemma.
   % \nithin{The above sentence looks strange.}

    \smallskip 
    
    \noindent \underline{Case 1.} If $\E[\sum_{z\in \Good} Q_{t,z}] \geq \frac18 \cdot \dist_{\Good}(c_i, c_j)$, using Chernoff bound and \cref{lemma:chernoff_for_good_l1}, we have 
    
    $$\Pr\left[\left|\sum_{t\in [T], z\in \Good} Q_{t,z} - \E\left[\sum_{t\in [T], z\in \Good} Q_{t,z}\right]\right| \geq \epsilon_2 \E\left[\sum_{t\in [T], z\in \Good} Q_{t,z}\right]\right] \leq \frac{\delta'}{2n},$$
    by choosing 
    $T = \Theta(\frac{\log(n/\delta)}{\epsilon^3 \delta^2} ) $. By union bound, with probability $1-\delta'/n$, we have 
    \begin{equation*}
        \sum_{t\in [T],z\in \Good} P_{t,z}  \approx_{\frac{\epsilon}{16}} \E\left[\sum_{t\in [T],z\in \Good} P_{t,z}\right], \text{ and } \sum_{t\in [T],z\in \Good} Q_{t,z} \approx_{\frac{\epsilon}{16}} \E\left[\sum_{t\in [T],z\in \Good} Q_{t,z}\right].
    \end{equation*}
    Therefore the second item of the lemma is proved.
\smallskip

    \noindent \underline{Case 2.} If $\E[\sum_{z\in \Good} Q_{t,z}] < \frac18 \dist_{\Good}(c_i, c_j)$. Then we have $\E[\sum_{z\in \Good} P_{t,z}] > (1-\frac18) \dist_{\Good}(c_i, c_j)$, which gives that $\frac{\dist_{\Good} (c_i, q)}{\dist_{\Good} (c_j, q)} > 7$.
    
    % by choosing $T = \Theta( \frac{2^p}{p^2\epsilon^{p+2}}\log(1/\delta)) \geq \Theta(\frac{1}{s}(1+\epsilon)^p \log(1/\delta))$.

    Using Chernoff bound and \cref{lemma:chernoff_for_good_l1}, we get
    \begin{align*}
    & \Pr\left[\left|\sum_{t\in [T],z\in \Good} Q_{t,z} - \E[\sum_{t\in [T], z\in \Good} Q_{t,z}]\right| > \frac{T}8 \cdot \dist_{\Good}(c_i, c_j) \right] \\
    \leq & 2\exp\left( -\frac{T \cdot \frac18\dist_{\Good}(c_i, c_j)}{3(1+\frac{8}{\epsilon\delta'}) \|c_i - c_j\|_1}\right) \\
    \leq & 2\exp\left( -\frac{T  }{24(1+\frac{8}{\epsilon\delta'})} \frac{\dist_{\Good}(c_i, c_j)}{\|c_i - c_j\|_1} \right)\leq \frac{\delta'}{2n}. 
    \end{align*}
    by choosing $T = \Omega(\frac{\log(n/\delta)}{\epsilon \delta} ) $.

    By union bound, with probability $1-\delta'/n$, we have 
    \begin{equation*}
        \sum_{t\in [T],z\in \Good} P_{t,z} \approx_{\frac{\epsilon}{16}} \E\left[\sum_{t\in [T], z\in \Good} P_{t,z}\right], \text{ and } \sum_{t\in [T],z\in \Good} Q_{t,z} \leq \frac{T}{4} \dist_{\Good}(c_i, c_j).
    \end{equation*}
    Thus 
    \begin{equation*}
        \frac{\sum_{t\in [T],z\in \Good} P_{t,z}}{\sum_{t\in [T],z\in \Good} Q_{t,z}} \geq \frac{(1-\frac{\epsilon}{16})(1-\frac18)\dist_{\Good}(c_i, c_j)}{\frac14 \dist_{\Good}(c_i, c_j)} \geq 3.
    \end{equation*}
    The third item of the lemma is proved.
\end{proof}

\begin{proof}[Proof of \cref{lemma:robust-feature-selection-l1}]
    We know that $\|r_{i^*} - u\|_1 \leq \frac{T}{\delta'} \|c_{i^*} - q\|_1$ holds with probability $1-\delta'$ by \Cref{claim:Markovlone}.
    Thus, it is sufficient to show the statement of the lemma holds under this assumption with probability $1-\delta'$.

    Fix $i\in [n]$. Since we only have 2 points from $C$ that are $c_i$ and $c_{i^*}$, we use $\Good$ for $\Good(c_i, c_{i^*}, q)$.

    Denote $b_1 = \dist_\Bad(c_i, q)$, $b_2 = \dist_\Bad(c_{i^*}, q)$, $g_1 = \dist_\Good(c_i, q)$, $g_2 = \dist_\Good(c_{i^*}, q)$. Also denote $b_1' = \dist_\Bad(r_i, u)/T$, $b_2' = \dist_\Bad(r_{i^*}, u)/T$, $g_1' = \dist_\Good(r_i, u)/T$, $g_2' = \dist_\Good(r_{i^*}, u)/T$.

    For a fixed $i \in [n]$, we have $\frac{\|c_i - q\|_1}{\|c_{i^*} - q\|_1} \geq 1+\epsilon$, which is $\frac{g_1 + b_1}{g_2 + b_2} \geq 1+\epsilon$. We are given that $\|r_{i^*} - u\|_1 \leq T/\delta' \|c_{i^*} - q\|_1$, which is $g_2' + b_2' \leq \frac{1}{\delta'}(g_2 + b_2)$. From the definition of bad region, we know that $1 - \frac{\epsilon\delta'}{8} \leq \frac{b_1}{b_2} \leq 1 + \frac{\epsilon\delta'}{8}$. Also $1 - \frac{\epsilon\delta'}{8} \leq \frac{b_1'}{b_2'} \leq 1 + \frac{\epsilon\delta'}{8}$, because we select and rescale coordinates, which does not affect the ratio.

    Applying \cref{lemma:chernoff_for_good_l1} on $c_i$ and $c_{i^*}$ gives that with probability $1-\delta'/n$,
    \begin{itemize}
        \item either $g_1' \approx_{\frac{\epsilon}{16}} g_1$, $g_2' \approx_{\frac{\epsilon}{16}} g_2$,
        \item or $\frac{g_1}{g_2} > 7$, $g_1' \approx_{\frac{\epsilon}{16}} g_1$, $\frac{g_1'}{g_2'} \geq 3$.
    \end{itemize}
Given these conditions, the following \cref{claim:robust-estimate} gives a lower bound on the ratio $\frac{g_1' + b_1'}{g_2' + b_2'}$. We defer the proof of the claim to the appendix.

    \begin{claim}\label{claim:robust-estimate}
    Suppose $b_1, b_2, g_1, g_2, b_1', b_2', g_1', g_2' \geq 0$. We are given the following conditions:
    \begin{itemize}
        \item $\frac{g_1+b_1}{g_2+b_2} \geq 1+\epsilon$,
        \item $1-\frac{\epsilon\delta'}{8} \leq \frac{b_1}{b_2} \leq 1+\frac{\epsilon\delta'}{8}$, $1-\frac{\epsilon\delta'}{8} \leq \frac{b_1'}{b_2'} \leq 1+\frac{\epsilon\delta'}{8}$,
        \item $g_2' + b_2' \leq \frac{1}{\delta'}(g_2+b_2)$,
        \item $g_1'$ approximates $g_1$, and $g_2'$ approximates $g_2$, i.e.
        \begin{itemize}
            \item either $g_1' \approx_{\frac{\epsilon}{16}} g_1$, $g_2' \approx_{\frac{\epsilon}{16}} g_2$,
            \item or $\frac{g_1}{g_2} > 7$, $g_1' \approx_{\frac{\epsilon}{16}} g_1$, $\frac{g_1'}{g_2'}\geq 3$.
        \end{itemize}
    \end{itemize}  
    then we have $\frac{g_1'+b_1'}{g_2'+b_2'} \geq 1+\epsilon\delta'/4.$
\end{claim}

 Applying \cref{claim:robust-estimate}, we get that $\frac{g_1' + b_1'}{g_2' + b_2'} \geq 1+\epsilon\delta'/4$, which implies $\frac{\|r_i - u\|_1}{\|r_{i^*} - u\|_1} \geq 1+\epsilon\delta'/4$.
    By union bound, this inequality holds for all $i \in [n]$ with probability $1-\delta'$. 
\end{proof}

\begin{claim} \label{claim:approximate-diff-l1}
    Suppose $g_1, w, g_1', w' \geq 0$, and $\frac{g_1}{w} \geq 1+\frac{\epsilon}{2}$. $g_1'$ and $w'$ are estimates of $g_1$ and $w$, where $g_1' \approx_{\epsilon/16} g_1$ and $w' \approx_{\epsilon/16} w$. Then $g_1' - w' \geq \left(\frac{3}{4} - \frac{\epsilon}{16}\right) (g_1 - w)$. 
\end{claim}

\begin{proof}
    From $g_1' \approx_{\epsilon/16} g_1$ and $w' \approx_{\epsilon/16} w$, we have 
    $
    g_1' - w' \geq \left(1-\frac{\epsilon}{16}\right) g_1 - \left(1+\frac{\epsilon}{16}\right) w.
    $
    Let $\rho = \frac{3}{4} - \frac{\epsilon}{16}$. We rewrite $\frac{g_1}{w} \geq 1+\frac{\epsilon}{2}$ as $\left(1-\frac{\epsilon}{16} - \rho\right) g_1 \geq \left(1+\frac{\epsilon}{16} - \rho\right) w$. This shows $\left(1-\frac{\epsilon}{16}\right) g_1 - \left(1+\frac{\epsilon}{16}\right) w \geq \rho (g_1 - w)$. 
    Thus we have $g_1' - w' \geq \rho (g_1 - w)$.
\end{proof}

\begin{proof}[Proof of \Cref{claim:robust-estimate}]

\noindent \underline{Case 1.} If $g_1' \approx_{\frac{\epsilon}{16}} g_1$, $g_2' \approx_{\frac{\epsilon}{16}} g_2$.

Given \( \frac{g_1 + b_1}{g_2 + b_2} \geq 1 + \epsilon \), we start by deducting \( 1 + \frac{\epsilon \delta'}{8} \) on both sides:
$
\frac{g_1 + b_1 - (1 + \frac{\epsilon \delta'}{8})(g_2 + b_2)}{g_2 + b_2} \geq 1 + \epsilon - (1 + \frac{\epsilon \delta'}{8}) \geq \frac{7}{8} \epsilon.
$
Since \( \frac{b_1}{b_2} \leq 1 + \frac{\epsilon \delta'}{8} \), we get
$
\frac{g_1 - (1 + \frac{\epsilon \delta'}{8})g_2}{g_2 + b_2} \geq \frac{7}{8} \epsilon.
$
Using \cref{claim:approximate-diff-l1} and letting \( w = (1 + \frac{\epsilon}{8})g_2 \), we have 
$
\frac{g_1' - (1 + \frac{\epsilon \delta'}{8})g_2'}{g_2 + b_2} \geq \frac{21 - \frac{7 \epsilon}{4}}{32} \epsilon \geq \frac{5 \epsilon}{8}.
$
Since \( g_2' + b_2' \leq \frac{1}{\delta'}(g_2 + b_2) \), we get
$
\frac{g_1' - (1 + \frac{\epsilon \delta'}{8})g_2'}{g_2' + b_2'} \geq \frac{5 \epsilon \delta'}{8}.
$
Rewriting the left-hand side as
$
\frac{g_1' - g_2' + b_1' - b_2' - \frac{\epsilon \delta'}{8}(g_2' + b_2') - b_1' + (1 + \frac{\epsilon \delta'}{8})b_2'}{g_2' + b_2'} \geq \frac{5 \epsilon \delta'}{8}.
$
Rearranging terms gives
$
\frac{g_1' - g_2' + b_1' - b_2' - b_1' + (1 + \frac{\epsilon \delta'}{8})b_2'}{g_2' + b_2'} \geq \frac{5 \epsilon \delta'}{8} + \frac{\epsilon \delta'}{8}.
$
Using \( \frac{b_1'}{b_2'} \geq 1 - \frac{\epsilon \delta'}{8} \), we get
$
\frac{g_1' - g_2' + b_1' - b_2' + (\frac{\epsilon \delta'}{4})b_2'}{g_2' + b_2'} = \frac{g_1' - g_2' + b_1' - b_2' - (1 - \frac{\epsilon \delta'}{8})b_2' + (1 + \frac{\epsilon \delta'}{8})b_2'}{g_2' + b_2'} \geq \frac{3 \epsilon \delta'}{4}.
$

This implies
$
\frac{g_1' - g_2' + b_1' - b_2'}{g_2' + b_2'} \geq \frac{3 \epsilon \delta'}{4} - \frac{\epsilon \delta'}{4} = \frac{\epsilon \delta'}{2}.
$
Finally, we obtain
$
\frac{g_1' + b_1'}{g_2' + b_2'} \geq 1 + \frac{\epsilon \delta'}{2}.
$

    \noindent \underline{Case 2.} If $\frac{g_1}{g_2} > 7$, $g_1' \approx_{\frac{\epsilon}{16}} g_1$, $\frac{g_1'}{g_2'}\geq 3$.
    
Given \( \frac{g_1 + b_1}{g_2 + b_2} \geq 1 + \epsilon \), we start by deducting \( 1 + \frac{\epsilon \delta'}{8} \) on both sides:
$
\frac{g_1 + b_1 - (1 + \frac{\epsilon \delta'}{8})(g_2 + b_2)}{g_2 + b_2} \geq 1 + \epsilon - (1 + \frac{\epsilon \delta'}{8}) \geq \frac{7}{8} \epsilon,
$
which simplifies to
$
\frac{g_1}{g_2 + b_2} \geq \frac{g_1 - (1 + \frac{\epsilon \delta'}{8})g_2}{g_2 + b_2} \geq \frac{7}{8} \epsilon
$
since \( \frac{b_1}{b_2} \leq 1 + \frac{\epsilon \delta'}{8} \) and \( g_2 \geq 0 \).
Using \( g_1' \approx_{\epsilon / 16} g_1 \), we get
$
\frac{g_1'}{g_2 + b_2} \geq \frac{7 \epsilon}{8 (1 + \epsilon / 16)}.
$
Since \( g_2' + b_2' \leq \frac{1}{\delta'}(g_2 + b_2) \), we have
$
\frac{g_1'}{g_2' + b_2'} \geq \frac{7 \epsilon \delta'}{8 (1 + \epsilon / 16)}.
$
Using \( \frac{g_1'}{g_2'} \geq 3 \Rightarrow \frac{3}{2}(g_1' - g_2') \geq g_1' \), we obtain
$
\frac{\frac{3}{2}(g_1' - g_2')}{g_2' + b_2'} \geq \frac{7 \epsilon \delta'}{12 (1 + \epsilon / 16)},
$
which can be rewritten as
$
\frac{g_1' - g_2'}{g_2' + b_2'} \geq \frac{7 \epsilon \delta'}{12 (1 + \epsilon / 16)}.
$
Since \( (1 - \frac{\epsilon \delta'}{8})g_2' \leq g_2' \), we get
$
\frac{g_1' - (1 + \frac{\epsilon \delta'}{8})g_2'}{g_2' + b_2'} \geq \frac{7 \epsilon \delta'}{12 (1 + \epsilon / 16)}.
$
Subtracting \( \frac{\epsilon \delta'}{8} \) from both sides:
$
\frac{g_1' - (1 + \frac{\epsilon \delta'}{8})g_2' - \frac{\epsilon \delta'}{8}(g_2' + b_2')}{g_2' + b_2'} \geq \frac{7 \epsilon \delta'}{12 (1 + \epsilon / 16)} - \frac{\epsilon \delta'}{8}
$
which simplifies to
$
\frac{g_1' - g_2' - \frac{\epsilon \delta'}{8} b_2'}{g_2' + b_2'} \geq \frac{7 \epsilon \delta'}{12 (1 + \epsilon / 16)} - \frac{\epsilon \delta'}{8}.
$
Rewriting gives
$
\frac{g_1' - g_2' + (1 - \frac{\epsilon \delta'}{8}) b_2' - b_2'}{g_2' + b_2'} \geq \frac{7 \epsilon \delta'}{12 (1 + \epsilon / 16)} - \frac{\epsilon \delta'}{8}.
$
Since \( \frac{b_1'}{b_2'} \geq 1 - \frac{\epsilon \delta'}{8} \), we get
$
\frac{g_1' - g_2' + b_1' - b_2'}{g_2' + b_2'} \geq \frac{7 \epsilon}{48 (1 + \epsilon / 16)} - \frac{\epsilon}{8}.
$
Finally, we obtain
$
\frac{g_1' + b_1'}{g_2' + b_2'} \geq 1 + \frac{7 \epsilon \delta'}{12 (1 + \epsilon / 16)} - \frac{\epsilon \delta'}{8} \geq 1 + \frac{\epsilon \delta'}{4}.
$

\end{proof}

% This proof shows intuition why our experiments show that, for a random $C$, $\sum p_b$ is much smaller than $n$. Feature selection (leverage score sampling) always sample $n$ coordinates, and we sample less and the gap is due to the Cauchy Schwartz Ineq.

% \subsection{Full proof for Claim \ref{claim:sample works}}

\subsection{Data Structure for \texorpdfstring{$\ell_2$}{l2}-metric with Near-linear Space and Query Time}\label{sec:linear-space-l2}

The preprocessing and querying procedures of our data structure are described in \cref{alg:sample-reuse-l2}. At a high level, the preprocessing and querying phases of \cref{alg:sample-reuse-l2} are very similar to that of \cref{alg:sample-reuse-l1}. The main differences are (1) in the choice of the probability distribution used to sample coordinates and (2) in the specific dimension reduction technique. In this case, we use the Johnson Lindenstrauss mapping for the dimension reduction.

\begin{algorithm}[htp]\label{alg:sample-reuse-l2}
    \caption{An approximate nearest neighbor data structure for $n$ points in $\R^d$ for $\ell_2$ metric.}
    \SetKwProg{preprocessing}{Preprocessing}{}{}
    \SetKwProg{query}{Query}{}{}
    \SetKwProg{comparator}{$E$}{}{}
    \SetKw{store}{store}
    \preprocessing{$(C, \epsilon, \delta)$\tcp*[f]{Inputs: $C=\{c_1, \dots, c_n\}\subseteq \R^d, \epsilon \in (0, \frac{1}{4}), \delta\in (0,1)$}}{
        % Let $\epsilon \leftarrow \epsilon'/\log \log n$ and $\delta \leftarrow \delta' / \Theta(n \log \log n)$\;
        Let $I \leftarrow \emptyset$ be a multiset and
        $T \leftarrow O( \log (n/\delta)/(\epsilon^4\delta^2))$\;
        % \footnote{According to Theorem \ref{theorem:Lp-NN}, it should be $O(\log(1/\delta)/\epsilon^3)$, but we can achieve a better bound of $O(\log(1/\delta)/\epsilon^2)$  by using $0\leq X_t, Y_t\leq 1$, instead of using $0\leq X_t, Y_t\leq 1/\epsilon$ in the proof. }
        
        \For{$t\in [T]$}{
            \For{$b\in [d]$}{
                Add $b$ to $I$ with probability  $p^{(b)}\triangleq \max_{(i,j)\in {\binom{n}{2}}} \frac{|c_i^{(b)} - c_j^{(b)}|^2}{\|c_i- c_j\|_2^2}$\;
            }
        }
        Let $R=\{r_1, r_2, \dots, r_n\}\subseteq \R^{I}$, where for $i \in [n], b \in I$, we have $r_i^{(b)} = c_i^{(b)} / \sqrt{p^{(b)}}$\;

        Let $m = O(\log(n/\delta)/(\epsilon^2\delta^2))$\;
        Sample $M\in \mathbb{R}^{[m]\times I}$, where each entry of $M$ is chosen uniformly at random from $\left\{\pm \sqrt{\frac1m}\right\}$ \;
        $M(\cdot)\colon \mathbb{R}^{I} \rightarrow \mathbb{R}^{m}$ is an oblivious, linear Johnson-Lindenstrauss mapping\;

        \store{$I$, $M(R) = \{Mr_i\mid i\in [n]\}$, $\{p^{(b)}\mid b\in I\}$, $M$ }
    }
    \query{$(q)$ \tcp*[f]{Inputs: $q\in \R^d$}}{
        Query $q^{(b)}, b\in I$\;
        Let $u\in \mathbb{R}^{I}$, where $\forall b \in I, u^{(b)} = q^{(b)} / \sqrt{p^{(b)}}$\;
        Let $\hat{i} = \argmin_{i\in [n]} \|Mr_i - Mu\|_2$\;
        \Return{$\hat{i}$}
    }
\end{algorithm}
%\joachim{... each entry ... is chosen from ... uniformly at random}
%\joachim{use $\backslash$mid inside set declarations}

\begin{theorem}\label{theorem:sample-reuse-l2}
    Consider a set $C$ of $n$ points $c_1, \dots, c_n \in \R^d$ equipped with $\ell_2$ metric. Given $0 < \epsilon < 1/4$ and $0 < \delta < 1$, with probability $(1-\delta)$, we can efficiently construct a randomized data structure (See \cref{alg:sample-reuse-l2}) such that for any query point $q\in \R^d$ the following conditions hold with probability $1-\delta$. (i) The data structure reads $O(n \log (n/\delta)/(\epsilon^4\delta^2))$ coordinates $q$. (ii) It returns $i\in [d]$ where $c_i$ is a $(1+\epsilon)$-approximate nearest neighbor of $q$ in $C$. (iii) The data structure has a size of $O(n \log^2(n/\delta) \log(d)/(\epsilon^6\delta^4))$ words. 
\end{theorem}

We show the proof in the following subsections in line with the proof structure described in \cref{sec:sample-reuse}. 
In particular, we show that each of the following good events, which we restate here for the sake of completeness, hold with probability at least $1-\delta'$, where $\delta' = \delta/4$.
\begin{enumerate}
    \item The first event is that the distance between a query and its nearest neighbor is overestimated at most by the factor $4/\delta$.
    \item The second event is that for a query point the estimated distance to its nearest neighbor is significantly smaller than the estimated distance to any center that is not an approximate nearest neighbor.
    \item The third event is that the dimension reduction preserves the distances between the coordinate-selected points and the query.
    \item The fourth event is that at most near-linear many coordinates are sampled in the preprocessing and accessed in the query phase. 
\end{enumerate}
Although the proof at a high level follows the same structure as for the $\ell_1$ case, there are significant differences between the proofs themselves. The main differences are in the dimension reduction as well as in bounding the space and query complexities.

    \subsubsection{Upper Bound for Estimate of Distance between Query and its Nearest Neighbor}\label{sec:first-section-l2}

    We first show that the squared distance between a query point and its nearest neighbor is overestimated at most by a factor $1/\delta'$ with probability $1-\delta'$. 

    \begin{claim}\label{claim:Markovl2}
    Let $q\in \R^n$ be a query and let $i^*\in [n]$ be the index of its nearest neighbor. It holds that
    \begin{align*}
        \frac{1}{T}\|r_{i^*} - u\|_2^2 \leq \frac{1}{\delta'} \|c_{i^*} - q\|_2^2,
    \end{align*}
    with probability $1-\delta'$.
    \end{claim}
    \begin{proof}
        For any $i\in [n]$, we have $\E[\|r_i - u\|_2^2] = T \|c_i - q\|_2^2$, because 
    \begin{align*}
        \E[\|r_i - u\|_2^2] =& \E\left[\sum_{b\in I} \left(r_i^{(b)} - u^{(b)}\right)^2\right] \\
        =& \E\left[\sum_{t\in [T], b\in [d]} \mathbbm{1}_{\{\text{in iteration $t$, we add $b$ to $I$\}}}\left(r_i^{(b)} - u^{(b)}\right)^2\right] \\
        =& \sum_{t\in [T], b\in [d]} p^{(b)}\left(r_i^{(b)} - u^{(b)}\right)^2 = \sum_{t\in [T], b\in [d]} p^{(b)}\left(\frac{c_i^{(b)}}{\sqrt{p^{(b)}}} - \frac{q^{(b)}}{\sqrt{p^{(b)}}}\right)^2 \\
        =& \sum_{t\in [T], b\in [d]} \left(c_i^{(b)} - q^{(b)}\right)^2 = T \|c_i - q\|_2^2.
    \end{align*}
    Using Markov's inequality, we have with probability $1-\delta'$,
    \begin{equation*}
    \|r_{i^*} - u\|_2^2 \leq \frac{T}{\delta'} \|c_{i^*} - q\|_2^2.\qedhere
    \end{equation*}
    \end{proof}

    \subsubsection{Coordinate Selection Preserves Distance Ratio}

    In this section, we show that, the following holds: 
for any query point $q\in \R^n$, whose nearest neighbor is $c_{i^*}$, with probability $1-2\delta'$, for all $c_i\in C$ such that $\frac{\|c_i-q\|_2}{\|c_{i^*}-q\|_2} \geq 1+\epsilon$,  the estimated ratio is significantly  larger than one, that is $\frac{\|r_i-u\|_2}{\|r_{i^*}-u\|_2} \geq 1+\epsilon\delta'/12$.
This statement is formalized in \Cref{lemma:robust-feature-selection-l2}.

As before, we first partition the coordinates in $[d]$ into good and bad regions for each pair of points. Although the specific definition of the regions are different compared to the $\ell_1$ case, the intuition behind the definitions is the same.

%In the good region are the coordinates on which the query is not too far away from one of the points. 
%The remaining coordinates are in the bad region.
%The intuition is that, restricted to the coordinates in the good region, the distance between one of the points and the query  is bounded.
%This implies that we can estimate these distances using a Chernoff bound.
%For the bad coordinates the intuition is that the distances from the query to the 2 points are so large that the ratio between these distances is close to one. Additionally these distances are not overestimated because of the previous event in \cref{sec:first-section-l2}.

\begin{definition}[Good region and bad region]
    Consider two points $c_1, c_2\in \R^d$, and a query point $q\in\R^d$. We define the good region $\Good(c_1, c_2, q)$ w.r.t.\ $(c_1,c_2,q)$ as $$\{z\in [d] \mid\min\{c_1^{(z)}, c_2^{(z)}\} - \frac{24}{\epsilon \delta'} |c_1^{(z)} - c_2^{(z)}| \leq q^{(z)} \leq \max\{c_1^{(z)}, c_2^{(z)}\} + \frac{24}{\epsilon \delta'} |c_1^{(z)} - c_2^{(z)}|\}.$$ We define the bad region w.r.t.\ $(c_1,c_2,q)$ as $\Bad(c_1, c_2, q) = [d]\setminus \Good(c_1, c_2, q)$. When there is no ambiguity, we use $\Good$ for $\Good(c_1, c_2, q)$, and $\Bad$ for $\Bad(c_1, c_2, q)$. 
\end{definition}
We also define the distance between points restricted to subsets of coordinates as before.
\begin{definition}[Restricted distance]
    Given $x,y \in \R^d$, we define the distance between $x$ and $y$ restricted to region $J\subseteq [d]$ as 
    $\dist_J(x,y) = \left(\sum_{z\in J} \left(x^{(z)} - y^{(z)}\right)^2\right)^{\frac12}$. 
    Furthermore, for $u,v \in \R^I$, where $I$ is a multiset, whose elements are from $[d]$, $\dist_J(u,v) = \left(\sum_{z\in I} \mathbbm{1}_{z\in J}\cdot\left(x^{(z)} - y^{(z)}\right)^2\right)^{\frac12}.$
    
\end{definition}

%For bounding distance estimates on the good region, we can use Chernoff bound. For distance estimates on the bad region, we need Markov's inequality. 
%This motivates our name of the regions. 

The following lemma is the $\ell_2$ version of \cref{lemma:good-region-is-large-l1}.

\begin{lemma}\label{lemma:good-region-is-large-l2}
    Consider 3 points $a,b,q\in\R^d$. If $\frac{\|b-q\|_2}{\|a-q\|_2} \geq 1+\epsilon$, then $\frac{\dist_{\Good(a,b,q)}(b,a)}{\|b-a\|_2} \geq \frac{13}{16}$.
\end{lemma}
\begin{proof}
We first follow the same steps as in the proof of \cref{lemma:good-region-is-large-l1} and obtain:
 \[
        \dist_{\Bad(a,b,q)}(b,q) + \dist_{\Bad(a,b,q)}(a,q) \leq \left(1+\frac2\epsilon\right) \|b-a\|_2.
    \]
    Squaring both sides, we get:
    \[
        \dist^2_{\Bad(a,b,q)}(b,q) + \dist^2_{\Bad(a,b,q)}(a,q) \leq \left(1+\frac2\epsilon\right)^2 \|b-a\|_2^2.
    \]
    This implies that:
     \begin{align*}
         \frac{1152}{\epsilon^2\delta
    '^2} \left(\|b-a\|_2^2 - \dist^2_{\Good(a,b,q)}(b,a)\right) = \sum_{z\in \Bad(a,b,q)} 2\left(\frac{24}{\epsilon\delta'}\right)^2\left|b^{(z)} - a^{(z)}\right|^2  \\
    \leq \left(1+\frac2\epsilon\right)^2 \|b-a\|_2^2.
     \end{align*}

    Rearranging the terms and simplifying, we get:
    \[
    \frac{\dist^2_{\Good(a,b,q)}(b,a)}{\|b-a\|_2^2} \geq \frac{1143}{1152} \Rightarrow \frac{\dist_{\Good(a,b,q)}(b,a)}{\|b-a\|_2} \geq \frac{13}{16}.\qedhere
    \]
\end{proof}

The following lemma shows that for two points $c_i, c_j$ and a query $q$ such that $q$ is significantly closer to $c_j$, we can estimate the distances, restricted to the good region, between one of the points and the query.
\begin{lemma}\label{lemma:chernoff_for_good_l2}

    Let $c_i, c_j$ where $i,j\in [n]$ be such that for the query point $q \in \R^d$ it holds that $\|c_i - q\|_2 \geq (1+\epsilon)\|c_j - q\|_2$. Let $\Good$ be shorthand for $\Good(c_i, c_j, q)$ and $\Bad$ for $\Bad(c_i, c_j, q)$. The following statements are true:

    \begin{enumerate}
        \item $\dist_{\Good} (c_i, q) \geq (1+\epsilon) \dist_{\Good} (c_j, q)$.
        \item If $\dist^2_{\Good} (c_j, q) \geq \frac1{16} \dist^2_{\Good} (c_i, c_j)$, then with probability $1-\delta'/n$, 
        \begin{align*}
            \dist^2_{\Good} (r_i, u) \approx_{\frac{\epsilon}{16}} T\cdot \dist^2_{\Good}(c_i, q),\\
            \dist^2_{\Good} (r_j, u) \approx_{\frac{\epsilon}{16}} T\cdot \dist^2_{\Good}(c_j, q).
        \end{align*}
        \item If $\dist^2_{\Good} (c_j, q) < \frac1{16} \dist^2_{\Good} (c_i, c_j)$, then $\frac{\dist^2_{\Good} (c_i, q)}{\dist^2_{\Good} (c_j, q)} > 7$. With probability $1-\delta'/n$, 
        \begin{align*}
            \dist^2_{\Good} (r_i, u) &\approx_{\frac{\epsilon}{16}} T\cdot \dist^2_{\Good}(c_i, q),\\
            \frac{\dist^2_{\Good} (r_i, u)}{\dist^2_{\Good} (r_j, u)} &\geq 3.
        \end{align*}
    \end{enumerate}

%     \begin{enumerate}
%         \item $\dist_{\Good} (c_i, q) \geq (1+\epsilon) \dist_{\Good} (c_j, q)$.
%         \item With probability $1-\frac{\delta'}{2n}$ we have $\dist^2_{\Good} (r_i, u) \approx_{\frac{\epsilon}{16}} T\cdot \dist^2_{\Good}(c_i, q)$.
%         \item 
% \begin{enumerate}[(i)]        
%         \item If $\dist^2_{\Good} (c_j, q) \geq \frac1{16} \dist^2_{\Good} (c_i, c_j)$, then with probability $1-\frac{\delta'}{2n}$ we have             $\dist^2_{\Good} (r_j, u) \approx_{\frac{\epsilon}{16}} T\cdot \dist^2_{\Good}(c_j, q)$.

%         \item If $\dist^2_{\Good} (c_j, q) < \frac1{16} \dist^2_{\Good} (c_i, c_j)$,  which implies $\frac{\dist^2_{\Good} (c_i, q)}{\dist^2_{\Good} (c_j, q)} \geq 7$, then with probability $1-\frac{\delta'}{2n}$ we have
%             $\frac{\dist^2_{\Good} (r_i, u)}{\dist^2_{\Good} (r_j, u)} \geq 3$.
%     \end{enumerate}
% \end{enumerate}
    \end{lemma}
\begin{proof}
    Combining the following two properties, 
    \begin{equation*}
        \frac{ \dist^2_{\Bad} (c_i, q)}{\dist^2_{\Bad} (c_j, q)} \leq (1+\frac{\epsilon\delta'}{24})^2 < (1+\epsilon)^2,
    \end{equation*}
    \begin{equation*}
         \frac{\dist^2_{\Good} (c_i, q) + \dist^2_{\Bad} (c_i, q)}{\dist^2_{\Good} (c_j, q)+\dist^2_{\Bad} (c_j, q)} = \frac{\|c_i - q\|_2^2}{\|c_j - q\|_2^2} \geq (1+\epsilon)^2.
    \end{equation*}
    implies
    \begin{equation*}
        \frac{\dist^2_{\Good} (c_i, q)}{\dist^2_{\Good} (c_j, q)} \geq (1+\epsilon)^2.
    \end{equation*}
    Therefore the first item of the lemma is proved.
    
    In the algorithm, we construct $\{r_i\mid i\in [n]\}$ and $u$ through sampling a multiset $I$. To analyze how the data structure samples $I$ in the preprocessing, we define the following random variables $\forall t\in [T], z\in \Good$
    \[P_{t, z} =  \begin{cases} 
      \frac{|q^{(z)} - c_i^{(z)}|^2}{p^{(z)} }, & w.p.\; p^{(z)},\hspace{0.95cm}\text{i.e. if we add $z$ to $I$ at iteration $t$.}\\
      0, & w.p.\; 1-p^{(z)}, \quad \text{i.e. if we do not add $z$ to $I$ at iteration $t$.}
   \end{cases}\]
   \[Q_{t, z} =  \begin{cases} 
      \frac{|q^{(z)} - c_j^{(z)}|^2}{p^{(z)} } , & w.p.\; p^{(z)},\hspace{0.95cm}\text{i.e. if we add $z$ to $I$ at iteration $t$.}\\
      0, & w.p.\; 1-p^{(z)}, \quad\text{i.e. if we do not add $z$ to $I$ at iteration $t$.}
   \end{cases}\]
   Notice that $\sum_{t\in [T],z\in \Good} P_{t,z} = \dist^2_\Good(r_i, u)$ and $\sum_{t\in [T],z\in \Good} Q_{t,z} = \dist^2_\Good(r_j, u)$.  

    Observe that $\forall t\in [T], z\in \Good$, $0 \leq P_{t,z} \leq \left(1+\frac{24}{\epsilon\delta'}\right)^2 \|c_i - c_j\|_2^2$. This is because, if $c_i^{(z)} = c_j^{(z)}$, from the definition of good region we have $q^{(z)} = c_i^{(z)}$, thus $P_{t,z} = 0$; otherwise $c_i^{(z)} \neq c_j^{(z)}$, 
    \begin{align*}
        0 \leq P_{t,z} \leq \frac{|q^{(z)} - c_i^{(z)}|^2}{|c_i^{(z)} - c_j^{(z)}|^2 } \cdot \|c_i - c_j\|_2^2 \leq \left(1+\frac{24}{\epsilon\delta'}\right)^2 \|c_i - c_j\|_2^2,
    \end{align*}
    where the second inequality comes from ``oversampling'': $p^{(z)} \geq \frac{|c_i^{(z)} - c_j^{(z)}|^2}{\|c_i - c_j\|_2^2}$, and the third inequality comes from the definition of good region.
    By a similar argument, we have $0 \leq Q_{t,z}\leq (1+\frac{24}{\epsilon\delta'})^2 \|c_i - c_j\|_2^2$. 

    Notice $\forall t\in [T], \E[\sum_{z\in \Good} P_{t,z}] = \dist^2_{\Good}(c_i, q)$ and $\E[\sum_{z\in \Good} Q_{t,z}] = \dist^2_{\Good}(c_j, q)$. Because 
    \begin{align*}
        \E[\sum_{z\in \Good} P_{t,z}] =& \sum_{z\in \Good} \E[P_{t,z}] \\
        =& \sum_{z\in \Good} p^{(z)}\cdot \left(\frac{q^{(z)} - c_i^{(z)}}{\sqrt{p^{(z)}}}\right)^2 = \sum_{z\in \Good} |q^{(z)} - c_i^{(z)}|^2 = \dist^2_{\Good}(c_i, q).
    \end{align*}
    Since $\forall t\in [T], \E[\sum_{z\in \Good} (P_{t,z} + Q_{t,z})] = \dist^2_{\Good}(c_i, q) + \dist^2_{\Good}(c_j, q) \geq \frac{1}{2} \dist^2_{\Good}(c_i, c_j)$, we have $\E[\sum_{z\in \Good}P_{t,z}] \geq \frac{1}{4} \dist^2_{\Good}(c_i, c_j)$. In the rest of the proof, we use $\E[\sum_{z\in \Good} P_{t,z}]$ without specify which $t\in [T]$, because they have the same value for all $t\in [T]$. 

    Let $0<\epsilon_2 = \epsilon/16 < 1$. Using Chernoff bound and \cref{lemma:good-region-is-large-l2}, we have 
    \begin{align*}
    &\Pr\left[\left|\sum_{t\in [T], z\in \Good} P_{t,z} - \E\left[\sum_{t\in [T], z\in \Good} P_{t,z}\right]\right| \geq \epsilon_2 \E\left[\sum_{t\in [T], z\in \Good} P_{t,z}\right]\right] \\
    \leq & 2\exp\left(-\frac{\epsilon_2^2 T \E[\sum_{z\in \Good} P_{t,z}]}{3(1+\frac{24}{\epsilon\delta'})^2 \|c_i - c_j\|_2^2}\right) \\
    = & 2\exp\left(-\frac{\epsilon_2^2 T }{3(1+\frac{24}{\epsilon\delta'})^2} \frac{\dist^2_{\Good}(c_i, q)}{\dist^2_{\Good}(c_i, c_j)} \frac{\dist^2_{\Good}(c_i, c_j)}{\|c_i - c_j\|_2^2}\right) \\
    \leq & 2\exp\left(-\frac{\epsilon_2^2 T }{3(1+\frac{24}{\epsilon\delta'})^2} \cdot \frac{1}{4} \cdot \frac{1143}{1152}\right) \leq \frac{\delta'}{2n},
    \end{align*}
    by choosing 
    $T = \Theta(\frac{1}{\epsilon^4 \delta^2} \log(n/\delta)) $. 
    
    \noindent \underline{Case 1.} If $\E[\sum_{z\in \Good} Q_{t,z}] \geq \frac{1}{16} \cdot \dist^2_{\Good}(c_i, c_j)$, using Chernoff bound and \cref{lemma:good-region-is-large-l2}, we have $$\Pr\left[\left|\sum_{t\in [T], z\in \Good} Q_{t,z} - \E\left[\sum_{t\in [T], z\in \Good} Q_{t,z}\right]\right| \geq \epsilon_2 \E\left[\sum_{t\in [T], z\in \Good} Q_{t,z}\right]\right] \leq \frac{\delta'}{2n},$$
    by choosing 
    $T = \Theta(\frac{1}{\epsilon^4 \delta^2} \log(n/\delta)) $. 

    By union bound, with probability $1-\delta'/n$, we have 
    \begin{equation*}
        \sum_{t\in [T],z\in \Good} P_{t,z}  \approx_{\frac{\epsilon}{16}} \E[\sum_{t\in [T],z\in \Good} P_{t,z}], \text{ and } \sum_{t\in [T],z\in \Good} Q_{t,z} \approx_{\frac{\epsilon}{16}} \E[\sum_{t\in [T],z\in \Good} Q_{t,z}],
    \end{equation*}
    Therefore the second item of the lemma is proved.

    \noindent \underline{Case 2.} If $\E[\sum_{z\in \Good} Q_{t,z}] < \frac{1}{16} \cdot \dist^2_{\Good}(c_i, c_j)$. Then we have $\E[\sum_{z\in \Good} P_{t,z}] > (\frac12-\frac1{16}) \dist^2_{\Good}(c_i, c_j)$, which gives that $\frac{\dist^2_{\Good} (c_i, q)}{\dist^2_{\Good} (c_j, q)} > 7$. Using Chernoff bound and \cref{lemma:good-region-is-large-l2}, we have
    \begin{align*}
    &\Pr\left[\left|\sum_{t\in [T],z\in \Good} Q_{t,z} - \E[\sum_{t\in [T], z\in \Good} Q_{t,z}]\right| > \frac{1}{16}T \cdot \dist^2_{\Good}(c_i, c_j) \right] \\
    <& 2\exp\left( -\frac{T \cdot \frac{1}{16}\dist^2_{\Good}(c_i, c_j)}{3(1+\frac{24}{\epsilon\delta'})^2 \|c_i - c_j\|_2^2}\right) \\
    \leq & 2\exp\left( -\frac{T  }{48(1+\frac{24}{\epsilon\delta'})^2} \frac{\dist^2_{\Good}(c_i, c_j)}{\|c_i - c_j\|_2^2} \right)\leq \frac{\delta'}{2n}. 
    \end{align*}
    by choosing $T = \Omega(\frac{1}{\epsilon^2 \delta^2} \log(n/\delta)) $.

    By union bound, with probability $1-\delta'/n$, we have 
    \begin{equation*}
        \sum_{t\in [T],z\in \Good} P_{t,z} \approx_{\frac{\epsilon}{16}} \E\left[\sum_{t\in[T],z\in \Good} P_{t,z}\right], \text{ and } \sum_{t\in [T],z\in \Good} Q_{t,z} \leq \frac{T}8 \dist^2_{\Good}(c_i, c_j).
    \end{equation*}
    Thus 
    \begin{equation*}
        \frac{\sum_{t\in [T],z\in \Good} P_{t,z}}{\sum_{t\in [T],z\in \Good} Q_{t,z}} \geq \frac{(1-\frac{\epsilon}{16})(\frac12 - \frac18)\dist^2_{\Good}(c_i, c_j)}{\frac18 \dist^2_{\Good}(c_i, c_j)} \geq 3.
    \end{equation*}
    The third item of the lemma is proved.
\end{proof}

In the following lemma we state that our bounds on the distance estimates are good enough. Concretely, for a point that is not an approximate nearest neighbor, the estimated distance between this point and the query is significantly larger than the estimated distance between the nearest neighbor and the query.
\begin{lemma}\label{lemma:robust-feature-selection-l2}
    For any query point $q \in \R^d$, with probability $1-2\delta'$, it holds for all $i\in [n]$, if $\frac{\|c_i-q\|_1}{\|c_{i^*}-q\|_1} \geq 1+\epsilon$, then $\frac{\|r_i-u\|_2}{\|r_{i^*}-u\|_2} \geq 1+\frac{\epsilon\delta'}{12}$.
\end{lemma}

\begin{proof}
    We know that $\|r_{i^*} - u\|_2^2 \leq \frac{T}{\delta'} \|c_{i^*} - q\|_2^2$ holds with probability $1-\delta'$ by \Cref{claim:Markovl2}.
    Thus, it is sufficient to show the statement of the lemma holds under this assumption with probability $1-\delta'$.

    Fix $i\in [n]$. Since we only have 2 points from $C$ that are $c_i$ and $c_{i^*}$, for simplicity we use $\Good$ for $\Good(c_i, c_{i^*}, q)$ and $\Bad$ for $\Bad(c_i, c_{i^*}, q)$.

    Denote $b_1 = \dist^2_\Bad(c_i, q)$, $b_2 = \dist^2_\Bad(c_{i^*}, q)$, $g_1 = \dist^2_\Good(c_i, q)$, $g_2 = \dist^2_\Good(c_{i^*}, q)$. Also denote $b_1' = \dist^2_\Bad(r_i, u)/T$, $b_2' = \dist^2_\Bad(r_{i^*}, u)/T$, $g_1' = \dist^2_\Good(r_i, u)/T$, $g_2' = \dist^2_\Good(r_{i^*}, u)/T$.

    We have $\frac{\|c_i - q\|_2}{\|c_{i^*} - q\|_2} \geq 1+\epsilon$, which is $\frac{g_1 + b_1}{g_2 + b_2} \geq (1+\epsilon)^2 \geq 1+\epsilon$. We are given that $\|r_{i^*} - u\|_2^2 \leq T/\delta' \|c_{i^*} - q\|_2^2$, which is $g_2' + b_2' \leq \frac{1}{\delta'}(g_2 + b_2)$. From the definition of bad region, we know that $1 - \frac{\epsilon\delta'}{8} \leq (1 - \frac{\epsilon\delta'}{24})^2 \leq \frac{b_1}{b_2} \leq (1 + \frac{\epsilon\delta'}{24})^2 \leq 1 + \frac{\epsilon\delta'}{8}$. Also $1 - \frac{\epsilon\delta'}{8} \leq \frac{b_1'}{b_2'} \leq 1 + \frac{\epsilon\delta'}{8}$, because we select and rescale coordinates.

    Applying \cref{lemma:chernoff_for_good_l2} on $c_i$ and $c_{i^*}$ gives that with probability $1-\delta'/n$,
    \begin{itemize}
        \item either $g_1' \approx_{\frac{\epsilon}{16}} g_1$, $g_2' \approx_{\frac{\epsilon}{16}} g_2$,
        \item or $\frac{g_1}{g_2} > 7$, $g_1' \approx_{\frac{\epsilon}{16}} g_1$, $\frac{g_1'}{g_2'} \geq 3$.
    \end{itemize}
    
    Applying \cref{claim:robust-estimate}, we prove that $\frac{g_1' + b_1'}{g_2' + b_2'} \geq 1 + \frac{\epsilon\delta'}{4}$, that is $\frac{\|r_i - u\|^2_2}{\|r_{i^*} - u\|^2_2} \geq 1 + \frac{\epsilon\delta'}{4}$. Thus we have $\frac{\|r_i - u\|_2}{\|r_{i^*} - u\|_2} \geq 1 + \frac{\epsilon\delta'}{12}$.
\end{proof}

\subsubsection{\texorpdfstring{$\ell_2$}{l2} Dimension Reduction Preserves the Distance Ratio}

In this section, we bound the distortion introduced by the dimension reduction that our data structure performs. %This is a point of siginificant difference between the proof

\begin{theorem}[distributional Johnson-Lindenstrauss Lemma]\label{thm:distJL}
    Consider 2 points $a, b \in \R^d$ and parameters $0<\epsilon, \delta < 1$. We sample a random matrix $M\in \R^{m \times d}$, where $m = O(\log(1/\delta)/(\epsilon^2))$ and each entry of $M$ is sampled uniformly from $\{\pm \sqrt{1/m}\}$. Then with probability $1-\delta$, $\|Ma- Mb\|_2 \approx_\epsilon \|a-b\|_2$.
\end{theorem}

\begin{claim}
    Let $q\in \R^n$ be a query. Let $i^*\in [n]$ be the index of its nearest neighbor. \Cref{alg:sample-reuse-l2} outputs an index $i$ such that
    \begin{align*}
        \frac{\|c_i - q\|_2}{\|c_{i^*} - q\|_2} \leq 1+\epsilon
    \end{align*}
    with probability $1-3\delta'$.
\end{claim}
\begin{proof}
Applying \cref{thm:distJL} with $m = O(\log(n/\delta')/(\epsilon^2 \delta'^2))$, given any $v \in R$, it holds with probability $1-\delta'/(n+1)$, that $\|Mv- Mu\|_2 \in [1-\epsilon\delta'/100, 1+\epsilon\delta'/100] \|v- u\|_1$. By union bound, with probability $1-\delta'$, it holds for all $v\in R$.
By \Cref{lemma:robust-feature-selection-l2} it holds with probability $1-2\delta'$ that for all $i\in [n]$, if $\frac{\|c_i - q\|_2}{\|c_{i^*} - q\|_2} \geq 1+\epsilon$, then $\frac{\|r_i - u\|_1}{\|r_{i^*} - u\|_1} \geq 1+\epsilon\delta'/12$, thus $\frac{\|Mr_i - Mu\|_2}{\|Mr_{i^*}- Mu\|_2} > 1$ with probability $1-3\delta'$, i.e. the algorithm will not output $i$ as the index of the approximate nearest neighbor.
\end{proof}
\begin{remark}
    Using approximate nearest neighbour technqiues, e.g.\ \cite{indyk-wagner18:apx-nn-limited-space}, can further reduce the space at the cost of slower query, i.e.\ reducing the power of $\frac{\log(n)}{\epsilon\delta}$ in space complexity.
\end{remark}

\subsubsection{Space and Query Complexity}
Next, we argue the space and query complexity.
The main result in this section is that for any set $C$ of $n$ centers the sum of probabilities $\sum_{b\in [d]} p^{(b)}$ is upper bounded by $n$. The techniques used to obtain the bound here is linear algebraic in contrast to the more combinatorial argument given in the $\ell_1$ case. 

\begin{claim}
    \[
    1\leq \sum_{b \in [d]} p^{(b)} = \sum_{b \in [d]} \max_{i\neq j}\frac{|c_i^{(b)} - c_j^{(b)}|^2}{\|c_i - c_j\|_2^2} \leq n.
    \] 
\end{claim}

\begin{proof}
    The argument for the lower bound on $\sum_{b\in [d]} p^{(b)}$ is the following. Fix arbitrary $i_1, j_1 \in [n]$, such that $\|c_{i_1} - c_{j_1}\|_2 \neq 0$. $\sum_{b\in [d]} p^{(b)} \geq \sum_{b\in [d]} \frac{\left|c_{i_1}^{(b)} - c_{j_1}^{(b)}\right|^2}{\|c_{i_1} - c_{j_1}\|_2^2} = 1$.

    Next, we show the upper bound on $\sum_{b\in [d]} p^{(b)}$. 
Let $C \in \R^{n\times d}$ denote a matrix of the $n$ input points. Assume for simplicity that there is no repetition of points and  that $\texttt{rank}(C) = n$. Let the singular value decomposition (SVD) of $C$ be $U\Sigma V^\transpose$, where $U \in \R^{n \times n}, \Sigma \in \R^{n \times d}, V \in \R^{d \times d}$. We use $u_i$ (and $v_i$) to denote the $i$-th column of $U$ (and $V$), and $\sigma_i$ to denote the $i$-th singular value.
    
    We prove that for each $b\in [d]$, 
    \begin{equation}\label{eq:math-program-l2-b}
        p^{(b)} \leq \sum_{k\in[n] } (V_{b,k})^2.
    \end{equation}
    This would imply that $\sum_{b \in [d]} p^{(b)} \leq \sum_{b \in [d], k\in[n]} (V_{b,k})^2 = \sum_{i\in [n]} \|v_k\|_2^2 = n,$ which proves our claim. 

    To prove \eqref{eq:math-program-l2-b}, it suffices to prove that for each $b\in[d], i,j\in [n], i\neq j$:
    \begin{equation}\label{eq:math-program-l2-bij}
        \frac{|c_i^{(b)} - c_j^{(b)}|^2}{\|c_i - c_j\|_2^2} \leq \sum_{k\in[n] } (V_{b,k})^2.
    \end{equation}
    %Now we prove \eqref{eq:math-program-l2-bij}.

    We know from the SVD decomposition that 
    \[
    C = \sigma_1 u_1 v_1^\transpose + \sigma_2 u_2 v_2^\transpose + \cdots + \sigma_n u_n v_n^\transpose.
    \]
    Now, for each $i\in[n]$
    \[
    c_i = C_{i, \cdot} = \sigma_1 U_{i,1} v_1^\transpose + \sigma_2 U_{i,2} v_2^\transpose + \cdots + \sigma_n U_{i,n} v_n^\transpose.
    \]
    %Let's rewrite LHS from SVD perspective:
    The above implies that:
    \begin{align*}
        |c_i^{(b)} - c_j^{(b)}|^2 &= ((\sigma_1 U_{i,1} V_{1,b} + \cdots + \sigma_n U_{i,n} V_{n,b}) - (\sigma_1 U_{j,1} V_{1,b} + \cdots + \sigma_n U_{j,n} V_{n,b}))^2 \\
        &= (\sigma_1 V_{1,b} (U_{i,1} - U_{j,1}) + \cdots + \sigma_n V_{n,b} (U_{i,n} - U_{j,n}))^2
    \end{align*}
    \begin{align*}
        \|c_i - c_j\|_2^2 &= \|C_{i,\cdot} - C_{j, \cdot}\|_2^2 \\
        &= \|\sigma_1 (U_{i,1} - U_{j,1})v_1^\transpose + \sigma_2 (U_{i,2} - U_{j,2})v_2^\transpose + \cdots + \sigma_n (U_{i,n}-U_{j,n}) v_n^\transpose \|_2^2  \\
        &= (\sigma_1 (U_{i,1} - U_{j,1}))^2 + (\sigma_2 (U_{i,2} - U_{j,2}))^2 + \cdots + (\sigma_n (U_{i,n}-U_{j,n}))^2 
    \end{align*}
    The last equality follows from the fact that $V$ is orthonormal.
    
    Finally, using Cauchy-Schwarz inequality, 
        \begin{align*}
            \left((\sigma_1 (U_{i,1} - U_{j,1}))^2 + \cdots + (\sigma_n (U_{i,n}-U_{j,n}))^2\right) \left(\sum_{k\in[n] } (V_{b,k})^2\right) \\
            \geq (\sigma_1 V_{1,b} (U_{i,1} - U_{j,1}) + \cdots + \sigma_n V_{n,b} (U_{i,n} - U_{j,n}))^2
        \end{align*}

        which proves \eqref{eq:math-program-l2-bij}.
\end{proof}
The proof of the following claim is analogous to the proof of~\cref{clm:I-size-upperbound-l1} and is omitted for succinctness.
\begin{claim}\label{clm:I-size-upperbound-l2}
    With probability $1-\delta'$, $|I| \leq 2Tn$.    
\end{claim}
\begin{comment}
\begin{proof}
    Since $I$ is constructed through sampling, we define the following random variables $\forall t\in [T], z\in [d]$,
    \begin{equation*}
        X_{t, z} =  \begin{cases} 
      1, & w.p.\; p^{(z)},\quad\text{i.e. if we add $z$ to $I$ at iteration $t$.}\\
      0, & w.p.\; 1-p^{(z)}, \quad\text{i.e. if we don't add $z$ to $I$ at iteration $t$.}
      \end{cases}
    \end{equation*}
    We can see that $|I| = \sum_{t\in[T], z\in[d]} X_{t,z}$.
    Using Chernoff bound, we have 
    \begin{equation*}
        \Pr[\sum_{t\in[T], z\in[d]} X_{t,z} - \E[\sum_{t\in[T], z\in[d]} X_{t,z}] \geq \E[\sum_{t\in[T], z\in[d]} X_{t,z}]] \leq \exp (-\E[\sum_{t\in[T], z\in[d]} X_{t,z}]/3) \leq \exp(-T/3) \leq \delta_4.
    \end{equation*}
    Thus, with probability $1-\delta_4$, $|I| \leq 2 E[|I|] = 2Tn$.
\end{proof}
\end{comment}

\begin{lemma}
    With prob. $1-\delta'$, the space complexity of \cref{alg:sample-reuse-l2} is $O(\frac{n\log^2(n/\delta)\log(d)}{\epsilon^6 \delta^4})$ wordsize, and the number of coordinates that we query is $O(\frac{n\log(n/\delta)}{\epsilon^4 \delta^2})$. 
\end{lemma}
\begin{proof}
    By \cref{clm:I-size-upperbound-l2}, we have that the number of coordinates of $q$ that we query is $|I| = O(Tn) = O(\frac{n\log(n/\delta)}{\epsilon^4 \delta^2})$. Next we analyze the space complexity of the data structure. Matrix $M$ is of wordsize $O(m|I|)$, $M(R)$ is of wordsize $O(mn)$, $I$ is of wordsize $O(|I| \log(d))$. Thus the space complexity of the our data structure is $O(m|I| + mn +|I|\log(d)) = O(\frac{n\log^2(n/\delta)\log(d)}{\epsilon^6 \delta^4})$ wordsize.
\end{proof}

 Using union bound, the success probability of the data structure is at least $1-4\delta' = 1-\delta$.

\section{\texorpdfstring{$O(1)$}{O(1)}-Approximate Sublinear Representation of Clusterings}\label{sec:expected-apx}

In this section, we show a data structure with near-linear space and query complexity for approximate nearest neighbor under $\ell_1$ and $\ell_2$ metrics, whose approximation ratio is constant in expectation.

\begin{theorem}
    Consider a set $C$ of $n$ points $c_1, \dots, c_n \in [-\Delta, \Delta]^d$ equipped with $\ell_1$ or $\ell_2$ metric. We can efficiently construct a randomized data structure such that for any query point $q\in [-\Delta,\Delta]^d$ the following conditions hold. (i) The data structure reads $O(n \polylog(n,d,\Delta)/(\epsilon^3))$ coordinates of $q$. (ii) It returns $i\in [n]$ where $c_i$ is a $(8+\epsilon)$-approximate nearest neighbor of $q$ in $C$ in expectation. (iii) The data structure has a size of $O(n \polylog(n,d,\Delta)/(\epsilon^3))$ words.
\end{theorem}
\begin{definition}[Cost of clustering]
    Consider a set $C=\{c_1, \dots, c_k\}$ of $k$ centers, and a set $P$ of points in $\R^d$ equipped with $\ell_p$-metric. The cost of $C$ w.r.t. partition $\sigma:P\rightarrow [k]$ is $\mathsf{cost}_p(C,P,\sigma) = \sum_{q\in P} \|q-\sigma(q)\|_p^p$. The cost $\mathsf{cost}_p(C,P)=\sum_{q\in P} \min_{c_i\in C}\|q-c_i\|_p^p$ of $C$ is defined as the minimum cost of any partitioning. We call $\mathsf{cost}_1(C,P)$ the \emph{$k$-Median} cost and $\mathsf{cost}_2(C,P)$ the \emph{$k$-Means} cost of $C$, respectively.
\end{definition}

\begin{corollary}
    Consider a set $C$ of $k$ centers $c_1, \dots, c_k \in [-\Delta, \Delta]^d$ equipped with $\ell_p$ metric ($p=1,2$). We can efficiently construct a randomized data structure such that for any set $P \subset [-\Delta,\Delta]^d$ of points the following conditions hold. (i) The data structure reads $O(k \polylog(k,d,\Delta)/(\epsilon^3))$ coordinates of $P$. (ii) It returns a partition $\sigma:P\rightarrow [k]$ (iii) $\mathsf{cost}_p(C,P,\sigma) \leq O(1) \mathsf{cost}_p(C, P)$.
    (iv) The data structure has a size of $O(k \polylog(k,d,\Delta)/(\epsilon^3))$ words.
\end{corollary}

Here we only show the analysis for $\ell_1$ metric, and the analysis for $\ell_2$ metric is similar.

\Cref{alg:expected-l1} describes the data structure that achieves the above guarantees. It is very similar to the data structure in \Cref{sec:l1samplereuse}. We obtain \Cref{alg:expected-l1} from \Cref{sec:l1samplereuse} essentially by restricting the points to $[-\Delta, \Delta]^d$. To be self-contained, we restate the whole algorithm.

\begin{algorithm}[htp]\label{alg:expected-l1}
    \caption{An approximate nearest neighbor data structure for $n$ points in $[-\Delta,\Delta]^d$ for $\ell_1$ metric.}
    \SetKwProg{preprocessing}{Preprocessing}{}{}
    \SetKwProg{query}{Query}{}{}
    \SetKwProg{comparator}{$E$}{}{}
    \SetKw{store}{store}
    \preprocessing{$(\Delta, C, \epsilon)$\tcp*[f]{Inputs: $\Delta \in \mathbb{N}, C=\{c_1, \dots, c_n\}\subseteq [-\Delta,\Delta]^d, \epsilon \in (0, \frac{1}{4})$}}{
        % Let $\epsilon \leftarrow \epsilon'/\log \log n$ and $\delta \leftarrow \delta' / \Theta(n \log \log n)$\;
        Let $I \leftarrow \emptyset$ be a multiset and
        $T \leftarrow O( \log (n)\log(\Delta d)/(\epsilon^3))$\;
        % \footnote{According to Theorem \ref{theorem:Lp-NN}, it should be $O(\log(1/\delta)/\epsilon^3)$, but we can achieve a better bound of $O(\log(1/\delta)/\epsilon^2)$  by using $0\leq X_t, Y_t\leq 1$, instead of using $0\leq X_t, Y_t\leq 1/\epsilon$ in the proof. }
        
        \For{$t\in [T]$}{
            \For{$b\in [d]$}{
                Add $b$ to $I$ with probability  $p^{(b)}\triangleq \max_{(i,j)\in {\binom{n}{2}}} \frac{\left|c_i^{(b)} - c_j^{(b)}\right|}{\|c_i- c_j\|_1}$\;
            }
        }
        Let $R=\{r_1, r_2, \dots, r_n\}\subseteq \R^{I}$, where for $i \in [n], b \in I$, we have $r_i^{(b)} = c_i^{(b)} / p^{(b)}$\;
        % \tcp{Apply the $l_1$ dimension reduction $(M, F)$ by \cite{indyk06:l1-dimension-reduction} to $R$. It guarantees that $\forall (i,j)\in {n \choose 2}, F(Mr_i, Mr_j) \in [1-\epsilon\delta/100, 1+\epsilon\delta/100] \|r_i- r_j\|_1$ with probability $1-\delta/4$.}
        Let $m \leftarrow O(\log(n)\log(\Delta d)/(\epsilon^2))$\;
        Sample $M\in \mathbb{R}^{[m] \times I}$, where each entry of $M$ follows the Cauchy distribution, whose density function is $c(x) = \frac{1}{\pi(1+x^2)}$. 
        Notice that $M(\cdot)\colon\mathbb{R}^{I} \rightarrow \mathbb{R}^m$ is an oblivious linear mapping \cite{indyk06:l1-dimension-reduction}\;

        \store{$I$, $M(R) = \{Mr_i|i\in [n]\}$, $\{p^{(b)}| b\in I\}$, $M$ }
    }
    \query{$(q)$ \tcp*[f]{Inputs: $q\in [-\Delta,\Delta]^d$}}{
        Query $q^{(b)}, \forall b\in I$\;
        Let $u\in \mathbb{R}^{I}$, where $\forall b \in I, u^{(b)} = q^{(b)} / p^{(b)}$\;
        Let $F((x_1, \dots, x_m), (y_1, \dots, y_m)) \coloneq \mathsf{median}(|x_1-y_1|, \dots, |x_m- y_m|)$\;
        Let $\hat{i} = \argmin_{i\in [n]} F(Mr_i, Mu)$\;
        \Return{$\hat{i}$}
    }
\end{algorithm}

Let $\Gamma=\dist(c_{\hat{i}},q)/\dist(c_{i^*},q)$ be the random variable that describes the quality of the solution $\hat{i}$ returned by the algorithm compared to the optimal solution $i^*$. Note that $\hat{i}$ is also a random variable. 

We show that in expectation the algorithm outputs an $O(1)$-nearest neighbor.

\begin{theorem}\label{thm:expectedthm}
$\mathbb{E}[\Gamma] = O(1)$.
\end{theorem}

In the following lemma, we bound the probability of the event that the dimension reduction preserves the distances well. It is a direct consequence of \Cref{thm:indyk06} by plugging in $\delta=1/(4\Delta dn)$ and applying the union bound.   

Let $G_1$ be the event that $F(Mr_i,Mu)\in [1-\eps/2,1+\eps/2]\|r_i-u\|_1$ for all centers $c_i$.

\begin{lemma}\label{lem:jlfails}
    It holds that $\mathbb{P}[G_1] \ge 1- \epsilon/(4\Delta d )$
\end{lemma}

The definitions of good and bad regions in this case are similar to the ones in \cref{sec:l1samplereuse} except for the specific parameters. We define them in the following for completeness.

%We define the good region and the bad region from \cref{sec:l1samplereuse} again because in \Cref{alg:expected-l1} there is no parameter $\delta$ that the regions can depend on. Thus, the definition below is the same as in \cref{sec:l1samplereuse} but we fix $\delta=1$. 

\begin{definition}[Good and bad regions]
    Consider two points $c_1, c_2\in \R^d$, and a query point $q\in\R^d$. We define the good region $\Good(c_1, c_2, q)$ w.r.t.\ $(c_1,c_2,q)$ as $$\{z\in [d] \mid\min\{c_1^{(z)}, c_2^{(z)}\} - \frac{8}{\epsilon } |c_1^{(z)} - c_2^{(z)}| \leq q^{(z)} \leq \max\{c_1^{(z)}, c_2^{(z)}\} + \frac{8}{\epsilon } |c_1^{(z)} - c_2^{(z)}|\}$$. We define the bad region w.r.t.\ $(c_1,c_2,q)$ as $\Bad(c_1, c_2, q) = [d]\setminus \Good(c_1, c_2, q)$. When there is no ambiguity, we use $\Good$ for $\Good(c_1, c_2, q)$, and $\Bad$ for $\Bad(c_1, c_2, q)$. 
\end{definition}

Let $G_2$ be the event that $\dist_\Good(r_i, u) \approx_\epsilon T\cdot \dist_\Good(c_i, q)$ for all centers $c_i$ such that $(\|c_i- q\|_1)/(\|c_{i^*}-q\|_1)\geq 1+\epsilon$. 
\begin{lemma}\label{lem:G2}
    $\mathbb{P}[G_2] \ge 1-\epsilon/(4\Delta d)$
\end{lemma}

\begin{proof}
Let $c_i\ne c_{i^*}$ be an arbitrary center such that $(\|c_i- q\|_1)/(\|c_{i^*}-q\|_1)\geq 1+\epsilon$.    
In the algorithm, we construct $\{r_i\mid i\in [n]\}$ and $u$ through sampling a multiset $I$. To analyze how the data structure samples $I$ in the preprocessing, we define the following random variable $\forall t\in [T], z\in \Good$
    \[P_{t, z} =  \begin{cases} 
      \frac{|q^{(z)} - c_i^{(z)}|}{p^{(z)} }, & w.p.\; p^{(z)},\hspace{0.95cm}\text{i.e. if we add $z$ to $I$ at iteration $t$.}\\
      0, & w.p.\; 1-p^{(z)}, \quad\text{i.e. if we do not add $z$ to $I$ at iteration $t$.}
   \end{cases}\]
   In the rest of the proof, we use $\E[\sum_{z\in \Good} P_{t,z}]$ without specifying $t\in [T]$, because they have the same value for all $t\in [T]$.
   We obtain the following four observations by the same argument as in the proof of \cref{lemma:chernoff_for_good_l1}.
   \begin{itemize}
       \item $\sum_{t\in [T],z\in \Good} P_{t,z} = \dist_\Good(r_i, u)$.
       \item $\forall t\in [T], z\in \Good$, $0 \leq P_{t,z} \leq \left(1+\frac{8}{\epsilon}\right) \|c_i - c_{i^*}\|_1$.
        \item  $\E[\sum_{z\in \Good} P_{t,z}] = \dist_{\Good}(c_i, q)$.
        \item $\E[\sum_{z\in \Good}P_{t,z}] \geq \frac{1}{2} \dist_{\Good}(c_i, c_{i^*})$.
   \end{itemize}
   
     Using Chernoff bound, 
    \begin{align*}
    &\Pr\left[\left|\sum_{t\in [T], z\in \Good} P_{t,z} - \E\left[\sum_{t\in [T], z\in \Good} P_{t,z}\right]\right| \geq \epsilon_2 \E\left[\sum_{t\in [T], z\in \Good} P_{t,z}\right]\right] \\
    \leq & 2\exp\left(-\frac{\epsilon_2^2 T \E[\sum_{z\in [d]} P_{t,z}]}{3(1+\frac{8}{\epsilon}) \|c_i - c_{i^*}\|_1}\right) \\
    = & 2\exp\left(-\frac{\epsilon_2^2 T }{3(1+\frac{8}{\epsilon})} \frac{\dist_{\Good}(c_i, q)}{\dist_{\Good}(c_i, c_{i^*})} \frac{\dist_{\Good}(c_i, c_{i^*})}{\|c_i - c_{i^*}\|_1}\right) \\
    \leq & 2\exp\left(-\frac{\epsilon_2^2 T }{3(1+\frac{8}{\epsilon})} \cdot \frac12 \cdot \frac{13}{16}\right) \leq \frac{1}{2\Delta n},
    \end{align*}
    by choosing 
    $T = \Theta(\frac{\log(n\Delta )}{\epsilon^3} )$. Note that the second last inequality above follows from \cref{lemma:good-region-is-large-l1}.

    We conclude the lemma by union bound over all $c_i\ne c_{i^*}$ such that $(\|c_i- q\|_1)/(\|c_{i^*}-q\|_1)\geq 1+\epsilon$.
\end{proof}

Let $G = G_1 \wedge G_2$. \Cref{lem:jlfails,lem:G2}, together with union bound gives the following.

\begin{lemma}\label{lem:goodevent}
    $\mathbb{P}[G] \ge 1-\epsilon/(2\Delta d)$.
\end{lemma}

Let $\Lambda = \dist(r_{i^*},u)/(T\cdot\dist(c_{i^*},q))$ be the random variable that describes how well the coordinate selection preserves the distance from the query $q$ to its nearest neighbor $c_{i^*}$.
Let $L$ describe the set of all possible values $\Lambda$ could take over all possible random outcomes of the algorithm. Note that it is finite since there are only finitely many outcomes for the multiset~$I$.
 For $\lambda \in L$, we define $N_\lambda = \{i\in [n] \mid \dist(c_i,q) \le 4\max\{ \lambda,1\} \dist(c_{i^*},q)\}$ and $F_\lambda = \{i\in [n] \mid \dist(c_i,q) > 4\max\{ \lambda,1\}\dist(c_{i^*},q) \}$. For $i\in N_\lambda$ we call $c_i$ $\lambda$-near. Accordingly, we call $c_i$ $\lambda$-far if $i\in F_\lambda$.

We bound the probability of the event that  the distance of ``far'' centers is close to the distance of the nearest neighbor after the coordinate selection.
\begin{lemma}\label{lem:coordinateselectionfails}
    Let $\lambda \in L$. Let $c_i$ be an arbitrary $\lambda$-far center (i.e.\ $i\in F_\lambda$). If event $G$ holds and $\Lambda = \lambda$, then $\dist(r_{i},u)> (1+\epsilon)^2 \dist(r_{i^*},u)$. 
\end{lemma}
\begin{proof}
    If $G$ holds, then $\dist_\Good(r_i, u) \approx_\epsilon T\cdot\dist_\Good(c_i, q)$.
    From the definition of bad region, we have 
    \begin{itemize}
        \item $\dist_\Bad(r_i, u) \approx_\epsilon \dist_\Bad(r_{i^*}, u)$
        \item $\dist_\Bad(c_i, q) \approx_\epsilon \dist_\Bad(c_{i^*}, q)$
    \end{itemize}

    Since $i\in F_\lambda$, we know $\frac{\|c_i - q\|_1}{\|c_{i^*} - q\|_1} \geq 4\max\{1,\lambda\}$. By separating coordinates into good and bad regions, we have $\frac{\dist_\Good(c_i, q) + \dist_\Bad(c_i, q)}{\dist_\Good(c_{i^*}, q) + \dist_\Bad(c_{i^*}, q)} \geq 4\max\{1,\lambda\}$. Subtracting both sides by $(1+\epsilon)$ gives $$\frac{\dist_\Good(c_i, q) + \dist_\Bad(c_i, q) - (1+\epsilon)\cdot\left(\dist_\Good(c_{i^*}, q) + \dist_\Bad(c_{i^*}, q)\right)}{\dist_\Good(c_{i^*}, q) + \dist_\Bad(c_{i^*}, q)} \geq 2\max\{1,\lambda\}.$$
    Since we know that $\dist_\Bad(c_i, q) - (1+\epsilon) \dist_\Bad(c_{i^*}, q) \leq 0$, we have $\frac{\dist_\Good(c_i, q)}{\dist_\Good(c_{i^*}, q) + \dist_\Bad(c_{i^*}, q)} \geq 2\max\{1,\lambda\}$, which is $\frac{\dist_\Good(c_i, q)}{\|c_{i^*} - q\|_1} \geq 2\max\{1,\lambda\}$. Since $\dist_\Good(r_i, u) \approx_\epsilon T\cdot\dist_\Good(c_i,q)$, and $\|r_{i^*} - u\|_1 = \lambda\cdot T\cdot\|c_{i^*}-q\|_1$, we have $\frac{\dist_\Good(r_i, u)}{\|r_{i^*} - u\|_1} \geq \frac{2\max\{1,\lambda\}}{(1+\epsilon)\lambda}$. Thus $\frac{\|r_i - u\|_1}{\|r_{i^*} - u\|_1} \geq 2/(1+\epsilon) \geq (1+\epsilon)^2$.
\end{proof}

We next bound the probability that the algorithm outputs a far center.
\begin{lemma}\label{lem:badevent}
Let $\lambda \in L$. Let $c_i$ be an arbitrary $\lambda$-far center (i.e.\ $i\in F_\lambda$). If event $G$ holds and $\Lambda =\lambda$, then $\hat{i}\ne i$. 
\end{lemma}
\begin{proof}
By \Cref{lem:coordinateselectionfails} it follows that $\dist(r_{i},u)> (1+\epsilon)^2 \dist(r_{i^*},u)$. Since the event $G$ contains the event $G_1$ it follows that $F(Mr_i,Mu) > F(Mr_{i^*},Mu)$. Thus, the algorithm does not output $i$.
\end{proof}
We show that the approximation ratio of the output is linearly bounded by how well the algorithm approximates the distance to the closest center. 
\begin{lemma}\label{lem:linearina}
For $\lambda \in L$ it holds that 
$\mathbb{E}[\Gamma\mid \Lambda=\lambda\wedge G] \le 4\max\{\lambda,1\}$
\end{lemma}
\begin{proof}
The idea is to distinguish between the expected cost for returning ``near'' and ``far'' centers.
\begin{align*}
\mathbb{E}[\Gamma\mid \Lambda=\lambda\wedge G]&\le \mathbb{P}[\hat{i}\in N_\lambda\mid \Lambda = \lambda\wedge G]4\max\{\lambda,1\}+\mathbb{P}[\hat{i}\in F_\lambda\mid \Lambda =\lambda\wedge G]2\Delta d\\
&\le 4\max\{\lambda,1\} + \mathbb{P}[\hat{i}\in F_\lambda \mid \Lambda =\lambda\wedge G]2\Delta d\\
&\le 4\max\{\lambda,1\}
\end{align*}
The last inequality follows from \Cref{lem:badevent}.
\end{proof}

We use \Cref{lem:linearina} to conclude the proof for the constant approximation factor. 
\begin{proof}[Proof of \ref{thm:expectedthm}]
\begin{align}
\mathbb{E}[\Gamma] &=\mathbb{E}[\Gamma\mid \overline{G}]\mathbb{P}[\overline{G}]+\mathbb{E}[\Gamma\mid G]\mathbb{P}[G]\nonumber\\
&\le 2\Delta d \frac{\epsilon}{2\Delta d}+\mathbb{P}[G]\mathbb{E}[\Gamma\mid G]\label{eq:maxdist}\\
&= \epsilon+ \mathbb{P}[G]\sum_{\lambda\in L}\mathbb{E}[\Gamma\mid \Lambda=\lambda\wedge G]\mathbb{P}[\Lambda = \lambda\mid G]\nonumber\\
&\le \epsilon + \sum_{\lambda \in L}\mathbb{E}[\Gamma\mid \Lambda = \lambda\wedge G]\mathbb{P}[\Lambda=\lambda]\label{eq:removecondition}\\
&\le \epsilon + \sum_{\lambda \in L}4\max\{\lambda,1\}\mathbb{P}[\Lambda=\lambda]\label{eq:linearina}\\
&\le \epsilon + \sum_{\lambda \in L}4(\lambda+1)\mathbb{P}[\Lambda = \lambda]\nonumber\\
&\le \epsilon + 4\sum_{\lambda \in L} \lambda\mathbb{P}[\Lambda =\lambda] + 4\sum_{\lambda \in L} \mathbb{P}[\Lambda=\lambda]\nonumber\\
&\le \epsilon + 4\mathbb{E}[\Lambda] + 4 \leq 8+\epsilon\label{eq:lastline}
\end{align}
In Line~\ref{eq:maxdist} we use the fact that the maximal distance inside the $[-\Delta,\Delta]^d$ grid is $2\Delta d$ and \Cref{lem:goodevent}. In Line~\ref{eq:removecondition} we use the fact $\mathbb{P}[\Lambda=\lambda\mid G] \le \mathbb{P}[\Lambda =\lambda]/\mathbb{P}[G]$. In Line~\ref{eq:linearina} we use \cref{lem:linearina}. In Line~\ref{eq:lastline} we use the fact that $\E[\Lambda]=1$.
\end{proof}

It remains to show that with high probability the data structure does not violate the space and query complexity. The analysis is analogous to \cref{sec:l1-space-query-analysis}.
\begin{claim}\label{clm:I-size-upperbound-l1expected}
    With probability $1-1/(nd\Delta)$, we have $|I| \leq 2Tn$.    
\end{claim}
\begin{proof}
   We define $X_{t,z}$ for $t \in [T], z \in [d]$ to be the indicator random variable for the event that we add $z$ to $I$ at iteration $t$. Note that $\Pr[X_{t,z} = 1] = p^{(z)}$.
   % we define the following random variables $\forall t\in [T], z\in [d]$,
    %\begin{equation*}
     %   X_{t, z} =  \begin{cases} 
     % 1, & w.p.\; p^{(z)},\quad\text{i.e. if we add $z$ to $I$ at iteration $t$.}\\
     % 0, & w.p.\; 1-p^{(z)}, \quad\text{i.e. if we don't add $z$ to $I$ at iteration $t$.}
      %\end{cases}
    %\end{equation*}
    We can see that $|I| = \sum_{t\in[T], z\in[d]} X_{t,z}$.
    Hence, $\E[|I|] = \sum_{t\in[T], z\in[d]} p^{(z)} \leq T\cdot \sum_{b \in [d]} p^{(b)} \leq T\cdot n$, where the last inequality follows from \cref{lemma:math-program}. Similarly, we can see that 
    $\E[|I|] \geq T$.
    
    Using Chernoff bound, we have 
    \begin{equation*}
        \Pr\left[|I| - \E\left[|I|\right] \geq \E\left[|I|\right]\right] \leq \exp (-\E[|I|]/3) \leq \exp(-T/3) \leq 1/(nd\Delta).
    \end{equation*}
    Thus, with probability $1/(nd\Delta)$, we have $|I| \leq 2 E[|I|] \leq 2Tn$.
\end{proof}

\begin{lemma}
    With high probability the data structure has $O(n \polylog(n,d,\Delta)/(\epsilon^3))$ space and $O(n \polylog(n,d,\Delta)/(\epsilon^3))$ query complexity.
\end{lemma}
\begin{proof}
    By \cref{clm:I-size-upperbound-l1expected}, we have that the number of coordinates of $q$ that we query is $|I| = O(Tn) = O(\frac{n\log(n)\log(\Delta d)}{\epsilon^3})$. Next we analyze the space complexity of the data structure. Matrix $M$ is of wordsize $O(m|I|)$, $M(R)$ is of wordsize $O(mn)$, $I$ is of wordsize $O(|I| \log(d))$. Thus the space complexity of the our data strcuture is $O(m|I| + mn +|I|\log(d)) = O(n\polylog(n,d,\Delta)/(\epsilon^3))$ wordsize.
\end{proof}

\section{Data Structure with \texorpdfstring{$\tilde{O}(n^2)$}{near-quadratic} Space for General \texorpdfstring{$\ell_p$}{lp} Metrics}\label{sec:point-reduction}

In this section, we describe our data structure with $O(n^2 \log(n/\delta) \log(d) (\log\log n)^p/\epsilon^{p+2})$ space and $O(n \log (n/\delta) (\log\log n)^p/\epsilon^{p+2})$ query time for the approximate nearest neighbor problem for the $\ell_p$ metric for any $p$ and prove~\cref{theorem:Lp-NN}. 

\begin{theorem}\label{theorem:Lp-NN}
    Consider a set $C$ of $n$ points $c_1, \dots, c_n \in \mathbb{R}^d$ with $\ell_p$ metric, $1\leq p < \infty$. Given $0 < \epsilon < 1/(4p)$, we can construct a data structure such that the data structure is of word size $O(n^2 \log(n/\delta) \log(d)(\log\log n)^p/\epsilon^{p+2})$, and for any query $q\in R^d$, the following conditions hold. (i) The data structure reads $O(n \log (n/\delta)(\log\log n)^p/\epsilon^{p+2})$ coordinates of $q$. (ii) With probability $1-\delta$, it returns $i$ where $c_i$ is a $(1+\epsilon)$-approximate nearest neighbor of $q$ from $C$.
\end{theorem}

\subsection{A Data Structure for Two Input Points}\label{sec:2-center-lp}

We start with a data structure (see \cref{alg:alg-2points}) answering approximate nearest neighbor queries for the case of $2$ points, say $a$ and $b$. 
It generates \cref{thm:warm-up-l1}, and the main differences are: (1) It works for $\ell_p$ metric. (2) It explicitly use truncation to remove the bounding box assumption.

\begin{lemma}(2-point comparator, $\ell_p$)\label{lem:lp-2-point-comparator}
    Given $0<\epsilon<1/(4p)$, there exists a data structure $A_{(a,b)}$ for any two points $a, b \in \R^d$, such that for any query $q\in \R^d$, the following conditions hold: (i) $A_{(a,b)}$ is of wordsize $O(\log (1/\delta) \log (d)2^p/\epsilon^{p+2})$. (ii) $A_{(a,b)}$ reads $O( \log (1/\delta)2^p/\epsilon^{p+2})$ coordinates of $q$. (iii) $A_{(a,b)}$ returns 
    a $(1+\epsilon)$-approximate nearest neighbor of $q$ in $\ell_p$ metric with probability $1-\delta$.
\end{lemma}

\begin{algorithm}[htbp]\label{alg:alg-2points}
    \caption{An approximate nearest neighbor data structure $A_{(a,b)}$ for two points $a,b\in \R^d$.}
    \SetKwProg{preprocessing}{Preprocessing}{}{}
    \SetKwProg{query}{Query}{}{}
    \SetKw{store}{store}
    \preprocessing{$(a,b, \epsilon, \delta)$\tcp*[f]{Inputs: $a,b\in \R^d, \epsilon \in (0, \frac{1}{4p}), \delta\in (0,1)$}}{
        Let $I \leftarrow \emptyset$ be a multiset and
        $T \leftarrow O( \log (1/\delta)2^p/\epsilon^{p+2})$\;
        \For{$t\in [T]$}{
            Sample $i_t\in [d]$, where $
        \forall i\in [d], \Pr[i_t = i] = \frac{|b^{(i)} - a^{(i)}|^p}{\sum_{j\in [d]} |b^{(j)} - a^{(j)}|^p}$\;
        }
        \store{$I, a^{(I)}, b^{(I)}$}.
    }
    \query{$(q)$ \tcp*[f]{Inputs: $q\in \R^d$}}{
        \For{$i\in I$}{
            Query $q$ on coordinate ${i}$, i.e. $q^{(i)}$\; 
            \If{$p>1$ }{
                Let $m(\epsilon, p) = (1+\epsilon)^{p/(p-1)}-1$\;
                $l^{(i)} \leftarrow -\frac{1}{m(\epsilon, p)}(b^{(i)} - a^{(i)}) + a^{(i)}$ and $u^{(i)} \leftarrow (1+\frac{1}{m(\epsilon, p)})(b^{(i)} - a^{(i)})+a^{(i)}$\; \label{line:def-lu-1}
            }
            \Else($p=1$){
                $l^{(i)} \leftarrow a^{(i)}$ and $u^{(i)} \leftarrow b^{(i)}$\; \label{line:def-lu-2}
            }
            Truncate $q'^{(i)} \leftarrow \trunc_{\{l^{(i)}, u^{(i)}\}}(q^{(i)})$\;
            Let $X_t = \frac{|q'^{(i_t)} - a^{(i_t)}|^p}{|b^{(i_t)} - a^{(i_t)}|^p}, Y_t = \frac{|b^{(i_t)} - q'^{(i_t)}|^p}{|b^{(i_t)} - a^{(i_t)}|^p}$\;
        }
        \lIf{$\sum_{t\in [T]} X_t \leq \sum_{t\in [T]} Y_t$}{\Return{$a$} as the approximate nearest neighbor}\lElse{\Return{$b$}}
    }
\end{algorithm}
\begin{proof}
    Assume w.l.o.g.\ that $q\notin \{a,b\}$ to avoid the ratio $\frac{\dist(a,q)}{\dist(b,q)}$ to be $\infty$, and our data structure works without this assumption. 
    The idea is to determine if the ratio of distances $\frac{\dist(a,q)}{\dist(b,q)}$ is greater or smaller than $1$. To achieve this, we \emph{truncate} $q$ to $q'$ so that if
    %to bound the random variables $X_t, Y_t$, disturbing the ratio in a good way, i.e. if 
    $1+\epsilon \leq \frac{\dist(a,q)}{\dist(b,q)}$, then $\frac{\dist(a,q)}{\dist(b,q)} \leq \frac{\dist(a,q')}{\dist(b,q')}$, making it potentially easier to distinguish between $a$ and $b$ based on their distances to $q'$. %, making it easier to distinguish. 
    
    % If the ratio $\frac{\dist(a,q')}{\dist(b,q')}$ lies in $[\frac{1}{8(1+\epsilon/4)^{9p}}, 8(1+\epsilon/4)^{9p}]$ \zihang{change params}, we give a $(1+\epsilon)$-approximation for the ratio; otherwise, we can distinguish whether the ratio is in $(0, \frac{1}{8(1+\epsilon/4)^{9p}})$ or in $(8(1+\epsilon/4)^{9p}, \infty)$, which suffices to choose the nearest neighbor. 
    
    In the following, we assume w.l.o.g.\ that $a = (0, 0, \dots, 0)$, $b = (b^{(1)}, b^{(2)}, \dots, b^{(d)})$, where $b^{(i)} \geq 0$. It is okay to make this assumption because we can flip and shift the coordinates. Let $q = (q^{(1)}, \dots, q^{(d)})$. In the rest of the proof, we use $A$ to denote $A_{(a,b)}$.
    \begin{definition}{(Truncate)}\label{def:truncation}
    For $l,u \in \R$, and $l< u$, 
            \begin{equation*}
                \trunc_{\{l,u\}}(x) = 
                \begin{cases}
                    l, & x \leq l,\\
                    x, & l < x < u,\\
                    u, & x \geq u.
                \end{cases}
            \end{equation*}
\end{definition}

    % \paragraph{Preprocessing.} The preprocessing phase of $A$ is to sample a multiset of coordinates $I = \{i_1, i_2, \dots, i_T\}$, where $T = O( \log (1/\delta)2^p/\epsilon^{p+2})$, and each $i_t$ is sampled independently according to the following distribution:
    % \begin{equation*}
    %     \forall i\in [d], \Pr[i_t = i] = \frac{|b_i - a_i|^p}{\sum_{j\in [d]} |b_j - a_j|^p} = \frac{|b_i|^p}{\sum_{j\in [d]} |b_j|^p}.
    % \end{equation*}
    % $A$ stores $I, a|_I, b|_I$, where $a|_I \triangleq \{a_i | i\in I\}$ and similarly for $b|_I$.

    % \paragraph{Query.} $A$ answers query $q$ as follows:
    % \begin{itemize}
    %     \item At iteration $t \in [T]$:
    %     \begin{itemize}
    %         \item Query $q$ on coordinate ${i_t}$, i.e. $q_{i_t}$.
    %         \item Truncate $q_{i_t}$ to $q_{i_t}'$, where $q_{i_t}' = \trunc_{\{l_{i_t}, u_{i_t}\}}(q_{i_t})$, where 
    %         \begin{align}
    %         \label{eq:truncation-threshold}
    %         \begin{split}
    %          \text{for } &p > 1, l_i \triangleq -\frac{1}{(1+\epsilon)^{p/(p-1)}-1}(b_i - a_i) + a_i, u_i \triangleq (1+\frac{1}{(1+\epsilon)^{p/(p-1)}-1})(b_i - a_i)+a_i.
    %         \\
    %          \text{for } &p = 1, l_i \triangleq a_i, u_i \triangleq b_i.
    %         \end{split}
    %         \end{align}

    %         \item Let $X_t = \frac{|q_{i_t}' - a_{i_t}|^p}{|b_{i_t} - a_{i_t}|^p}, Y_t = \frac{|b_{i_t} - q_{i_t}'|^p}{|b_{i_t} - a_{i_t}|^p}$. 
    %     \end{itemize}
    %     \item If $\sum_{t\in [T]} X_t \leq \sum_{t\in [T]} Y_t$, return $a$ as the approximate nearest neighbor; o.w. return $b$.
    % \end{itemize}

\paragraph{Analysis.}  Let $q'$ be a truncated point where $q'^{(i)} = \trunc_{\{l^{(i)}, u^{(i)}\}}(q^{(i)})$, $ \forall i\in [d]$.  

We can bound the random variables:
    $0 \leq X_t, Y_t \leq (1+1/\epsilon)^p$ because 
    \begin{equation}\label{eq:rv-bound}
        0 \leq X_t \leq (1+\frac{1}{(1+\epsilon)^{p/(p-1)}-1})^p \leq (1+\frac{1}{(1+\epsilon(p/(p-1))-1})^p = (1+\frac{p-1}{p\epsilon})^p \leq (1+1/\epsilon)^{p}.          
    \end{equation}
        
    and similar arguments hold for $Y_t$. %\footnote{Notice that this bound is not tight for $\ell_1$, where we can give bound $0\leq X_t \leq 1$. This tighter bound will be useful in our second algorithm.}
    
    % Denote the true nearest neighbor of $q$ in $\{a,b\}$ as $c^*(q)$, the nearest neighbor of $q'$ in $\{a,b\}$ as $c^*(q')$, and output of $A$ as $\hat{c}$.
    
    \begin{claim}{(Truncation preserves ratio)}\label{claim:truncate} 
    For $\ell_p$ metric, let $q'$ be a truncated point, where
    $q'^{(i)} = \trunc_{\{l^{(i)}, u^{(i)}\}}(q^{(i)})$. $l^{(i)}, u^{(i)}$ are defined in Line \ref{line:def-lu-1} and \ref{line:def-lu-2} of \cref{alg:alg-2points}.
    % Eq \eqref{eq:truncation-threshold}. 

    % For $\ell_p$ metric $(p>1)$, 
    % let $q'$ be a truncated point, where
    % $q_i' = \trunc_{\{-\frac{1}{(1+\epsilon)^{p/(p-1)}-1}(b_i - a_i) + a_i, (1+\frac{1}{(1+\epsilon)^{p/(p-1)}-1})(b_i - a_i)+a_i\}}(q_{i})$. 
    % For $\ell_1$ metric, let $q'$ be a truncated point, where $q_i' = \trunc_{\{a_i, b_i\}}(q_i)$.
    
    If $\frac{\dist(a,q)}{\dist(b,q)} \geq (1+\epsilon)$, then 
    \begin{equation*}
        \frac{\dist(a,q)}{\dist(b,q)} \leq \frac{\dist(a,q')}{\dist(b,q')}.
    \end{equation*}
    If $\frac{\dist(a,q)}{\dist(b,q)} \leq 1/(1+\epsilon)$, then
    \begin{equation*}
        \frac{\dist(a,q)}{\dist(b,q)} \geq \frac{\dist(a,q')}{\dist(b,q')}.
    \end{equation*}

    \end{claim}
        
    This claim shows that if we can return a $(1+\epsilon)$-approximate nearest neighbor of $q'$, it would also be a $(1+\epsilon)$-approximate nearest neighbor of $q$. This is because, 
    \begin{itemize}
        \item If $\dist(a,q) > (1+\epsilon)\dist(b,q)$, then from above claim $\dist(a,q') > (1+\epsilon)\dist(b,q')$. Since it returns $(1+\epsilon)$-approximate nearest neighbor of $q'$, it must return $b$. This is the nearest neighbor for $q$.
        \item If $\frac{1}{1+\epsilon} \leq \frac{\dist(a,q)}{\dist(b,q)} \leq (1+\epsilon)$, then both $a$ and $b$ are $(1+\epsilon)$-approximation nearest neighbor. Thus it does not matter which one is returned.
        \item If If $\dist(a,q) < 1/(1+\epsilon)\dist(b,q)$, this is similar to the first case, and it must returns $a$.
    \end{itemize}
    
    Now let us prove that the data structure $A$ returns a $(1+\epsilon)$-approximate nearest neighbor of $q'$. 

    Assume that $\dist(a,q') \leq \dist(b,q')$, i.e. $a$ is the nearest neighbor of $q'$. Notice $(X_t)^{1/p} + (Y_t)^{1/p} \geq 1$ which gives $X_t + Y_t \geq \frac{1}{2^{p-1}}$, we have $\E[Y_t] \geq \frac{1}{2^p}.$ 
    Since $\E[X_t] = \dist^p(a,q')$ and $\E[Y_t] = \dist^p(b,q')$, we can use Chernoff bound to approximate $\dist(a,q')$ and $\dist(b,q')$.

    Let $s = \frac{1}{2^{p+2}} $.
    
    \noindent \underline{Case 1.} If $\E[X_t] \geq s$, using Chernoff bound, for $0<\epsilon_2 = p\epsilon/10 < 1$, we have 
    \begin{equation}\label{eq:chernoff-case1-x}
    \Pr\left[\left|\frac1T\sum_t X_t - \E[X_t]\right| \geq \epsilon_2 \E[X_t]\right] \leq 2\exp(-\frac{\epsilon_2^2 T \E[X_t]}{3(1+1/\epsilon)^p}) \leq \delta/2.
    \end{equation}
    by choosing 
    $T = \Theta( \frac{2^p}{p^2\epsilon^{p+2}}\log(1/\delta)) \geq \Theta((1+\frac{1}{\epsilon})^p \frac{1}{\epsilon_2^2} \frac{1}{s}\log(1/\delta)) $.
    
    Similarly for $Y_t$, using the same choice of $T$, since $\E[Y_t] \geq \frac{1}{2^p} \geq s$, $\Pr[\left|\frac1T\sum_t Y_t - \E[Y_t]\right| > \epsilon_2 \E[Y_t]] < \delta/2.$  

    By union bound, with probability $1-\delta$, we have 
    \begin{align}\label{eq:chernoff-case1}
    \begin{split}
        \frac{1}{T}\sum_t X_t &\in [1-\epsilon_2, 1+\epsilon_2] \E[X_t],\\
    \frac1T\sum_t Y_t &\in [1-\epsilon_2, 1+\epsilon_2] \E[Y_t].
    \end{split}
    \end{align}

    Recall that $a$ is the nearest neighbor of $q'$. If $\dist(b,q')/\dist(a,q') \geq 1+\epsilon$, 
    \begin{equation*}
        \frac{\sum Y_t}{\sum X_t} \geq \frac{1-\epsilon_2}{1+\epsilon_2} \frac{\E[Y_t]}{\E[X_t]} \geq \frac{1-\epsilon_2}{1+\epsilon_2} (1+\epsilon)^p \geq 1.
    \end{equation*}
    Thus $A_{(a,b)}$ will returns $a$ correctly.

    \noindent \underline{Case 2.} If $\E[X_t] < s$. Using Chernoff bound, we have 
    \begin{equation}
    \Pr\left[\left|\frac1T\sum_t X_t - \E[X_t]\right| > s\right] < 2\exp\left( -\frac{Ts}{3(1+1/\epsilon)^p}\right) \leq \delta/2. \label{eq:chernoff-case2-x}
    \end{equation}
    by choosing $T = \Theta( \frac{2^p}{p^2\epsilon^{p+2}}\log(1/\delta)) \geq \Theta(\frac{1}{s}(1+\epsilon)^p \log(1/\delta))$.
    
    For $Y_t$, since $\E[Y_t] \geq \frac{1}{2^p} \geq s$, $\Pr[\left|\frac1T\sum_t Y_t - \E[Y_t]\right| > \frac12 \E[Y_t]] < \delta/2.$ 

    By union bound, with probability $1-\delta$, we have 
    \begin{equation}\label{eq:chernoff-case2}
        \frac1T\sum_t X_t \leq 2s, \quad\frac1T\sum_t Y_t \in \left[\frac12, \frac32\right] E[Y_t].
    \end{equation}
    Since 
    \[
    \frac{\sum Y_t}{\sum X_t} \geq \frac{1/2 \E[Y_t]}{2s} \geq 1,
    \]
    $A$ will output $a$ with probability $(1-\delta)$.

    To conclude, we proved that $A$ outputs a $(1+\epsilon)$-approximate nearest neighbor of $q'$ with probability $1-\delta$, which is also a $(1+\epsilon)$-approximate nearest neighbor of $q$.

\end{proof}

\begin{proof}{(For Claim \ref{claim:truncate})}
    Recall our assumption that $a = \mathbbm{0}$ and $\forall i, b_i\geq 0 $. We prove that if $\dist(a,q)/\dist(b,q) \geq 1+\epsilon$, \[
    \frac{\dist(a,q)}{\dist(b,q)} \leq \frac{\dist(a,q')}{\dist(b,q')}. \]
    We first argue for the cases of $\ell_1$ and $\ell_2$ metrics to give some intuition, and then give a general proof for the $\ell_p$ metric.
    
    (1) For $\ell_1$, it is true, because $\dist(a,q) - \dist(a,q') = \dist(b,q) - \dist(b,q') \geq 0$. 
    \smallskip
    
    (2) For $\ell_2$ (see \cref{fig:truncation}),
    $q'^{(i)} 
    % = \trunc_{\{-\frac{1}{2\epsilon+\epsilon^2}(b^{(i)} - a^{(i)}) + a^{(i)}, (1+\frac{1}{2\epsilon+\epsilon^2})(b^{(i)} - a^{(i)}) + a^{(i)}\}}(q^{(i)}) 
    = \trunc_{\{-\frac{1}{2\epsilon+\epsilon^2}b^{(i)}, (1+\frac{1}{2\epsilon+\epsilon^2})b^{(i)}\}}(q^{(i)})$.
    
     For the function $f(q) = \frac{\dist(a,q)}{\dist(b,q)}$, consider the contour $B(t) \triangleq \{x | f(x) = t\}$. For each $t\neq 1, B(t)$ is ball, because
    \[
    f(x) = t \Leftrightarrow \|x\|_2 = t \|x-b\|_2.
    \]
    
     Denote the center of $B(t)$ as $C(t)$ and the radius as $R(t)$. Denote $H$ as the truncation hyperrectange where $H \triangleq \{x | l^{(i)} \leq x^{(i)} \leq u^{(i)}, i\in [d]\}$.  Denote $c_1  \triangleq C(1+\epsilon) = (1+\frac{1}{2\epsilon+\epsilon^2}) b$, which is a corner of $H$. 
    Define $1+m \triangleq \frac{\dist(a,q)}{\dist(b,q)} \geq 1+\epsilon$. Denote $c_2 \triangleq C(1+m) = (1+\frac{1}{2m+m^2}) b$. 

    For simplicity, assume that there is only one coordinate $i$ that requires truncation. We can make this assumption because, if we have multiple coordinates to be truncated, we just repeat the 1-coordinate argument, and the ratio only increases.

    To prove that the ratio increases, it suffices to show that $q'$ is inside $B(1+m)$, and $q$ is on the surface of $B(1+m)$. The intuition is that, $B(1+\epsilon)$ is a ball centered in the corner of the truncation hyperrectangle $H$. If we project points from $B(1+\epsilon)$ to $H$, then the projected points lie inside $B(1+\epsilon)$. For $B(1+m)$, which is a smaller ball contained in $B(1+\epsilon)$, and whose center $C(1+m)$ is between $C(1+\epsilon)$ and $b$, this is also true. 
    \begin{figure}[h]
    \centering  
    \includegraphics[width=0.6\textwidth]{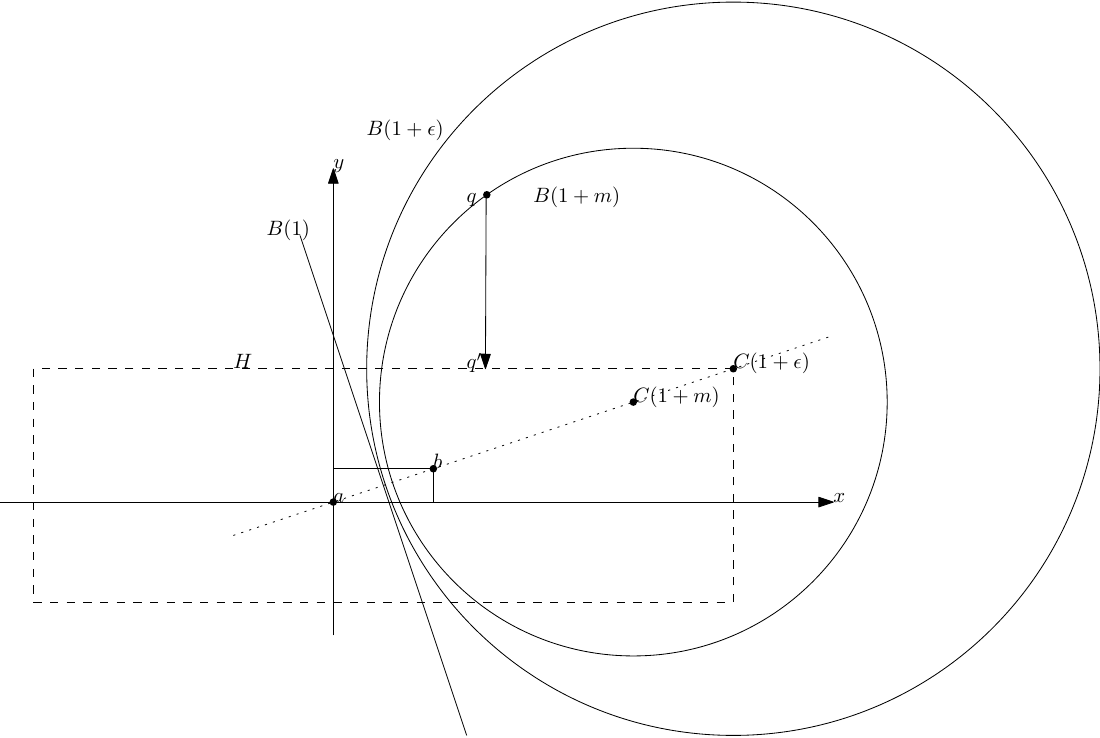}  \captionsetup{justification=centering}
    \caption{An example for $\ell_2$ metric: truncating $q$ to $q'$.}
    \label{fig:truncation}
    \end{figure}
    
    (3) For $\ell_p$, we do $1$-coordinate truncation repeatedly like in the $\ell_2$ case. Suppose we we are truncate the first coordinate.

    We have 
    \begin{equation*}
        \frac{\dist(a,q)}{\dist(b,q)}=\frac{|q^{(1)}|^p + |q^{(2)}|^p + \cdots + |q^{(d)}|^p }{|q^{(1)} - b^{(1)}|^p + |q^{(2)} - b^{(2)}|^p + \cdots + |q^{(d)} - b^{(d)}|^p} = (1+m)^p \geq (1+\epsilon)^p,
    \end{equation*}
    and we want to prove 
    \begin{equation*}
        \frac{|q'^{(1)}|^p + |q^{(2)}|^p + \cdots + |q^{(d)}|^p }{|q'^{(1)} - b^{(1)}|^p + |q^{(2)} - b^{(2)}|^p + \cdots + |q^{(d)} - b^{(d)}|^p} \geq (1+m)^p.
    \end{equation*}
    It suffices to prove that 
    \begin{equation*}
        |q'^{(1)}|^p - (1+m)^p |q'^{(1)} - b^{(1)}|^p \geq |q^{(1)}|^p - (1+m)^p |q^{(1)} - b^{(1)}|^p.
    \end{equation*}
    Case 1: $b^{(1)} \leq q'^{(1)} \leq q^{(1)}$.
    Define $g(x) = |x|^p - (1+m)^p |x-b^{(1)}|^p, x\geq b^{(1)}$. We claim that $q_*^{(1)} = (1+\frac{1}{(1+m)^{p/(p-1)}-1})b^{(1)}$ is the maximizer of $g(x)$. This is because derivative $g'(x) = p(x^{p-1} - (1+m)^p(x-b^{(1)})^{p-1})$. For $x\in [b^{(1)}, q_*^{(1)}])$, $g'(x)$ is positive, and for $x\in (q_*^{(1)}, \infty)$, $g'(x)$ is negative, which shows that $g(x)$ is decreasing in $[q^*, \infty)$. Since $q'^{(1)} = (1+\frac{1}{(1+\epsilon)^{p/(p-1)}-1})b_1 \geq q_*^{(1)}$ and $q^{(1)} \geq q'^{(1)}$, we know that $g(q'^{(1)}) \geq g(q^{(1)})$.  

    \noindent Case 2: $q^{(1)} \leq q'^{(1)} \leq 0$, the argument is similar.
\end{proof}

\subsection{A Data Structure for \texorpdfstring{$n$}{n} Points}\label{sec:transformation-2-to-n}
In this section, we prove Theorem \ref{theorem:Lp-NN}. 
Let $C = \{c_1, c_2, \dots c_n\} \subseteq \R^d$ denote the set of input points and let $q \in \R^d$ denote the query point. 
Our goal is to find an approximate nearest neighbor of $q$ in the set $C$. The argument used here is from \cite{AjtaiFHN16}, but we still present it to be self-contained.

\paragraph{Preprocessing.} For all $(c_i, c_j) \in \binom{C}{2}$, we run the preprocessing steps of the algorithm guaranteed by Lemma~\ref{lem:lp-2-point-comparator}, where we set $a \gets c_i$ and $b \gets c_j$. 

\paragraph{Query.} We first cast the nearest neighbor problem in the language of finding the minimum value in an array of length $n$ using only \emph{imperfect} comparison queries as follows. 
Consider an array of length $n$ whose values are $\dist(q,c_i)$ for $i \in [n]$. Although we do not have access to these values directly, Lemma~\ref{lem:lp-2-point-comparator} guarantees an algorithm that, given two indices $i, j \in [n]$, with probability at least $1 - \delta$, returns the symbol $>$ if $\dist(q,c_i) > (1+\epsilon) \dist(q,c_j)$ and returns the symbol $<$ if $(1+\epsilon) \dist(q,c_i) < \dist(q,c_j)$. In other words, we have access to a comparison operator for the entries of the array such that the comparator works correctly with high probability, if the entries are more than a factor of $1+\epsilon$ apart. 
We use the following recursive algorithm, called initially with $S$ set to $[n]$, can be used to find an \emph{approximate minimum} in this case. 

\smallskip

\noindent\textsf{Min-Finding($S$).}
The input to the algorithm is $S \subseteq [n]$, a set of array indices or, equivalently, the labels of input points.
\begin{enumerate}
    \item If $|S| = 2$, run the algorithm guaranteed by Lemma~\ref{lem:lp-2-point-comparator} and return the index of the smaller value. Note that this step corresponds to the use of the \emph{imperfect comparison operator}.
    \item If $|S| > 2$:
    \begin{enumerate}
        \item Partition $S$ into $\Theta(\sqrt{|S|})$ many subsets $\{T_1, T_2, \dots, T_{\Theta(\sqrt{|S|})}\}$ of size $\sqrt{|S|}$ each. 
        \item Recursively call \textsf{Min-Finding}$(T_i)$ for all $i \in \Theta(\sqrt{|S|})$. 
        \item Let $M = \{m_1, m_2, \dots,  m_{\Theta(\sqrt{|S|})}\}$ be the indices returned by the recursive calls.
        \item For each pair $(m_i, m_j) \in \binom{M}{2}$, run the algorithm guaranteed by Lemma~\ref{lem:lp-2-point-comparator}. 
        \item Return the index $m \in M$ that is returned during the execution of the above step the most number of times, where we break ties arbitrarily.
    \end{enumerate}
\end{enumerate}

\begin{lemma}
Let $c^* \in C$ be the nearest neighbor of $q$ among points in $C$. \textsf{Min-Finding}$([n])$, with probability at least $1 - \Theta(n \log \log n)\cdot \delta$, returns an index $i \in [n]$ such that $\dist(q,c_i) \leq (1+\epsilon)^{2t} \cdot \dist(q,c^*)$, where $t = \log \log n$. Moreover, the total number of times the query algorithm of Lemma~\ref{lem:lp-2-point-comparator} is called is $O(n\log \log n)$.
\end{lemma}
\begin{proof}
The recursion depth of the algorithm is $\log \log n$ by virtue of the recurrence $t(n) = t(\sqrt{n}) + 1$. 
The total number of calls to the query algorithm of Lemma~\ref{lem:lp-2-point-comparator} at each recursion level is $\Theta(n)$ and hence, overall, this 2-point query algorithm is invoked $O(n\log \log n)$ times. 

We first condition on the good event that every call to the query algorithm guaranteed by Lemma~\ref{lem:lp-2-point-comparator} is correct. This event happens with probability at least $1 - \Theta(n \log \log n)\cdot \delta$.
Next, we argue about the deterioration in the approximation guarantee.
The main idea is that Step 2(d) and 2(e) together return a $(1+\epsilon)^2$-approximate nearest neighbor of $q$ in the set $M$, given our conditioning.
Let $m^* \in M$ denote the nearest neighbor of $q$ in $M$.  
Assume for the sake of contradiction that the point $m \in M$ returned in Step 2(e) is such that $\dist(q,m) > (1+\epsilon)^2 \cdot \dist(q,m^*)$.  By the guarantee of our imperfect comparison operator, $m$ does not "win" against $m^*$ in the 2-point comparison game. Consider now an $m'' \in M$ that $m$ wins against in the 2-point comparison game. If $\dist(m'',q) > (1+\epsilon) \cdot \dist(m,q)$, then it must also be the case that $m^*$ wins against $m''$ in the comparison game.
Even if $\dist(m'',q) \leq (1+\epsilon) \cdot \dist(m,q)$ and $m$ wins against $m''$ because of the do not care condition,
the distance of $m''$ to $q$ is still a $(1+\epsilon)$ factor more than the distance of $m^*$ to $q$. Thus, $m^*$ one more 2-point comparison game as compared to $m$ and this is a contradiction.
Since we incur a loss of  a factor of $(1+\epsilon)^2$ in each recursion level, overall, the loss in approximation guarantee is $(1+\epsilon)^{2\log \log n}$.
\end{proof}

\begin{proof}[Proof of Theorem \ref{theorem:Lp-NN}]
To get our data structure, we run each invocation of the algorithm of Lemma~\ref{lem:lp-2-point-comparator} above (i.e., in Min-Finding) with $\epsilon = \epsilon'/\log \log n$ and $\delta = \delta' / \Theta(n \log \log n)$, we can get a $(1+\epsilon')$-approximate nearest neighbor data structure with success probability at least $1- \delta'$. 

The overall query time is $O(n \frac{2^p}{\epsilon^{p+2}} \log \frac{n}{\delta}  (\log \log n)^{p+3})$. Since we are running the preprocessing algorithm for all pairs of points, the space is $O(n^2 \log d \frac{2^p}{\epsilon^{p+2}} \log \frac{n}{\delta}  (\log \log n)^{p+3})$.
\end{proof}

\section{Implication for Range Search} \label{sec:range-search}

In this section, we describe the implications of our results to the approximate range search problem. 

%\subsection{Range Search} 
\begin{definition}[$(1+\epsilon)$-approximate range search problem]
    Given a set $C = \{c_1, c_2, \dots, c_n\}$ of $n$ points in a euclidean space $\mathbb{R}^d$, build a data structure $S$ that, given any query range $R\subseteq \mathbb{R}^d$ defined by two corner points $q_1,q_2 \in \mathbb{R}^d$ with the promise that there is an $i\in [n]$ such that $c_i \in R$, returns all indices $i \in [n]$, such that 
    \[
    c_i \in R
    \]
    but no $i\in [n]$ such that 
    \[
    \mathsf{dist}(c_i,R) > \epsilon \cdot \mathsf{diam}(R).
    \]
\end{definition}
\begin{algorithm}[htp]\label{alg:range-search}
    \caption{An approximate range search data structure for $n$ points in $\R^d$ for $\ell_1$ metric.}
    \SetKwProg{preprocessing}{Preprocessing}{}{}
    \SetKwProg{query}{Query}{}{}
    \SetKwProg{comparator}{$E$}{}{}
    \SetKw{store}{store}
    \preprocessing{$(C, \epsilon, \delta)$\tcp*[f]{Inputs: $C=\{c_1, \dots, c_n\}\subseteq \R^d, \epsilon \in (0, \frac{1}{4}), \delta\in (0,1)$}}{
        % Let $\epsilon \leftarrow \epsilon'/\log \log n$ and $\delta \leftarrow \delta' / \Theta(n \log \log n)$\;
        Let $I \leftarrow \emptyset$ be a multiset and
        $T \leftarrow O( \log (n/\delta)/\epsilon)$\;
        % \footnote{According to Theorem \ref{theorem:Lp-NN}, it should be $O(\log(1/\delta)/\epsilon^3)$, but we can achieve a better bound of $O(\log(1/\delta)/\epsilon^2)$  by using $0\leq X_t, Y_t\leq 1$, instead of using $0\leq X_t, Y_t\leq 1/\epsilon$ in the proof. }
        
        \For{$t\in [T]$}{
            \For{$b\in [d]$}{
                Add $b$ to $I$ with probability  $p^{(b)}\triangleq \max_{(i,j)\in {\binom{n}{2}}} \frac{\left|c_i^{(b)} - c_j^{(b)}\right|}{\|c_i- c_j\|_1}$\;
            }
        }
        
        \store{$I$, $C^{(I)} = \{c_1^{(I)}, c_2^{(I)}, \dots, c_n^{(I)}\}$. }
    }
    \query{$(q_1, q_2)$ \tcp*[f]{Inputs: $q_1, q_2\in \R^d$}}{
        Query $q_1^{(b)}$ and $q_2^{(b)}, \forall b\in I$\;
        Let $W = \emptyset$\;
        \For{$i\in [n]$}{
            \If{$\forall z\in I, \min\{q_1^{(z)},q_2^{(z)}\} \leq c_i^{(z)} \leq \max\{q_1^{(z)},q_2^{(z)}\}$}{
            Add $i$ to $W$\;
            }
        }
        \Return{$W$}
    }
\end{algorithm}

\begin{theorem}
    Consider a set $C$ of $n$ points $c_1, \dots, c_n \in \mathbb{R}^d$ with $\ell_1$ metric. Given $0 < \epsilon < 1/4$, we can construct a data structure (see \cref{alg:range-search}) of word size $O(n  \log(n/\delta) \log(d)/\epsilon)$ that, given any query range $R$ by the corner points $q_1,q_2 \in R^d$, reads $O(n \log (n/\delta)/\epsilon)$ coordinates of $q_1,q_2$, returns an $(1+\epsilon)$-approximate range search of $q_1,q_2$ from $C$ with probability $1-\delta$.
\end{theorem}
\begin{proof}
    First note that when $c_i\in R$, then we always return $i$.
    It remains to show that for $c_i$ such that $\mathsf{dist}(c_i,R) > \epsilon \cdot \mathsf{diam}(R)$, we do not return $i$ with probability at least $1-\delta$.

    We fix an arbitrary $i\in [n]$ such that $\mathsf{dist}(c_i,R) > \epsilon \cdot \mathsf{diam}(R)$. We partition the coordinates in two sets.
    \begin{align*}
        \mathsf{Good}&= \left\{b\in [d] \mid \min \{q_1^{(b)},q_2^{(b)}\} \le c_i^{(b)} \le \max\{q_1^{(b)},q_2^{(b)}\}\right\}\\
        \mathsf{Bad} &= [d]\setminus \mathsf{Good}
    \end{align*}
    For intuition $\mathsf{Bad}$ is the set of coordinates that certify that $c_i \notin R$.
    We show that a coordinate from $\mathsf{Bad}$ is queried with probability at least $1-\delta$ which concludes the proof.

    Let $\mathsf{dist}_{J}$ for a set of coordinates $J\subseteq [d]$ be the distance of two points projected to the subspace on the coordinates in $J$. 
    Furthermore, let $c_{i^*}$ be the point that is inside the range $R$ which exists by the promise of the problem. 
    \begin{claim}
        \begin{align*}
            \mathsf{dist}_{\mathsf{Bad}}(c_i,c_{i^*}) \ge \epsilon \mathsf{dist}_{\mathsf{Good}}(c_i,c_{i^*})
        \end{align*}
    \end{claim}
    \begin{proof}
        Assume towards contradiction that $\mathsf{dist}_{\mathsf{Bad}}(c_i,c_{i^*}) < \epsilon \mathsf{dist}_{\mathsf{Good}}(c_i,c_{i^*})$.
        Since $c_i$ projected to $\mathsf{Good}$ is in $R$ projected to $\mathsf{Good}$ by the definition of $\mathsf{Good}$, it follows that $\mathsf{dist}_{\mathsf{Good}}(c_i,c_{i^*})\le \mathsf{diam}(R)$.
        We conclude 
        \begin{align*}
            \epsilon \mathsf{diam}(R) \ge \epsilon \mathsf{dist}_{\mathsf{Good}}(c_i,c_{i^*}) > \mathsf{dist}_{\mathsf{Bad}}(c_i,c_{i^*}) = \mathsf{dist}(c_i,R).
        \end{align*}
        This is a contradiction to the choice of $c_i$. 
    \end{proof}
    We sample each coordinate $b\in [d]$ in an iteration with probability at least $p^{(b)}=(|c_i^{(b)}-c_{i^*}^{(b)}|)/(\sum_{b'\in [d]}|c_i^{(b')}-c_{i^*}^{(b')}|)$. 
    Thus we sample any coordinate from $\mathsf{Bad}$ with probability at least 
    \begin{align*}
        1-\left(\prod_{b\in \mathsf{Bad}}(1-p^{(b)})\right) &\ge 1-\left(\prod_{b\in \mathsf{Bad}} e^{-p^{(b)}}\right)
        = 1-\exp({-\sum_{b\in \mathsf{Bad}} {p^{(b)}}})
        \ge \left(1-\frac{1}{e}\right)\sum_{b\in \mathsf{Bad}}p^{(b)}\\
        &= \left(1-\frac{1}{e}\right)\frac{\sum_{b\in \mathsf{Bad}}|c_i^{(b)}-c_{i^*}^{(b)}|}{\sum_{b'\in \mathsf{Bad}}|c_i^{(b')}-c_{i^*}^{(b')}|+\sum_{b'\in \mathsf{Good}}|c_i^{(b')}-c_{i^*}^{(b')}|}.
    \end{align*}
    Since we consider $\ell_1$ metric the last expression is
    \begin{align*}
        &\left(1-\frac{1}{e}\right)\frac{\mathsf{dist}_{\mathsf{Bad}}(c_i,c_{i^*})}{\mathsf{dist}_{\mathsf{Bad}}(c_i,c_{i^*})+\mathsf{dist}_{\mathsf{Good}}(c_i,c_{i^*})} \\
        \ge & \left(1-\frac{1}{e}\right)\left(\frac{ \mathsf{dist}_{\mathsf{Bad}}(c_i,c_{i^*})+\mathsf{dist}_{\mathsf{Good}}(c_i,c_{i^*})}{\mathsf{dist}_{\mathsf{Bad}}(c_i,c_{i^*})}\right)^{-1}\\
        \ge &\left(1-\frac{1}{e}\right)\left(1+\frac{ \mathsf{dist}_{\mathsf{Good}}(c_i,c_{i^*})}{\mathsf{dist}_{\mathsf{Bad}}(c_i,c_{i^*})}\right)^{-1}\\
        \ge & \left(1-\frac{1}{e}\right)\left(1+\frac{ \mathsf{dist}_{\mathsf{Good}}(c_i,c_{i^*})}{\epsilon\mathsf{dist}_{ \mathsf{Good}}(c_i,c_{i^*})}\right)^{-1}\\
        = & \left(1-\frac{1}{e}\right)\frac{1}{1+\frac{1}{\epsilon}} \ge \frac{1}{4}\epsilon
    \end{align*}
    Since we repeat for $T= O( \log (n/\delta)/\epsilon)$ iterations, we sample a coordinate from $\mathsf{Bad}$ with probability at least $1-\delta/n$. Using union bound we succeed for all points $i\in [n]$ with probability at least $1-\delta$.
\end{proof}

\section{Reducing Approximate Nearest Neighbor under \texorpdfstring{$\ell_2$}{l2}-norm to Feature Selection for \texorpdfstring{$k$}{k}-means Clustering}\label{sec:reduction-clustering}

In this section, we show a $(\sqrt{3}+\epsilon)$-approximate nearest neighbor data structure under $\ell_2$-norm, using feature selection for $k$-means clustering. 

\begin{algorithm}[htp]
    \caption{An approximate nearest neighbor data structure for $n$ points in $\R^d$ for $\ell_2$ metric.}
    \label{alg:l2-ANN}
    \SetKwProg{preprocessing}{Preprocessing}{}{}
    \SetKwProg{query}{Query}{}{}
    \SetKwProg{selectfeature}{SelectFeature}{}{}
    \SetKwProg{comparator}{$E$}{}{}
    \SetKw{store}{store}
    \selectfeature{$(Z, t)$ \tcp*[f]{$Z \in \R^{n\times k}, t \in \R$}}{
        Let $Z_i$ be the $i$-th row of $Z$\;
        Let $p^{(i)} = \|Z_i\|_2^2 / \|Z\|_F^2$\;
        Initialize $\Omega = \mathbf{0}^{n \times t}, S = \mathbf{0}^{t \times t}$\;
        \For{$j=1$ to $t$}{
            Sample $i_j\in [n]$, with probabilities $\{p^{(i)}\}$\;
            Set $\Omega_{i_j, j} = 1$ and $S_{j,j} = 1/\sqrt{t p^{(i_j)}}$\;
        }
        \Return{$\Omega, S$}
    }
    \preprocessing{$(C, \epsilon, \delta)$\tcp*[f]{Inputs: $C\in \R^{n\times d}, \epsilon \in (0, \frac{1}{4}), \delta\in (0,1)$}}{
        Let $k = n$\; 
        Compute SVD of $C$, let $Z = V_k(C)$\;
        Determine the sampling probabilities and rescaling: $[\Omega, S] = SelectFeature(Z, t)$, where $t = O(n\log n/\epsilon^2)$\;
        Compute $R = C\Omega S \in \R^{n\times t}$ be the feature-selected centers, where $i$-th row of $R$ is $r_i$\;
        Sample $M\in \R^{t\times \log(n/\epsilon)/\epsilon^2}$ to be a JL-mapping\;        
        \store{$\Omega, S$, $RM = \{r_i M|i\in [n]\}$, $M$}
    }
    \query{$(q)$ \tcp*[f]{Inputs: $q\in \{0, 1\}^d$}}{
        Query $t$ coordinates of $q$ to compute $y^\transpose = q^\transpose \Omega S$\;
        Let $\hat{i} = \argmin_{i\in n} \|r_i M - y^\transpose M \|_2$\;
        \Return{$\hat{i}$}
    }
\end{algorithm}

\begin{theorem}\label{theorem:l2-ann}
    Consider $n$ points $c_1, \dots, c_n \in \R^d$ with $\ell_2$ metric. Given $0 < \epsilon < 1$, we can construct a  data structure (see \cref{alg:l2-ANN}) of word size $\tilde{O}(n/\poly(\epsilon))$ that, given any query $q\in \R^d$, reads $\tilde{O}(n/\poly(\epsilon))$ coordinates of $q$, and with probability $1-\delta$, returns $i$ where $c_i$ is a $(\sqrt{3}+\epsilon)$-approximate nearest neighbor of $q$ in $C$.
\end{theorem}

\subsection{Preliminaries}
Before proving the \cref{theorem:l2-ann}, we introduce some notations. 

For a matrix $A \in \R^{n \times d}$, we write the singular value decomposition (SVD) of $A$ as $A = U(A)\Sigma(A) V(A)^\transpose$, where $U(A) \in \R^{n \times n}$ and $V(A) \in \R^{d \times d}$ are orthonormal matrices, and $\Sigma(A) = diag\{\sigma_1, \sigma_2, \dots, \sigma_{rank(A)}, 0, \dots, 0\}$. For an integer $1 \leq k \leq d$, we use $V_k\in \R^{d\times k}$ to denote $V$ restricted to its first $k$ columns of $V$. The matrix $A_k = U_k(A) \Sigma_k(A) V_k(A)$ is the best rank $k$ approximation for $A$.

\subsection{Main Proof}

Let $C\in \R^{n\times d}$ be the matrix whose rows are the $n$ centers. Let $P \in \R^{1 \times d}$ be the row vector of the query. 
For some large enough $w>0$, we define $B\in \R^{(wn+1) \times d}$ to be
\[
B = \begin{bmatrix}
C\\
C\\
\vdots\\
C\\
P
\end{bmatrix}
% , 
% D = \begin{bmatrix}
% C\\
% C\\
% \vdots\\
% C
% \end{bmatrix}
\]
$B$ can be viewed as a point set of size $(wn+1)$ in $\R^d$ and we use $B$ as the instance of feature selection for $k$-means clustering, where $k = n$.

The following claim shows that $V_n(C) \in \R^{d \times n}$ approximates $V_n(B) \in \R^{d \times n}$.

\begin{claim}
    $\|B - B V_n(C) V_n(C)^\transpose \|_F^2 \leq (1+\epsilon) \| B - B_n\|_F^2$ for sufficiently large $w$.
\end{claim}

\begin{proof}
    $\|B - B V_n(C)V_n(C)^\transpose \|_F^2$ does not increase with $w$, because
    \begin{align*}
        \|B - B V_n(C)V_n(C)^\transpose \|_F^2 &= w \|C - C V_n(C)V_n(C)^\transpose \|_F^2 + \|P - P V_n(C)V_n(C)^\transpose \|_F^2 \\
        &= 0 + \|P - P V_n(C)V_n(C)^\transpose \|_F^2
    \end{align*}

    However, $\| B - B_n\|_F^2$ changes with $w$.
    Define a function $h(w)$:
    \begin{align*}
        h(w) := \|B - B_n \|_F^2 = \sigma_{n+1}^2(B) 
    \end{align*}

    % Since $B$ can be viewed as adding a row to $D$, thus $\sigma_{k}^2(D)\geq \sigma_{k+1}^2(B) \geq \sigma_{k+1}^2(D) = 0$.

    Denote $u = V(C)^\transpose P^\transpose$.
    \begin{align*}
        h(w) &= \lambda_{n+1} (B^\transpose B) = \lambda_{n+1} (wC^\transpose C +P^\transpose P )\\
        &= \lambda_{n+1}(w\Sigma^2(C) + V(C)^\transpose P^\transpose P V(C)) = \lambda_{n+1}(w\Sigma^2(C) + uu^\transpose) 
    \end{align*}
    Denote the $d_i = \sigma_i^2(C)$, thus $w\Sigma^2(C) = diag(wd_1, \dots, wd_n, 0, \dots, 0)\in \R^{d\times d}$.

    To calculate the eigenvalues of $B^\transpose B$, we let $\lambda I - B^\transpose B = 0$. After some manipulation\cite{DBLP:books/ox/07/Golub07e}, we know the eigenvalues $\lambda$ of $B^\transpose B$ corresponds to the solutions of 
    \[
    \sum_{i=1}^n \frac{u_i^2}{wd_i - \lambda} = -1+\sum_{i=n+1}^d \frac{u_i^2}{\lambda}
    \]

    And $h(w) = \lambda_{n+1} (B^\transpose B)$ is the smallest non-zero solution, since it is the smallest non-zero eigenvalue. Consider both sides as functions of $\lambda \in [0, wd_n]$. The LHS increases from a positive number to $\infty$, while the RHS decreases from $\infty$ to $0$. Thus $h(w) \in [0, wd_n]$. If we increase $w$, the RHS remains the same, and the LHS can be viewed as shifted rightwards. This means that $h(w)$ increases if $w$ increases, and $\lim_{w\rightarrow \infty} h(w) = \sum_{i=n+1}^d u_i^2$, which is 

    \begin{align*}
        \lim_{w\rightarrow \infty} h(w) &= \sum_{i=n+1}^d u_i^2 \\
        &=\sum_{i=n+1}^d (pv_{i})^2, \text{ where } v_i \text{ i-th column of V}\\
        &= \|P - PV_n(C)V_n(C)^\transpose\|_F^2.
    \end{align*}

    Thus for any $0< \epsilon < 1$, there exists sufficiently large $w$, such that 
    $\|P - PV_n(C)V_n(C)^\transpose\|_F^2 \allowbreak\leq (1+\epsilon) h(w)$, i.e. 
    $\|B - B V_n(C) V_n(C)^\transpose \|_F^2 \leq (1+\epsilon) \| B - B_n\|_F^2$.

    \end{proof}

    \begin{definition}(Approximate partition for $k$-means clustering) Consider a $k$-means clustering problem with $n$ points $C = \{c_1, \dots, c_n\} \subseteq \R^d$. A partition $\rho: [n] \rightarrow [k]$, describes $c_i$ should be in $\rho(i)$-th cluster. The $j$-th cluster $C_j$ of $\rho$ is ${c_i | i \in \rho^{-1}(j)}$ and the $j$-th center of $\rho$ is $center_j = \sum_{c_i \in C_j} c_i / |C_j|$. The cost of $(C, \rho)$ is $cost(C, \rho) = \sum_{i\in [n]} \|c_i - center_{\rho(i)} \|_2^2$.     
    Let $\rho^*$ be the optimal partition, which minimizes the k-means cost. We say $\rho$ is a $\alpha$-approximation partition if $cost(C, \rho) \leq \alpha cost(C, \rho^*)$.        
    \end{definition}

    \begin{theorem}(Theorem 11 and Lemma 4 of \cite{DBLP:journals/tit/BoutsidisZMD15})\label{theorem:bzmd15}
        Let $B\in \R^{m\times d}$ be a set of $m$ points in $\R^d$, and $k$ be the input of k-means clustering. Let $Z\in \R^{d\times k}$ be an approximation of $V_k(B)$, such that $Z^\transpose Z = I_{k\times k}, (B - BZZ^\transpose)Z = \mathbf{0}_{m\times k}$, and $\|B - BZZ^\transpose\|_F^2 \leq (1+\epsilon) \|B - B_k\|_F^2$. Then let $[\Omega, S] = SelectFeature(Z, t)$, with $t = O(k\log k/\epsilon^2)$. Let $B' = B\Omega S \in \R^{m\times t}$ be the points after feature selection. Then a $\gamma$-approxmation partition $\rho$ for $B'$ is $(1+(2+\epsilon)\gamma)$-approximation parition for $B$ with probability $0.2$.
    \end{theorem}

    \begin{theorem}(Theorem 1.3 of \cite{makarychev-etal19:jl-k-means}) \label{theorem:mmr19}Let $B' \in \R^{m\times t}$ be a set of $m$ points in $\R^t$ and $k$ be the input of $k$-means clustering. Given $0<\epsilon, \delta < 1/4$. Let $r = \log(k/\epsilon\delta)/\epsilon^2$. There exists a family of random maps $\pi:\R^t \rightarrow \R^r$, which is independent of $B'$, such that 
    $\prob{M \sim \pi}{\forall \rho, cost(B', \rho) \approx_{\epsilon} cost(M(B'), \rho) } \geq 1-\delta$. 
    \end{theorem}

    \begin{proof}(for \cref{theorem:l2-ann})
        Recall that $B$ can be viewed as a point set of size $(wn+1)$ in $\R^d$. Let $B' = B\Omega S$ as defined in \cref{theorem:bzmd15}, which can be viewed as a point set of size $(wn+1)$ in $\R^t$. Let $M$ be the random dimension reduction map in \cref{theorem:mmr19}. Then $M(B')$ can be viewed as a point set of size $(wn+1)$ in $\R^r$.
        
        Let $\rho_1$ be the optimal partition of $k$-means clustering on $B'$ with $k = n$. Let $\rho_2$ be the optimal partition of $k$-means clustering on $M(B')$ with $k = n$.

        From \cref{theorem:mmr19}, we know that with probability $1-\delta$, 
        \begin{align*}
            \cost(B', \rho_1) \approx_{\epsilon} \cost(M(B'), \rho_1) \\
            \cost(B', \rho_2) \approx_{\epsilon} \cost(M(B'), \rho_2)
        \end{align*}
        Because $\rho_1, \rho_2$ are respectively optimal solutions, we have
        \begin{align*}
            \cost(B', \rho_1) \leq \cost(B', \rho_2) \\
            \cost(M(B'), \rho_2) \leq \cost(M(B'), \rho_1)
        \end{align*}

        Combine them gives $\cost(B', \rho_2) \leq (1+\epsilon) \cost(M(B'), \rho_2) \leq (1+\epsilon) (1+\epsilon) \cost(M(B'), \rho_1) \leq (1+\epsilon)^2 \cost(B', \rho_2)$. This shows that $\rho_2$ is $(1+\epsilon)^2$-approximation for $B'$ with probability $1-\delta$. Furthermore, \cref{theorem:bzmd15} tells us that $\rho_2$ is $(3+10\epsilon)$-approximation for $B$ with probability $0.2 - \delta$, i.e. 
        $$\cost(B, \rho_2) \leq (3+10\epsilon) \cost(B, \rho^*).$$

        We claim that $\cost(B, \rho^*) = (\frac{w}{w+1} \dist(q, c^*))^2 + w(\frac{1}{w+1} \dist(q, c^*))^2$, where $c^*$ is the nearest neighbor of $q$ from $C$. Because $B$ is $w$ copies of $n$ centers $c_i$ plus one $q$, for sufficiently large $w$, $\rho^*$ will not map two centers into one cluster. So $q$ will be assigned to a cluster together with $w$ copies of $c^*$.

        Similarly we claim that $\cost(B, \rho_2) = (\frac{w}{w+1} \dist(q, c_{\hat{i}}))^2 + w(\frac{1}{w+1} \dist(q, c_{\hat{i}}))^2$, where $\hat{i}$ is the output index. This is because $\rho_2$ is the optimal partition of $M(B')$, so $M(y)$ will be assigned to a cluster together with $w$ copies of its nearest neighbor, which is $r_{\hat{i}}$. 
        
        This means $\dist^2(c_{\hat{i}}, q) \leq (3+10\epsilon) \dist^2(c^*, q)$, i.e. $\dist(c_{\hat{i}}, q) \leq \sqrt{3+10\epsilon} \dist(c^*, q)$.
        
\end{proof}

The reason that we can not use the same trick in \cite{DBLP:conf/stoc/CohenEMMP15} to achieve $(1+\epsilon)$-ANN is, it select features based on $Z = V_n(C)$ and $B-BZZ^\transpose = \begin{bmatrix}
    0\\
    0\\
    \vdots\\
    0\\
    P - PZZ^T
\end{bmatrix}$, which requires some information of the query $q$ in the preprocessing of the nearest neighbor data structure.

\input{lowerbound}

\end{document}

%% file: lowerbound.tex
\section{Lower Bounds} \label{sec:lower-bound}

\subsection{Lower Bound for Approximate Nearest Neighbor}
In this section, we prove a $\Omega(n)$ lower bound for the sketch size and a $\Omega(n)$ lower bound for the number of coordinates queried for any data structure for  the approximate nearest neighbor problem. This shows that our data structures are nearly optimal.

\begin{theorem}\label{thm:approx-NN-lb}
    Consider a set $C$ of $n$ points $c_1, \dots, c_n \in \R^d$ equipped with $\ell_1$ or $\ell_2$ metric. Consider a data structure that preprocesses $C$ to store a sketch. For any query point $q\in \R^d$, the data structure uses the sketch, queries some coordinates of $q$, and then returns $i$, where $c_i$ is a $(1+\epsilon)$-approximate nearest neighbor of $q$ in $C$ with probability $2/3$. Then the size of the sketch is $\Omega(n)$ and the number of coordinates queried is $\Omega(n)$.
\end{theorem}
\begin{proof}
    The lower bound on the sketch size follows directly from Theorem 2 of \cite{AndoniIP06}. 
    \begin{theorem}[Theorem 2 of \cite{AndoniIP06}] \label{theorem:CC-ANN}
        Consider the communication complexity version of $(1+\epsilon)$-approximate nearest neighbor in $\{0,1\}^d$, $d = O(\frac{\log^2 n}{\epsilon^5})$, where Alice receives the query $q\in \{0,1\}^d$ and Bob receives the set $C\subset \{0,1\}^d$, then for any $\epsilon = \Omega(n^{-\gamma}), \gamma < 1/2$, in any randomized protocol deciding the $(1+\epsilon)$-approximate nearest neighbor problem, either Alice sends $\Omega(\frac{\log n}{\epsilon^2})$ bits or Bob sends $\Omega(n^{1-c})$ bits, for any $c>0$.
    \end{theorem}
    Assume for contradiction that we have a data structure of size $n^{0.99}$. Then Bob can send this sketch to Alice and let Alice decide the $(1+\epsilon)$-approximate nearest neighbor problem, which contradicts \cref{theorem:CC-ANN}.

    Next, we show a lower bound $\Omega(n)$ of the number of coordinate queries. Consider a set of points $\{c_i\} \subset \R^n$, where $c_i = e_i$ is the unit vector. Let the query to be co-located with one of the points, so that the approximate nearest neighbor has to be the exact nearest neighbor. To find out which coordinate of the query is non-zero, a deterministic data structure needs to look at all $n$ coordinates. One can make the same claim about randomized data structures by applying the Yao's Principle.
\end{proof}

\subsection{Lower Bound for Exact Nearest Neighbor}
We prove a simple and optimal lower bound on the size of the data structure to store $n$ centers if we want to compute nearest neighbor exactly. This lower bound follows from a two-party communication protocol where we assume that one of the parties, Alice, holds all centers and the other party, Bob, holds the query point. Alice is allowed to send a one-shot message to Bob, and upon receiving the message, Bob has to compute the nearest neighbor center of the query point. It is not hard to see that a data-structure and query protocol with data-structure of size $D$ will imply a communication protocol with communication cost of $D$.

\begin{lemma} \label{lem:exact-lb}
    If Alice sends a sketch of her centres from which Bob can exactly compute the nearest neighbour of $X$ among $C$, then Alice needs to send at least $\Omega(nd)$ bits.
\end{lemma}

\begin{proof}
    The proof is a slight modification of Theorem B.1 of \cite{indyk-wagner18:apx-nn-limited-space} (with a change in notation which we explain below). First, note that, if Bob can compute $\cost(C,X)$ exactly, then it must be the case that the map $\sigma$ that maps each point $x \in X$ to cluster numbered $\sigma(x)$ correctly computes the nearest cluster for each $x$. Indeed, if this is not the case, some $x \in X$ is mapped to a cluster further away from its nearest one. In this case, the additive error incurred cannot be negated by other points in $X$ as none of those points cannot be mapped to a cluster which is closer than its nearest one.

    We assume $d=d'+1+\log(n)$. Setting $\epsilon = \sqrt{2/d'}$ and fixing $\ell = 1/\epsilon^2$ in Theorem B.1 of \cite{indyk-wagner18:apx-nn-limited-space}, we have dataset that consists of $2n$ many points, $x_1, \cdots, x_n$ and $z_1, \cdots, z_n$ each of dimension $d$ such that: \begin{itemize}
        \item The first $d'$ coordinates of $x_i$ is an arbitrary $\ell = d'/2$ sparse vector in which every non-zero coordinate is set to $1/\sqrt{\ell}$. The $(d'+1)$-th coordinate is 0 and the remaining $\log n$ coordinates encode the binary encoding of $i$, with each coordinate multiplied by 10.

        \item The first $d'$ coordinates of $z_i$ is 0. The $(d'+1)$-th coordinate is $\sqrt{1-\epsilon}$ and the rest of the coordinates are encoded similar to $x_i$.
    \end{itemize}

    Similar to \cite{indyk-wagner18:apx-nn-limited-space}, we can see that the number of bits required to represent $x_1, \cdots, x_n$ is at least $\log (\binom{d'}{d'/2}^n) = \Omega(dn)$ which is our desired lower bound. Also, we note that for a query point $y_{ij}$ whose first $d'+1$ coordinates are 0 except the $j$-th coordinate which is set to 1, and the rest of the coordinates are similar to $x_i$ or $z_i$, the following properties hold: \begin{itemize}
        \item If $x_i(j) \neq 0$, then the closest point of $y_{ij}$ is $x_i$ at a distance $\sqrt{2-2\epsilon}$.
        \item if $x_i(j) = 0$, then the closest point of $y_{ij}$ is $z_i$ at a distance $\sqrt{2-\epsilon}$.
    \end{itemize}
    An exact nearest neighbor query, which can be obtained from $\sigma$ will reveal the value of $x_i(j)$ for any $i \in [k]$ and $j\in[d']$. This proves the $\Omega(dn)$ lower bound.
\end{proof}